\documentclass[a4paper,12pt]{article}

\usepackage{authblk}
\usepackage{amssymb}
\usepackage{amsmath}
\usepackage{amstext}
\usepackage{graphicx}
\usepackage[colorlinks, citecolor=blue, linkcolor=blue,
urlcolor=blue]{hyperref}

\usepackage{amsfonts}\usepackage{mathrsfs}
\usepackage[normalem]{ulem}
\usepackage{mathptmx}
\usepackage{bm}
\usepackage{fullpage}

\DeclareSymbolFont{cmsy}{OMS}{cmsy}{m}{n}
\DeclareSymbolFontAlphabet{\mathcalus}{cmsy}

\usepackage{amsthm}
\newtheorem{theorem}{Theorem}[]
\newtheorem{proposition}[theorem]{Proposition}
\newtheorem{lemma}[theorem]{Lemma}
\newtheorem{corollary}[theorem]{Corollary}

\newtheorem{remark}[theorem]{Remark}

\newcommand{\transp}[0]{\mathrm{T}}
\newcommand{\rel}[0]{\mathrm{rel}}
\newcommand{\secur}[0]{\mathrm{sec}}

\newcommand{\tr}{\operatorname{tr}}
\newcommand{\ket}[1]{\lvert{#1} \rangle}
\newcommand{\bra}[1]{{\langle {#1}\rvert}}
\newcommand{\braket}[2]{\langle {#1} | {#2} \rangle}
\allowdisplaybreaks
\newcommand{\bi}{\mathbf{i}}
\newcommand{\bj}{\mathbf{j}}
\newcommand{\bl}{{\bm{\ell}}}
\newcommand{\bc}{\mathbf{c}}
\newcommand{\bC}{\mathbf{C}}
\newcommand{\bk}{\mathbf{k}}
\newcommand{\bbm}{\mathbf{m}}

\newcommand{\bb}{\mathbf{b}}
\newcommand{\bT}{\mathbf{T}}

\newcommand{\ba}{\mathbf{a}}
\newcommand{\bap}{\mathbf{a'}}
\newcommand{\be}{\mathbf{e}}
\newcommand{\bep}{\mathbf{e'}}
\newcommand{\bx}{\mathbf{x}}
\newcommand{\by}{\mathbf{y}}
\newcommand{\bz}{\mathbf{z}}
\newcommand{\bu}{\mathbf{u}}
\newcommand{\bv}{\mathbf{v}}

\newcommand{\sA}{\mathrm{A}}
\newcommand{\sB}{\mathrm{B}}
\newcommand{\sC}{\mathrm{C}}
\newcommand{\sE}{\mathrm{E}}
\newcommand{\sM}{\mathrm{M}}
\newcommand{\sU}{\mathrm{U}}
\newcommand{\sK}{\mathrm{K}}
\newcommand{\sH}{\mathrm{H}}
\newcommand{\sI}{\mathrm{I}}
\newcommand{\sT}{\mathrm{T}}
\newcommand{\sX}{\mathrm{X}}
\newcommand{\sZ}{\mathrm{Z}}

\newcommand{\bxi}{{\boldsymbol{\xi}}}
\newcommand{\bs}{\mathbf{s}}
\newcommand{\Eve}{{\operatorname{Eve}}}
\newcommand{\Span}{{\operatorname{Span}}}
\newcommand{\except}{{\operatorname{exc}}}
\newcommand{\pr}{{\operatorname{Pr}}}
\newcommand{\invb}{{\operatorname{inverted-INFO-basis}}}
\newcommand{\sym}{{\operatorname{sym}}}
\newcommand{\ext}{{\operatorname{ext}}}

\begin{document}

\title{Composable Security of Generalized BB84 Protocols
Against General Attacks}
\author[1]{Michel Boyer}
\author[1,2,3]{Rotem Liss}
\author[2]{Tal Mor}
\affil[1]{{\small D\'epartement IRO, Universit\'e de Montr\'eal,
Montr\'eal (Qu\'ebec), H3C 3J7, Canada}}
\affil[2]{{\small Computer Science Department, Technion, Haifa, 3200003, Israel}}
\affil[3]{{\small ICFO---Institut de Ciencies Fotoniques,
The Barcelona Institute of Science and Technology, Av.\ Carl Friedrich Gauss, 3,
Castelldefels, Barcelona, 08860, Spain}}
\date{}

\maketitle

\begin{abstract}
Quantum key distribution (QKD) protocols make it possible for
two parties to generate a secret shared key.
One of the most important QKD protocols, BB84,
was suggested by Bennett and Brassard in 1984.
Various proofs of unconditional security for BB84 have been suggested,
but the first security proofs were not composable.
Here we improve a security proof of BB84 given by
Biham, Boyer, Boykin, Mor, and Roychowdhury~\cite{BBBMR06} to be composable
and match the state-of-the-art results for BB84,
and we extend it to prove unconditional security of
several variants of the BB84 protocol.
Our composable security proof for BB84 and its variants
is mostly self-contained, algebraic, and relatively simple,
and it gives tight finite-key bounds.
\end{abstract}

\section{Introduction}
Quantum key distribution (QKD) makes it possible for two legitimate
parties, Alice and Bob, to generate an information-theoretically
secure key~\cite{bb84}, that is secure against any possible attack
allowed by the laws of quantum physics.
Alice and Bob use an insecure quantum channel and an authenticated
classical channel. The adversary Eve may interfere with the quantum
channel and is limited only by the laws of nature; however, she cannot
modify the data sent in the authenticated classical channel
(she can only listen to it).

The main objective of analyzing a QKD protocol is proving its
{\em unconditional security}: namely, proving that even if
the adversary Eve applies the strongest and most general attacks
allowed by the laws of nature (the ``joint attacks''),
Eve's average information about the final key
is still negligible---that is, exponentially small
in the number of qubits.

The first and most important QKD protocol was BB84~\cite{bb84};
full definition of the BB84 protocol is available
in Section~\ref{generalized-bb84_description}.
BB84 has many proofs of unconditional security (see, among others,
\cite{mayers01,BBBMR06,SP00,RGK05,TLGR12,HT12,TL17}),
each of them having its own advantages and disadvantages.
Extending the toolkit available for security proofs
is still important, because many practical implementations of
theoretical protocols and many specifically designed practical
protocols do not yet have their full and unconditional security proved,
and because existence of many proofs makes the security result
more certain and less prone to errors.

In our work, we extend and generalize the security proof of BB84
presented by Biham, Boyer, Boykin, Mor, and Roychowdhury
(BBBMR)~\cite{BBBMR06}:
first, we change it to prove the {\em composable}
security of BB84---namely, to prove that the secret key remains secure
even if Alice and Bob use it for cryptographic applications
(see details in Subsection~\ref{subsec_secure});
second, we generalize it to apply
to several variants of BB84 and not only to the original protocol;
and third, we improve it to match the state-of-the-art results
and security bounds for BB84.

Our paper gives composable security and tight finite-key bounds
in a relatively simple way. This paper and the BBBMR
paper~\cite{BBBMR06} together are mostly self-contained.
On the other hand, our paper leaves further possible generalizations
for future research.

\subsection{\label{subsec_secure}Unconditional Security Definitions
of QKD Protocols}
Originally, a QKD protocol was defined to be secure if the
(classical) {\em mutual information} between Eve's information
and the final key, maximized over all the possible attack strategies
and measurements by Eve, was exponentially small in the number of
qubits, $N$. Examples of security proofs of BB84 that use this
security definition
are~\cite{mayers01,BBBMR06,SP00}:
these security proofs used the observation that one cannot analyze
the {\em classical} data held by Eve before privacy amplification,
but must analyze the {\em quantum} state held by Eve~\cite{BMS96}.
In other words, they assumed that Eve
could keep her quantum state until the end of the protocol,
and only {\em then} choose the optimal measurement
(based on all the data she observed) and perform the measurement.

Later, it was noticed that this security definition may not be
``composable''. In other words, the final key is secure if Eve
measures the quantum state she holds at the end of the QKD protocol,
but the proof does not apply to {\em cryptographic applications}
(e.g., encryption) of the final key:
Eve might gain non-negligible information after the key is used,
even though her information on the key itself was negligible.
This means that the proof is not sufficient for practical purposes.
In particular, these applications may be insecure if Eve keeps her
quantum state until Alice and Bob use the key (thus giving Eve some
new information) and only {\em then} measures.

Therefore, a new notion of ``(composable) full security''
was defined~\cite{BOHLMO05,RGK05,renner_thesis08}
using the trace distance,
following universal composability definitions
for non-quantum cryptography~\cite{compos01_universal}.
Intuitively, this notion requires that the final joint quantum state of
Alice, Bob, and Eve at the end of the protocol is {\em very close}
to their final state at the end of an {\em ideal} key distribution
protocol, that distributes a {\em completely random} and
{\em secret} final key to both Alice and Bob
(namely, the trace distance between the real state and the ideal state
is exponentially small in $N$).
In other words, if a QKD protocol is secure, then
except with an exponentially small probability,
one of the two following events happens:
the protocol is aborted, {\em or} the secret key generated
by the protocol is the same as a perfect key
that is uniformly distributed
(i.e., each possible key having the same probability),
is the same for both parties,
and is independent of the adversary's information.
Details about the formal definition of composable security
are available in Subsection~\ref{sec_full_composability}.

\subsection{Security Proofs of QKD}
Composable security proofs for many QKD protocols, including BB84, have
been presented~\cite{BOHLMO05,RGK05,renner_thesis08}.
The proofs of~\cite{RGK05,renner_thesis08} used different
methods from the earlier (non-composable) results
of~\cite{mayers01,BBBMR06,SP00}.

The security proof of BB84 presented by Biham, Boyer, Boykin, Mor,
and Roychowdhury (BBBMR)~\cite{BBBMR06}
(which follows previous works by some of its authors)
used a connection between the information obtained by Eve and
the disturbance she induces in the opposite (conjugate) basis.
The proof bounded the trace distance between the density matrices
held by Eve by using algebraical analysis, and then it used this bound
for calculating the mutual information between Eve and the final key.
BBBMR proved non-composable security of BB84 against
the most general theoretical attacks.

The security proof of BBBMR has various advantages and disadvantages
compared to other security proofs for QKD.
On the one hand, the security proof of BBBMR is mostly self-contained,
while other security approaches require many results from other areas
of quantum information (such as various notions of entropy needed for
the security proof of~\cite{RGK05,renner_thesis08},
and entanglement purification and quantum error correction needed for
the security proof of~\cite{SP00});
it gives tight finite-key bounds,
unlike several other proof methods~\cite{SP00}
and similarly to the most updated results (for the BB84 protocol)
in some proof methods~\cite{TLGR12,HT12,TL17};
after our improvements, it matches
state-of-the-art results (most notably, we proved its asymptotic error rate
threshold to be 11\%, as given by~\cite{SP00,RGK05},
instead of the 7.56\% threshold found by~\cite{mayers01,BBBMR06});
and, at least in some sense, it is simpler than other proof techniques.
On the other hand, it is currently limited to BB84-like protocols.

\subsection{Our Contribution}
In our work, we extend the security proof of BBBMR
in the three following ways:

First, we prove {\em composable} security
(instead of the non-composable security
proved in~\cite{BBBMR06}) against the most general attacks.
Intuitively, this improvement is achieved by replacing the
manipulations of classical mutual information in~\cite{BBBMR06}
by similar manipulations on the quantum trace distance.

Second, we prove security of
four QKD protocols similar to BB84 (``generalized BB84''):
\begin{enumerate}
\setlength\itemsep{0em}
\item the ``BB84-INFO-$z$'' protocol (defined and analyzed
by~\cite{bb84-info-z17,bb84-info-z20} against collective attacks),
in which the final key includes
only bits sent in the $z$ basis $\{\ket{0}_0, \ket{1}_0\}$;
\item the standard BB84 protocol;
\item the ``efficient BB84'' protocol described in~\cite{LCA05},
in which Alice and Bob choose the $z$ basis
$\{\ket{0}_0, \ket{1}_0\}$ or the $x$ basis
$\{\ket{0}_1 \triangleq \frac{\ket{0}_0 + \ket{1}_0}{\sqrt{2}},
\ket{1}_1 \triangleq \frac{\ket{0}_0 - \ket{1}_0}{\sqrt{2}}\}$
with non-uniform probabilities; and
\item the ``modified efficient BB84'' protocol
(a simplified version of ``efficient BB84'').
\end{enumerate}
Intuitively, this improvement is achieved by
using the correct statistical argument (namely, the correct application
of Hoeffding's theorem) for proving security of each protocol.
Full details are given in Section~\ref{sec_protocols}.

Third, we improve the security results of BBBMR to match
the state-of-the-art results achieved for BB84.
Most importantly, we improve the error rate threshold
in the ``asymptotic'' scenario, where an infinite number of qubits is transmitted
(and the corresponding error rate thresholds in the finite-key scenario,
where only a finite number of qubits is transmitted,
is decreased accordingly, depending on the specific number of qubits).
In the very first security proofs of BB84~\cite{mayers01,BBBMR06},
the asymptotic error rate threshold was found to be 7.56\%
(namely, above this error rate, the protocol must have been aborted);
however, later security proofs such as~\cite{SP00,RGK05} succeeded to improve
this error rate to 11\%. While the original BBBMR proof only supported
the original error rate of 7.56\%, here we present the algebraic improvement
required to make it support an error rate of 11\%,
as detailed in Subsection~\ref{subsec_ivd}.
(Intuitively, it is based on properties of error-correcting codes detailed
in Subsection~\ref{prelims_ecc}.)
This improvement assures that BBBMR methods give security results compatible
with other proof methods.

\section{Preliminaries}
\subsection{The Notations of Quantum Information}
In quantum information, information is represented by quantum states.
A quantum pure state is denoted by $\ket{\psi}$,
and it is a normalized vector in a Hilbert space.
The qubit Hilbert space is the $2$-dimensional Hilbert space
$\mathcal{H}_2 \triangleq \Span\{\ket{0}_0, \ket{1}_0\}$,
where the two states $\ket{0}_0, \ket{1}_0$
(denoted $\ket{0}, \ket{1}$ in most papers)
form an orthonormal
basis of $\mathcal{H}_2$ named ``the computational basis'' or
``the $z$ basis''. The two states
$\ket{0}_1 \triangleq \frac{\ket{0}_0 + \ket{1}_0}{\sqrt{2}}$ and
$\ket{1}_1 \triangleq \frac{\ket{0}_0 - \ket{1}_0}{\sqrt{2}}$
(denoted $\ket{\text{+}}, \ket{-}$ in most papers)
form another orthonormal basis of $\mathcal{H}_2$,
named ``the Hadamard basis'' or ``the $x$ basis''.
These two bases are said
to be {\em conjugate bases}.

A quantum mixed state is a probability distribution of
several pure states, and it is represented by a density matrix:
$\rho = \sum_j q_j \ket{\psi_j} \bra{\psi_j}$, where $q_j$ is
the probability that the system is in the pure state $\ket{\psi_j}$
(this definition should not be confused
with the probabilities of measurement results).
For example, if the mixed state of a system is
$\rho_1 = \frac{1}{3} \ket{0}_0 \bra{0}_0 +
\frac{2}{3} \ket{0}_1 \bra{0}_1$,
this means that the system is in the $\ket{0}_0$ state
with probability $\frac{1}{3}$
and in the $\ket{0}_1$ state with probability $\frac{2}{3}$.

The most general operations allowed by quantum physics for
the Hilbert space $\mathcal{H}$ are: performing any
unitary transformation $U:\mathcal{H} \rightarrow \mathcal{H}$;
adding an ancillary state inside another Hilbert space;
measuring a state with respect to some orthonormal basis;
and tracing out a quantum system (namely, ignoring and forgetting
a quantum system).

See~\cite{NCBook} for more background about quantum information.

\subsection{The Notations of Bit Strings}
In this paper, we denote bitstrings (of $t$ bits,
where $t \ge 0$ is some integer) by a bold letter
(e.g., $\bi=i_1\ldots i_t$ with $i_1,\ldots,i_t \in \{0,1\}$);
and we refer to these bitstrings as elements
of $\mathbf{F}_2^t$---namely,
{\em row vectors} in a $t$-dimensional vector space over
the field $\mathbf{F}_2 = \{0,1\}$, where addition of two vectors
corresponds to a XOR operation between them.
The number of $1$-bits in a bitstring $\bs$ is denoted by $|\bs|$,
and the Hamming distance between two strings $\bs$ and $\bs'$
is $d_\sH(\bs, \bs') \triangleq |\bs \oplus \bs'|$.

\subsection{\label{prelims_ecc}Linear Error-Correcting Codes}
\subsubsection{\label{subsubsec_prelims_ecc}Basic Notions}
Assume a sender and a receiver want to communicate
over a noisy channel. They can use an {\em error-correcting code}:
the sender adds redundancy to the input of the channel
(``encoding''), and the receiver can correct the errors induced
by the noisy channel (``decoding''). In QKD, binary linear
error-correcting codes can be used for correcting the errors
caused during the quantum transmission.

An $[n,k]$ binary linear error-correcting code $\sC$
includes $2^k$ bitstrings named ``codewords'',
where each codeword is of length $n$.
Formally, $\sC$ is a $k$-dimensional vector subspace of $\mathbf{F}_2^n$
(over the field $\mathbf{F}_2$).
Since $\sC$ is linear, for all $\ba, \bb \in \sC$
it holds that $\ba \oplus \bb \in \sC$.
In the standard use of the error-correcting code $\sC$,
the $k$-bit input string $\bu$ is translated into
an $n$-bit codeword $\ba \in \sC$, which is then
transmitted through the noisy channel.

The minimum distance $d$ of a binary linear error-correcting code $\sC$
is the smallest Hamming distance between any two codewords:
namely, $d \triangleq \min_{\ba \ne \bb \in \sC} d_\sH(\ba, \bb)$.
Equivalently, it is the smallest Hamming weight of a non-zero codeword:
$d = \min_{\ba \in \sC \setminus \{\mathbf{0}\}} |\ba|$.

The $k \times n$ {\em generator matrix} $G_\sC$
defines the code $\sC$, as follows:
\begin{equation}\label{ecc_generator}
\sC = \{\bu G_\sC \mid \bu \in \mathbf{F}_2^k\}.
\end{equation}
Namely, the code $\sC$ is the {\em row space} of $G_\sC$.

The $r \times n$ {\em parity-check matrix} $P_\sC$
(where $r \triangleq n - k$) similarly
defines the code $\sC$, as follows:
for any $\ba \in \mathbf{F}_2^n$,
\begin{equation}\label{ecc_syndrome}
\ba \in \sC ~ \Leftrightarrow ~ \ba P_\sC^\transp = \mathbf{0}.
\end{equation}
Namely, the code $\sC$ is the {\em kernel} of $P_\sC$;
we note that $P_\sC$ must be of rank $r$.
The parity-check matrix $P_\sC$ is also a generator matrix
of the {\em dual code} $\sC^\perp$, defined as:
\begin{equation}\label{ecc_dual}
\sC^\perp = \{\bv P_\sC \mid \bv \in \mathbf{F}_2^r\},
\end{equation}
which means, in particular, that the inner product $\ba \cdot \bb = 0$
for any $\ba \in \sC, \bb \in \sC^\perp$.
Therefore, the dual code $\sC^\perp$ is the
row space of $P_\sC$ (and it can be shown to be the kernel of $G_\sC$).

Standard error correction for {\em classical} channels
can be performed as follows:
the sender sends a codeword $\ba \in \sC$ through the noisy
channel, and the receiver gets a noisy result $\bb = \ba \oplus \be$
(where $\be \in \mathbf{F}_2^n$ is the {\em error word}).
A standard decoder, which we use throughout this paper,
is the {\em nearest-codeword decoder}: the word $\bb$
is decoded into the codeword minimally distanced from it.
The receiver knows the sender sent a codeword in $\sC$;
therefore, the receiver can check all possible codewords $\bap
\in \sC$, compute the error words $\bep = \bb \oplus \bap$ that
would change $\bap$ into $\bb$,
and choose the lowest-weight error word $\bep$
(which thus corresponds to the nearest codeword, $\bap$).
If the minimum distance $d$ equals $2t + 1$,
this decoding procedure can correct up to $t$ errors with certainty;
in Subsubsection~\ref{subsubsec_ecc} we show
how to improve this performance.

Error correction for standard QKD protocols
is a little different: the input word $\bi_\sI$ (sent by Alice
as a quantum state) is random and {\em is not necessarily a codeword}.
Let us assume that Bob gets the output word $\bj_\sI \triangleq
\bi_\sI \oplus \bc_\sI$ (where $\bc_\sI$ is the error word).
In this case, errors can still be corrected in the following way:
the sender Alice sends to Bob the {\em syndrome} $\bxi \in \mathbf{F}_2^r$,
defined as
\begin{equation}
\bxi \triangleq \bi_\sI P_\sC^\transp.
\end{equation}
Then, let us define the {\em code coset} $\sC_\bxi$ as follows:
\begin{equation}\label{ecc_c_xi_def}
\sC_\bxi \triangleq \{\bz \in \mathbf{F}_2^n \mid \bz P_\sC^\transp = \bxi\}.
\end{equation}
The set $\sC_\bxi$ is equal to $\bl_\bxi + \sC$, given any arbitrary word
$\bl_\bxi \in \mathbf{F}_2^n$ having the same syndrome $\bxi$
(namely, satisfying $\bl_\bxi P_\sC^\transp = \bxi$),
which is why we name it ``code coset'': this is true
because Equation~\eqref{ecc_syndrome} implies $\sC_{\mathbf{0}} = \sC$.
Therefore, since Bob knows $\bxi$, he also knows that the input word $\bi_\sI$
must be in the code coset $\sC_\bxi = \bl_\bxi + \sC$
(namely, there must exist $\ba \in \sC$
such that $\bi_\sI = \bl_\bxi \oplus \ba$). Therefore, Bob can apply
a similar method to standard error correction
on the code coset $\sC_\bxi$
(instead of applying it on the original code $\sC$)---namely,
he can use the affine code $\sC_\bxi$
instead of the linear code $\sC$.
Specifically, Bob can check all possible words in the code coset $\sC_\bxi$,
find the corresponding error words,
and choose the lowest-weight error word
(which thus corresponds to the nearest word in the code coset).
As before, if the code's minimum distance $d$ equals $2t + 1$,
this decoding procedure can correct up to $t$ errors with certainty.

See~\cite{CoverThomasBook,RothBook} for more details
about error-correcting codes.

\subsubsection{\label{subsubsec_ecc}Finding Codes for Error Correction
and Privacy Amplification}
We now present several important results
on linear error-correcting codes,
that are essential both for fully proving the results discussed
in Appendix~E of~\cite{BBBMR06}
and for our improved results in Subsection~\ref{subsec_ivd}.

For the nearest-codeword decoder used throughout this paper,
we assume that whenever it gets an input $\bx \in \mathbf{F}_2^n$
such that {\em two or more} codewords are nearest to $\bx$
and have equal distances from $\bx$, the decoder fails
(namely, its output may be arbitrary and is ignored).
Namely, we assume the decoder never ``breaks ties''
between candidates to being nearest codewords, but simply fails if there is a tie
and no unique decoding is possible.

First, we prove the success of error correction to be independent
of the original word:
\begin{lemma}\label{lemma_ecc_indep}
For any $[n,k]$ binary linear error-correcting code $\sC$
and any error word $\be \in \mathbf{F}_2^n$:
\begin{enumerate}
\item \label{lemma_ecc_indep_1}A nearest-codeword decoder
successfully corrects the error $\be$
if and only if the codeword nearest to $\be$ is uniquely $\mathbf{0}$
(namely, $\mathbf{0} \in \sC$ is strictly closer to $\be$
than any non-zero codeword in $\sC$).

Formally, if the input to the nearest-codeword decoder is
$\ba \oplus \be$, where $\ba \in \sC$ is a codeword
and $\be \in \mathbf{F}_2^n$ is the error word,
then $\ba \oplus \be$ is correctly decoded into $\ba$
if and only if the codeword nearest to $\be$ is uniquely $\mathbf{0}$.

\item \label{lemma_ecc_indep_2}In particular,
the success of correcting the error $\be$
using a nearest-codeword decoder is independent of the codeword,
and depends only on the error word.

Formally, for any two codewords $\ba, \bb \in \sC$,
the bitstring $\ba \oplus \be$ is correctly decoded into $\ba$
if and only if the bitstring $\bb \oplus \be$ is correctly decoded
into $\bb$.
\end{enumerate}
\end{lemma}
\begin{proof}
As explained in Subsubsection~\ref{subsubsec_prelims_ecc}, for $\ba \in \sC$,
the bitstring $\ba \oplus \be$ is decoded to its nearest codeword.
Let us denote by $\bx \in \sC$ a codeword nearest to $\be$,
and denote by $\by' \in \sC$ a codeword nearest to $\ba \oplus \be$.
We can easily denote $\by' = \ba \oplus \by$
(by denoting $\by \triangleq \by' \oplus \ba$;
let us remember that $\by', \ba \in \sC$, so $\by \in \sC$).
Then, because $\bx$ is a codeword nearest to $\be$
and $\ba \oplus \by$ is a codeword nearest to $\ba \oplus \be$, we get:
\begin{equation}
d_\sH(\bx, \be) = d_\sH(\ba \oplus \bx, \ba \oplus \be)
\ge d_\sH(\ba \oplus \by, \ba \oplus \be)
= d_\sH(\by, \be) \ge d_\sH(\bx, \be).
\end{equation}
Because the left-hand-side and right-hand-side are equal,
it follows that all inequalities must be equalities,
so $d_\sH(\bx, \be) = d_\sH(\by, \be)$
and $d_\sH(\ba \oplus \bx, \ba \oplus \be)
= d_\sH(\ba \oplus \by, \ba \oplus \be)$.
Therefore, both $\bx$ and $\by$ are codewords nearest to $\be$,
and both $\ba \oplus \bx$ and $\ba \oplus \by$
are codewords nearest to $\ba \oplus \be$.
This particularly means that the codeword nearest to $\be$ is uniquely $\bx$
if and only if the codeword nearest to $\ba \oplus \be$ is uniquely
$\ba \oplus \bx$, in which case $\bx = \by$.

Therefore, the nearest-codeword decoder correctly decodes $\ba \oplus \be$
into $\ba$ if and only if the codeword nearest to $\ba \oplus \be$ is uniquely
$\ba$~\footnote{If the codeword nearest to $\ba \oplus \be$ is {\em not} unique,
then according to our above assumption, the nearest-codeword decoder fails.},
which happens if and only if the codeword nearest to $\be$ is uniquely
$\mathbf{0}$; this proves item~\ref{lemma_ecc_indep_1}.

For deducing item~\ref{lemma_ecc_indep_2}, we apply item~\ref{lemma_ecc_indep_1}
twice, proving that for any two codewords $\ba, \bb \in \sC$,
the bitstring $\ba \oplus \be$ is correctly decoded into $\ba$
if and only if the codeword nearest to $\be$ is uniquely $\mathbf{0}$,
which happens if and only if the bitstring $\bb \oplus \be$ is correctly decoded
into $\bb$.
Therefore, the decoding's correctness is independent of the codeword
($\ba$ or $\bb$), and depends only on the error word $\be$.
\end{proof}

Lemma~\ref{lemma_ecc_indep} can be extended to
QKD's modified error correction process
described in Subsubsection~\ref{subsubsec_prelims_ecc}.
For this, we use a {\em modified nearest-codeword decoder}
(which we call ``nearest-word-in-code-coset decoder''),
which gets an input word $\bx \in \mathbf{F}_2^n$
and an input syndrome $\bxi \in \mathbf{F}_2^r$
and decodes $\bx$ into the nearest word in the code coset
$\sC_\bxi \triangleq \{\bz \in \mathbf{F}_2^n \mid \bz P_\sC^\transp = \bxi\}$
defined in Equation~\eqref{ecc_c_xi_def}.
As before, we assume that if {\em two or more words in the code coset}
are nearest to $\bx$ and have equal distances from $\bx$,
the decoder fails---namely, that the decoder never ``breaks ties''.
The generalization is now presented in the following Lemma:
\begin{lemma}\label{lemma_ecc_indep_coset}
For any $[n,k]$ binary linear error-correcting code $\sC$
and any error word $\be \in \mathbf{F}_2^n$:
\begin{enumerate}
\item \label{lemma_ecc_indep_coset_1}A nearest-word-in-code-coset decoder
successfully corrects the error $\be$
if and only if the codeword (in $\sC$) nearest to $\be$ is uniquely $\mathbf{0}$
(namely, $\mathbf{0} \in \sC$ is strictly closer to $\be$
than any non-zero codeword in $\sC$).

Formally, if the input word to the nearest-word-in-code-coset decoder
is $\bi \oplus \be$ and the input syndrome is $\bxi$,
where $\bi \in \mathbf{F}_2^n$ is the original word,
$\be \in \mathbf{F}_2^n$ is the error word,
and $\bxi = \bi P_\sC^\transp$ is the syndrome of the original word,
then $\bi \oplus \be$ is correctly decoded into $\bi$
if and only if the codeword (in $\sC$) nearest to $\be$ is uniquely $\mathbf{0}$.

\item \label{lemma_ecc_indep_coset_2}In particular,
the success of correcting the error $\be$
using the nearest-word-in-code-coset decoder is independent of the original word
\emph{or} the input syndrome, and depends only on the error word.

Formally, for any two original words $\bi_1, \bi_2 \in \mathbf{F}_2^n$,
the bitstring $\bi_1 \oplus \be$ is correctly decoded into $\bi_1$
(given the input syndrome $\bxi_1 \triangleq \bi_1 P_\sC^\transp$)
if and only if the bitstring $\bi_2 \oplus \be$ is correctly decoded
into $\bi_2$
(given the input syndrome $\bxi_2 \triangleq \bi_2 P_\sC^\transp$).
\end{enumerate}
\end{lemma}
\begin{proof}
As explained in Subsubsection~\ref{subsubsec_prelims_ecc}
and in the above definition of the nearest-word-in-code-coset decoder,
for the original word $\bi \in \mathbf{F}_2^n$
and the input syndrome $\bxi = \bi P_\sC^\transp$,
the bitstring $\bi \oplus \be$ is decoded
to the nearest word in the code coset
$\sC_\bxi \triangleq \{\bz \in \mathbf{F}_2^n \mid \bz P_\sC^\transp = \bxi\}$
(where also $\bi \in \sC_\bxi$, by definition).
Let us denote by $\bx \in \sC$ a codeword (in $\sC$) nearest to $\be$,
and denote by $\by' \in \sC_\bxi$ a word in the code coset
that is nearest to $\bi \oplus \be$.
We can easily denote $\by' = \bi \oplus \by$
(by denoting $\by \triangleq \by' \oplus \bi$;
let us remember that $\by', \bi \in \sC_\bxi$, so $\by \in \sC$).
Then, because $\bx \in \sC$ is a codeword nearest to $\be$
and $\bi \oplus \by \in \sC_\bxi$ is a word in the code coset
nearest to $\bi \oplus \be$, we get:
\begin{equation}
d_\sH(\bx, \be) = d_\sH(\bi \oplus \bx, \bi \oplus \be)
\ge d_\sH(\bi \oplus \by, \bi \oplus \be)
= d_\sH(\by, \be) \ge d_\sH(\bx, \be).
\end{equation}
Because the left-hand-side and right-hand-side are equal,
it follows that all inequalities must be equalities,
so $d_\sH(\bx, \be) = d_\sH(\by, \be)$
and $d_\sH(\bi \oplus \bx, \bi \oplus \be)
= d_\sH(\bi \oplus \by, \bi \oplus \be)$.
Therefore, both $\bx$ and $\by$ are codewords nearest to $\be$,
and both $\bi \oplus \bx$ and $\bi \oplus \by$
are words in the code coset $\sC_\bxi$ that are nearest to $\bi \oplus \be$.
This particularly means that the codeword nearest to $\be$ is uniquely $\bx$
if and only if the word in $\sC_\bxi$ that is nearest to $\bi \oplus \be$
is uniquely $\bi \oplus \bx$, in which case $\bx = \by$.

Therefore, the nearest-word-in-code-coset decoder
correctly decodes $\bi \oplus \be$ into $\bi$ if and only if
the word in $\sC_\bxi$ nearest to $\bi \oplus \be$ is uniquely
$\bi$~\footnote{If the word in $\sC_\bxi$ that is nearest to $\bi \oplus \be$ is
{\em not} unique, then according to our above assumption,
the nearest-word-in-code-coset decoder fails.},
which happens if and only if the codeword nearest to $\be$ is uniquely
$\mathbf{0}$; this proves item~\ref{lemma_ecc_indep_coset_1}.

For deducing item~\ref{lemma_ecc_indep_coset_2},
we apply item~\ref{lemma_ecc_indep_coset_1} twice,
proving that for any two original words $\bi_1, \bi_2 \in \mathbf{F}_2^n$,
the bitstring $\bi_1 \oplus \be$ is correctly decoded into $\bi_1$
if and only if the codeword nearest to $\be$ is uniquely $\mathbf{0}$,
which happens if and only if the bitstring $\bi_2 \oplus \be$ is correctly decoded
into $\bi_2$.
Therefore, the decoding's correctness is independent of the codeword
($\bi_1$ or $\bi_2$) {\em or} the respective input syndromes $\bxi_1, \bxi_2$,
and depends only on the error word $\be$.
\end{proof}
Lemmas~\ref{lemma_ecc_indep} and~\ref{lemma_ecc_indep_coset}
imply that for any code $\sC$,
the nearest-word-in-code-coset decoder of QKD
gives {\em identical results} as the standard nearest-codeword decoder
on any error string $\be \in \mathbf{F}_2^n$:
namely, if any of the two decoders
operates on the error string $\be$,
it will be successful in decoding if and only if
the codeword nearest to $\be$ is uniquely $\mathbf{0}$.
Therefore, both decoders work correctly on all error words
$\be \in \mathbf{F}_2^n$
whose nearest codeword is uniquely $\mathbf{0} \in \sC$,
and both fail on all other error words.
(The probability that this decoding fails for a random code is bounded
by Corollary~\ref{corollary_ecc_decoding}.)

Second, given a specific word $\bl \in \mathbf{F}_2^n$,
we bound the probability that for a randomly chosen code $\sC$
there is another low-weight word $\bz$
such that the difference $\bl \oplus \bz$ is inside $\sC$:
\begin{proposition}\label{prop_ecc_random}
Let $n, k, t \ge 1$ be integers
such that $0 \le \frac{t}{n} \le \frac{1}{2}$.
For any word $\bl \in \mathbf{F}_2^n$,
the probability that a {\em randomly chosen} $[n,k]$ binary
linear error-correcting code $\sC$ has another word $\bz \in \bl + \sC$
(namely, a word $\bz \in \mathbf{F}_2^n$
such that $\bl \oplus \bz \in \sC$)
that satisfies $\bz \ne \bl$ and $|\bz| \le t$ is
\begin{equation}
\pr_\sC \left[ \exists \bz \in \bl + \sC : \bz \ne \bl, |\bz| \le t
\ \mid \ \bl \right] \le 2^{n[H_2(t/n) - r/n]},
\end{equation}
where $r \triangleq n - k$ and
$H_2(x) \triangleq -x \log_2(x) - (1-x) \log_2(1-x)$.
\end{proposition}
\begin{proof}
This proof is based on the proof
of~\cite[Lemma~4.15 in page~115]{RothBook}.
Let $\bl \in \mathbf{F}_2^n$ be a specific word.
If we choose a random $[n,k]$ binary linear code $\sC$
by choosing a uniformly random $k \times n$ generator matrix $G_\sC$
over $\mathbf{F}_2$~\footnote{In fact, it is possible that the randomly
chosen generator matrix
$G_\sC$ satisfies $\mathrm{rank}(G_\sC) < k$, which gives a ``bad''
code (not a valid $[n,k]$ error-correcting code). In this case,
we can assume the process is run again
until we get a valid code---namely, until we get a matrix $G_\sC$
satisfying $\mathrm{rank}(G_\sC) = k$.},
then for any fixed non-zero vector
$\bu \in \mathbf{F}_2^k \setminus \{\mathbf{0}\}$,
the resulting product $\bu G_\sC$ is uniformly distributed
over $\mathbf{F}_2^n \setminus \{\mathbf{0}\}$,
and so the sum $\bl \oplus \bu G_\sC$ is uniformly distributed
over $\mathbf{F}_2^n \setminus \{\bl\}$. Therefore:
\begin{eqnarray}
\pr_\sC \left[ |\bl \oplus \bu G_\sC| \le t \ \mid \ \bl \right]
&=& \sum_{\bz \in \mathbf{F}_2^n ~ : ~ |\bz| \le t}
\pr_\sC \left[ \bl \oplus \bu G_\sC = \bz \ \mid \ \bl \right] \nonumber \\
&=& \sum_{\bz \in \mathbf{F}_2^n \setminus \{\bl\} ~ : ~ |\bz| \le t}
\frac{1}{2^n - 1} \le \frac{1}{2^n - 1} \cdot V_2(n,t),
\end{eqnarray}
where $V_2(n,t)$ is the number of words $\bz \in \mathbf{F}_2^n$
that satisfy $|\bz| \le t$
(see definition in~\cite[page~95]{RothBook}).
In~\cite[Lemma~4.7 in page~105]{RothBook} it was proved that
$V_2(n,t) \le 2^{n H_2(t/n)}$
(assuming $0 \le \frac{t}{n} \le \frac{1}{2}$), so:
\begin{equation}
\pr_\sC \left[ |\bl \oplus \bu G_\sC| \le t \ \mid \ \bl \right]
\le \frac{1}{2^n - 1} \cdot V_2(n,t)
\le \frac{1}{2^n - 1} \cdot 2^{n H_2(t/n)}.
\end{equation}

We can therefore compute:
\begin{eqnarray}
\pr_\sC \left[ \exists \bz \in \bl + \sC : \bz \ne \bl, |\bz| \le t
\ \mid \ \bl \right]
&\le& \pr_\sC \left[ \exists \bu \in \mathbf{F}_2^k
\setminus \{\mathbf{0}\} : |\bl \oplus \bu G_\sC|
\le t \ \mid \ \bl \right] \nonumber \\
&\le& \sum_{\bu \in \mathbf{F}_2^k \setminus \{\mathbf{0}\}}
\pr_\sC \left[ |\bl \oplus \bu G_\sC| \le t \ \mid \ \bl \right]
\nonumber \\
&\le& \sum_{\bu \in \mathbf{F}_2^k \setminus \{\mathbf{0}\}}
\frac{1}{2^n - 1} \cdot 2^{n H_2(t/n)} \nonumber \\
&=& \frac{2^k - 1}{2^n - 1} \cdot 2^{n H_2(t/n)}
\le \frac{2^k}{2^n} \cdot 2^{n H_2(t/n)} \nonumber \\
&=& 2^{n H_2(t/n) - (n-k)} = 2^{n [H_2(t/n) - r/n]}.
\end{eqnarray}
\end{proof}

Finally, we combine all results to study the failure probability
of correcting a specific error $\be$ in a randomly chosen code $\sC$:
\begin{corollary}\label{corollary_ecc_decoding}
Let $n, k, t \ge 1$ be integers
such that $0 \le \frac{t}{n} \le \frac{1}{2}$.
For any error string $\be \in \mathbf{F}_2^n$ satisfying $|\be| \le t$,
the probability that a {\em randomly chosen} $[n,k]$ binary
linear error-correcting code $\sC$ cannot correct the error $\be$
using a nearest-codeword decoder is
\begin{equation}
\pr_\sC \left[ \sC\textrm{ cannot correct }\be \mid \be \right]
\le 2^{n[H_2(t/n) - r/n]},
\end{equation}
where $r \triangleq n - k$ and
$H_2(x) \triangleq -x \log_2(x) - (1-x) \log_2(1-x)$.
\end{corollary}
\begin{proof}
According to Lemmas~\ref{lemma_ecc_indep}
and~\ref{lemma_ecc_indep_coset}, the error word $\be$
is correctly decoded by a nearest-codeword decoder if and only if
the codeword nearest to $\be$ is uniquely $\mathbf{0}$.
In other words, $\be$ is wrongly decoded by $\sC$ if and only if
there exists a codeword $\bc \in \sC \setminus \{\mathbf{0}\}$
such that $|\be \oplus \bc| \le |\be| \le t$,
which implies that $\bz \triangleq \be \oplus \bc$ satisfies
$\bz \in \be + \sC$, $\bz \ne \be$, and $|\bz| \le t$.
Therefore, according to Proposition~\ref{prop_ecc_random}:
\begin{equation}
\pr_\sC \left[ \sC\textrm{ cannot correct }\be \mid \be \right]
\le \pr_\sC \left[ \exists \bz \in \be + \sC : \bz \ne \be, |\bz| \le t
\ \mid \ \be \right] \le 2^{n[H_2(t/n) - r/n]}.
\end{equation}
\end{proof}

\section{\label{generalized-bb84_description}Full Definition
of the Generalized BB84 Protocols}
The protocols for which we prove security in this paper belong to
a generalized class of BB84-like protocols. Below we formally define
this general class of protocols. Some of the details in this definition
are decided by each specific protocol, but most details are
shared by all protocols.
\begin{enumerate}
\item \label{protocol_pre_parameters} Before the protocol begins,
Alice and Bob choose some shared (and public) parameters:
the integer numbers $N, n, r, m$ (such that $r + m \le n < N$),
the sets $B$ and $\{S_\bb\}_{\bb \in B}$
and probability distributions over them
(decided by the specific protocol)
that will control the future choice of bitstrings
$\bb, \bs \in \mathbf{F}^{N}_2$,
and the testing function $T$ (decided by the specific protocol).

Formally, for choosing the sets $B$ and $\{S_\bb\}_{\bb \in B}$
and the corresponding probability distributions,
Alice and Bob should choose
the set $B \subseteq \mathbf{F}^{N}_2$ of basis strings,
the probabilities $\pr(\bb)$ for all $\bb \in B$,
the sets $S_\bb \subseteq \mathbf{F}^{N}_2$ of $\bs$ strings
for all $\bb \in B$,
and the probabilities $\pr(\bs \mid \bb)$
for all $\bb \in B$ and $\bs \in S_\bb$.
We require that $|\bs| = n$ for all $\bs \in S_\bb$.
(Equivalently, Alice and Bob choose a subset $R$ of the set
$\{(\bb,\bs) \mid \bb, \bs \in \mathbf{F}^{N}_2, |\bs| = n\}$ and
a probability distribution $\pr(\bb,\bs)$.)
The testing function $T:\mathbf{F}^{N-n}_2 \times \mathbf{F}^{N-n}_2
\times \mathbf{F}^{N}_2 \rightarrow \{0, 1\}$
must get $(\bi_\sT \oplus \bj_\sT, \bb_\sT, \bs)$ as inputs
and give $0$ or $1$ as an output.
In Section~\ref{sec_protocols} we give examples of protocols
and their formal definitions using these notations.

\item \label{protocol_send_quantum} Alice randomly chooses
an $N$-bit string $\bi \in \mathbf{F}^{N}_2$
(the choice of $\bi$ is uniformly random),
an $N$-bit string $\bb \in B$,
an $N$-bit string $\bs \in S_\bb$ (that must satisfy $|\bs| = n$),
an $r \times n$ parity-check matrix $P_\sC$
(corresponding to a linear error-correcting code $\sC$),
and an $m \times n$ privacy amplification matrix $P_\sK$
(representing a linear key-generation function).
It is required that {\em all} $r+m$ rows of the matrices
$P_\sC$ and $P_\sK$ put together are linearly independent;
both matrices are chosen uniformly at random under this condition,
and Alice keeps them secret at the current sage.

Then, Alice sends the $N$ qubit states
$\ket{i_1}_{b_1}, \ket{i_2}_{b_2}, \ldots, \ket{i_{N}}_{b_{N}}$,
one after the other, to Bob using the quantum channel.
(The probability distributions $\pr(\bb),\pr(\bs \mid \bb)$ and
the sets $B,S_\bb$ were chosen in Step~\ref{protocol_pre_parameters}
according to the specific protocol.)
Bob keeps each received qubit in a quantum memory,
not measuring it yet\footnote{
Here we assume that Bob has a quantum memory and can delay his
measurement. In practical implementations, Bob usually cannot do that,
but he is assumed to choose his own random basis string $\bb'' \in B$
and measure in the bases it dictates;
later, Alice and Bob discard the qubits measured in the wrong
basis. In that case, we need to assume that Alice sends more than
$N$ qubits, so that $N$ qubits are finally detected by Bob
and measured in the correct basis. In Appendix~A
of~\cite{BBBMR06}
it is explained why this change of the protocol
does not hurt security.}.

\item \label{protocol_send_bases} Alice sends to Bob
over the classical channel the bitstring $\bb = b_1 \ldots b_N$.
Bob measures each of the qubits he saved in the correct basis (namely,
when measuring the $i$-th qubit, he measures it in the $z$ basis
if $b_i = 0$, and he measures it in the $x$ basis if $b_i = 1$).

The bitstring measured by Bob is denoted by $\bj$.
The XOR of $\bi$ and $\bj$ is denoted $\bc \triangleq \bi \oplus \bj$.
If there is no noise and no eavesdropping, then $\bi = \bj$
(that is, $\bc = \mathbf{0}$).

\item Alice sends $\bs$ to Bob over the classical channel.
The INFO bits (that will be used for creating the final key)
are the $n$ bits with $s_j = 1$,
while the TEST bits (that will be used for testing)
are the $N - n$ bits with $s_j = 0$.
We denote the substrings of $\bi, \bj, \bc, \bb$ that correspond to
the INFO bits by $\bi_\sI$, $\bj_\sI$, $\bc_\sI$,
and $\bb_\sI$, respectively;
and we denote the substrings of $\bi, \bj, \bc, \bb$ that correspond to
the TEST bits by $\bi_\sT$, $\bj_\sT$, $\bc_\sT$,
and $\bb_\sT$, respectively.

\item \label{protocol_calc_test} Alice and Bob both publish
the $N - n$ bit values they have for all TEST bits
($\bi_\sT$ and $\bj_\sT$, respectively),
and they compute their XOR $\bc_\sT = \bi_\sT \oplus \bj_\sT$.
They compute $T(\bc_\sT, \bb_\sT, \bs)$:
if it is $0$, they abort the protocol;
if it is $1$, they continue the run of the protocol.
(The testing function $T$ was chosen
in Step~\ref{protocol_pre_parameters}
according to the specific protocol.)

\item The values of the $n$ remaining bits
(the INFO bits, with $s_j = 1$) are kept in secret by Alice and Bob.
The bitstring of Alice is $\bi_\sI$,
the bitstring of Bob is $\bj_\sI$,
and their XOR is $\bc_\sI$.

\item Alice sends to Bob over the classical channel
her chosen parity-check matrix $P_\sC$
(that corresponds to the error-correcting code $\sC$)
and the {\em syndrome} of $\bi_\sI$ with respect to $\sC$,
which consists of $r$ bits and is defined as
$\bxi \triangleq \bi_\sI P_\sC^\transp$.
Bob uses $\bxi$ to correct the errors in his $\bj_\sI$ string
(so that it should be the same as $\bi_\sI$)
using the ``code coset'' decoding method presented
in Subsubsection~\ref{subsubsec_prelims_ecc}.

\item Alice sends to Bob over the classical channel
her chosen privacy amplification matrix $P_\sK$.
The final key $\bk$ consists of $m$ bits and is defined as
$\bk \triangleq \bi_\sI P_\sK^\transp$; both Alice and Bob compute it.
\end{enumerate}

\section{Bound on the Security Definition
for Generalized BB84 Protocols}
\subsection{\label{subsec_invb}The Hypothetical
``inverted-INFO-basis'' Protocol}
Our security proof uses an alternative, hypothetical protocol
in which Alice sends to Bob the qubits after
inverting the bases of the INFO bits
(without changing the bases of the TEST bits).
We call this protocol ``hypothetical'' because
it is never actually used by Alice and Bob,
and we do not perform any reduction to it (or from it),
but we compute probabilities of certain events
in the hypothetical protocol for use in our security bound.
In particular, we use the error rate in the hypothetical protocol
for bounding the trace distance in the security definition
of the real protocol.

In the hypothetical protocol, Alice, Bob, and Eve
do everything exactly as they
would do in the real protocol, except that Alice and Bob use
(and publish) the basis string $\bb^0 \triangleq \bb \oplus \bs$
instead of $\bb$:
namely, they use the basis string $\bb_\sT$ for the TEST bits
and the basis string $\overline{\bb_\sI}$
(the bitwise NOT of $\bb_\sI$) for the INFO bits.

Formally, this hypothetical protocol is defined by replacing
Steps~\ref{protocol_send_quantum}--\ref{protocol_send_bases}
of the original protocol
(as described in Section~\ref{generalized-bb84_description})
by the following steps:
\begin{enumerate}\addtocounter{enumi}{1}
\item \label{invb_protocol_send_quantum} Alice randomly chooses
an $N$-bit string $\bi \in \mathbf{F}^{N}_2$,
an $N$-bit string $\bb \in B$,
an $N$-bit string $\bs \in S_\bb$ (that must satisfy $|\bs| = n$),
an $r \times n$ parity-check matrix $P_\sC$,
and an $m \times n$ privacy amplification matrix $P_\sK$.

Then, Alice computes the $N$-bit string
$\bb^0 \triangleq \bb \oplus \bs$,
and she sends the $N$ qubit states
$\ket{i_1}_{b^0_1}, \ket{i_2}_{b^0_2}, \ldots, \ket{i_{N}}_{b^0_{N}}$,
one after the other, to Bob using the quantum channel.
Bob keeps each received qubit in a quantum memory,
not measuring it yet.

\item \label{invb_protocol_send_bases} Alice sends to Bob
over the classical channel the bitstring $\bb^0 = b^0_1 \ldots b^0_N$.
Bob measures each of the qubits he saved in the correct basis (namely,
when measuring the $i$-th qubit, he measures it in the $z$ basis
if $b^0_i = 0$, and he measures it in the $x$ basis if $b^0_i = 1$).
\end{enumerate}

We notice that in this protocol, Alice {\em chooses} $\bb$ and $\bs$
in the same way as she would choose them in the real protocol, but
{\em uses} (and sends to Bob for his use) $\bb^0$ and $\bs$ instead.

In the security proof, we will use the notation of $\pr_\invb$
for computing the probability of a certain event assuming that
Alice and Bob use the hypothetical protocol.
In particular, we note that $\pr_\invb(\cdot \mid \bb, \bs)$
is a conditional probability on Alice {\em choosing} the bitstrings
$\bb, \bs$ (while she actually {\em uses} the basis string $\bb^0$).

We should note that $\pr_\invb(\cdot \mid \bb, \bs)$ is
the same as $\pr(\cdot \mid \bb^0, \bs)$
{\em if both notations are well-defined}:
namely, the hypothetical protocol given that Alice chooses $\bb, \bs$
(and thus uses $\bb^0, \bs$)
is the same as the real protocol given that Alice chooses $\bb^0, \bs$.
However, the second notation is not always well-defined,
because it may be the case that $\bb \in B$ while $\bb^0 \notin B$,
or that $\bs \in S_\bb$ while $\bs \notin S_{\bb^0}$;
therefore, it may be the case that $\bb^0$
is not an allowed basis string for the real protocol.
In the standard BB84 protocol (see Subsection~\ref{subsec_bb84}),
such problems are impossible, and this is
why~\cite{BBBMR06} used
the notation of $\pr(\cdot \mid \bb^0, \bs)$ instead of
$\pr_\invb(\cdot \mid \bb, \bs)$; however, in our paper, we discuss
generalized BB84 protocols, so we must use the notation of
$\pr_\invb(\cdot \mid \bb, \bs)$.\footnote{
It is also possible that $\pr(\bb, \bs) \ne \pr(\bb^0, \bs)$,
in which case the use of $\bb^0, \bs$ in the real protocol does not
happen with the same probability as the use of $\bb^0, \bs$
in the hypothetical protocol. However, if both probabilities
$\pr(\bb, \bs), \pr(\bb^0, \bs)$ are non-zero,
then $\pr_\invb(\cdot \mid \bb, \bs) = \pr(\cdot \mid \bb^0, \bs)$.}

\subsection{\label{sec_joint}The General Joint Attack of Eve}
Before the QKD protocol is performed (and, thus, independently of
$\bi$, $\bb$, and $\bs$), Eve chooses a specific joint
attack she wants to perform. In a {\em joint attack},
all qubits are attacked
using one giant probe (ancillary state) kept by Eve.
Eve saves her probe in a quantum memory
and can keep it indefinitely, even after the final key
is used by Alice and Bob; and she can perform,
at any time of her choice, an optimal measurement
of her giant probe, chosen based on all the information
she has at the time of the measurement (including the classical
information sent during the protocol, and including information
she acquires when Alice and Bob use the final key).

Given that Alice sends to Bob the state
$\ket{\bi}_{\bb} \triangleq \otimes_{j = 1}^N \ket{i_j}_{b_j}$
(namely, the $N$-bit string is $\bi$
and the $N$-bit basis string is $\bb$),
Eve attaches a probe state $\ket{0}_\sE$
and applies a unitary operator $U$ of her choice to the
compound system $\ket{0}_\sE \ket{\bi}_{\bb}$.
Then, Eve keeps to herself (in a quantum memory) her probe state,
and she sends to Bob the $N$-qubit quantum state sent from Alice to Bob
(which may have been modified due to her attack $U$).

The most general joint attack $U$ of Eve is
\begin{equation}
U \ket{0}_\sE \ket{\bi}_{\bb} = \sum_{\bj \in \mathbf{F}_2^N}
\ket{E'_{\bi,\bj}}_\bb \ket{\bj}_{\bb},
\end{equation}
where $\ket{E'_{\bi,\bj}}_\bb$ are non-normalized states
in Eve's probe system. We note that\footnote{In~\cite{BBBMR06},
no probabilities were conditioned on $P_\sC, P_\sK$,
because these matrices were assumed to be constant and public.
In our paper we assume $P_\sC,P_\sK$ to be randomly chosen,
so it appears all probability distributions based on~\cite{BBBMR06}
should have been conditioned on $P_\sC, P_\sK$;
however, in our paper Alice chooses $P_\sC, P_\sK$ uniformly at random,
keeps them secret, and does not use them until
Eve's attack is over, so probabilities related only
to Alice's other random choices ($\bi, \bb, \bs$)
and Eve's attack ($\bj, \bc \triangleq \bi \oplus \bj$)
are completely independent of $P_\sC, P_\sK$.}
\begin{equation}
\braket{E'_{\bi,\bj}}{E'_{\bi,\bj}}_\bb =
\pr(\bj \mid \bi, \bb, \bs) =
\pr(\bj \mid \bi, \bb, \bs, P_\sC, P_\sK).
\end{equation}
Writing the INFO and TEST bits of Alice and Bob separately
($\bi_\sI, \bi_\sT$ instead of $\bi$,
and $\bj_\sI, \bj_\sT$ instead of $\bj$),
we can denote $\ket{E'_{\bi,\bj}}_\bb$
by $\ket{E'_{\bi_\sI,\bi_\sT,\bj_\sI,\bj_\sT}}_\bb$.

Subsection~3.4 of~\cite{BBBMR06} introduced the notation of
$\ket{E_{\bi_\sI,\bj_\sI}}_{\bb, \bs}$;
this notation is useful because it treats $\bi_\sT$ and $\bj_\sT$
as constants, assuming their values to be known
(since they are ultimately published by Alice and Bob,
and then they are known to Eve).
It is defined as
\begin{equation}
\ket{E_{\bi_\sI,\bj_\sI}}_{\bb, \bs} \triangleq
\frac{1}{\sqrt{\pr(\bj_\sT \mid \bi_\sI, \bi_\sT, \bb, \bs)}}
\ket{E'_{\bi_\sI,\bi_\sT,\bj_\sI,\bj_\sT}}_\bb.
\end{equation}
We note that $\ket{E_{\bi_\sI,\bj_\sI}}_{\bb, \bs}$ also depends on
$\bi_\sT,\bj_\sT$ (and not only on $\bi_\sI, \bj_\sI, \bb, \bs$).
According to Equations~(3.22)--(3.23) of~\cite{BBBMR06},
assuming that Alice chooses $\bi_\sI, \bi_\sT, \bb, \bs$
and Bob measures $\bj_\sT$, the normalized state of Eve and Bob is
\begin{equation}
\ket{\psi_{\bi_\sI}} = \sum_{\bj_\sI \in \mathbf{F}_2^n}
\ket{E_{\bi_\sI,\bj_\sI}}_{\bb, \bs} \ket{\bj_\sI}_{\bb},
\end{equation}
and it also holds that
\begin{equation}
\braket{E_{\bi_\sI,\bj_\sI}}{E_{\bi_\sI,\bj_\sI}}_{\bb, \bs} =
\pr(\bj_\sI \mid \bi_\sI, \bi_\sT, \bj_\sT, \bb, \bs).
\end{equation}

Let us define $\left( \rho_{\bi_\sI,\bj_\sI}^{\bb, \bs} \right)_\sE$
(which also depends on $\bi_\sT, \bj_\sT$) to be the normalized
state of Eve if Alice chooses $\bi_\sI, \bi_\sT, \bb, \bs$
and Bob measures $\bj_\sI, \bj_\sT$. In other words,
$\left( \rho_{\bi_\sI,\bj_\sI}^{\bb, \bs} \right)_\sE$
is the normalized density matrix of both
$\ket{E_{\bi_\sI,\bj_\sI}}_{\bb, \bs}$
and $\ket{E'_{\bi, \bj}}_\bb$, so
\begin{equation}\label{rho_Eve_ij}
\left( \rho_{\bi_\sI,\bj_\sI}^{\bb, \bs} \right)_\sE \triangleq
\frac{\ket{E_{\bi_\sI,\bj_\sI}}_{\bb, \bs}
\bra{E_{\bi_\sI,\bj_\sI}}_{\bb, \bs}}
{\pr(\bj_\sI \mid \bi_\sI, \bi_\sT, \bj_\sT, \bb, \bs)} =
\frac{\ket{E'_{\bi, \bj}}_\bb \bra{E'_{\bi, \bj}}_\bb}
{\pr(\bj \mid \bi, \bb, \bs)} =
\frac{\ket{E'_{\bi, \bj}}_\bb \bra{E'_{\bi, \bj}}_\bb}
{\pr(\bj \mid \bi, \bb, \bs, P_\sC, P_\sK)}.
\end{equation}

Eve's state after her attack (tracing out Bob) is
\begin{equation}\label{rho_Eve}
\left( \rho^{\bi_\sI} \right)_\sE
\triangleq \tr_{\mathrm{Bob}}(\ket{\psi_{\bi_\sI}}
\bra{\psi_{\bi_\sI}}) = \sum_{\bj_\sI \in \mathbf{F}_2^n}
\ket{E_{\bi_\sI,\bj_\sI}}_{\bb, \bs}
\bra{E_{\bi_\sI,\bj_\sI}}_{\bb, \bs}
= \sum_{\bj_\sI \in \mathbf{F}_2^n}
\pr(\bj_\sI \mid \bi_\sI, \bi_\sT, \bj_\sT, \bb, \bs)
\left( \rho_{\bi_\sI,\bj_\sI}^{\bb, \bs} \right)_\sE,
\end{equation}
and according to Equation~(4.2) of~\cite{BBBMR06},
we define the purification $\ket{\varphi_{\bi_\sI}}_\sE$
of $\left( \rho^{\bi_\sI} \right)_\sE$
(so that $\left( \rho^{\bi_\sI} \right)_\sE$
is a partial trace of $\ket{\varphi_{\bi_\sI}}_\sE$) as
\begin{equation}\label{phi_Eve_purified}
\ket{\varphi_{\bi_\sI}}_\sE
\triangleq \sum_{\bj_\sI \in \mathbf{F}_2^n}
\ket{E_{\bi_\sI,\bj_\sI}}_{\bb, \bs} \ket{\bi_\sI \oplus \bj_\sI}.
\end{equation}

We also define a ``Fourier basis'' of non-normalized states
$\{\ket{\eta_\bl}_\sE\}_{\bl \in \mathbf{F}_2^n}$
(according to Definition~4.2 in~\cite{BBBMR06}) as
\begin{equation}
\ket{\eta_\bl}_\sE \triangleq \frac{1}{2^n}
\sum_{\bi_\sI \in \mathbf{F}_2^n}
(-1)^{\bi_\sI \cdot \bl} \ket{\varphi_{\bi_\sI}}_\sE,
\end{equation}
define the two notations $d_\bl$ and $\ket{\widehat{\eta}_\bl}_\sE$
(according to the same definition in~\cite{BBBMR06}) as
\begin{eqnarray}
d_\bl
&\triangleq& \sqrt{\braket{\eta_\bl}{\eta_\bl}_\sE},\label{phi_d}
\label{norm_eta}\\
\ket{\widehat{\eta}_\bl}_\sE
&\triangleq& \frac{\ket{\eta_\bl}_\sE}{d_\bl},\label{normalized_eta}
\end{eqnarray}
and deduce (according to Equation~(4.4) in~\cite{BBBMR06}) that
\begin{equation}\label{phi_as_eta}
\ket{\varphi_{\bi_\sI}}_\sE = \sum_{\bl \in \mathbf{F}_2^n}
(-1)^{\bi_\sI \cdot \bl} \ket{\eta_\bl}_\sE.
\end{equation}

\subsection{The Symmetrized Attack of Eve}
In~\cite{BBBMR06}, the most general joint attack was
not directly analyzed: for simplicity,
it was assumed that Eve applies a process
named {\em symmetrization}, resulting in a {\em symmetrized} attack.
The process of symmetrization is always beneficial to Eve
(it does not change the error rate, and we prove in
Proposition~\ref{prop_symm_general} that it does not decrease
Eve's information), so security against all symmetrized attacks
implies security against all possible joint attacks.

In Eve's original attack, she has her own probe subsystem $\sE$;
in the symmetrization process, Eve adds another probe subsystem $\sM$,
in the initial state of $\ket{0_x}_\sM \triangleq \frac{1}{\sqrt{2^N}}
\sum_{\bbm \in \mathbf{F}_2^N} \ket{\bbm}_\sM$.
Given the original attack
$U$ (applied to Alice's qubits and the probe $\sE$),
the symmetrized attack $U^\sym$ (applied to Alice's qubits
and both probes $\sE$ and $\sM$) is defined by
\begin{equation}
U^\sym \triangleq (\mathbf{I}_\sE \otimes S^\dagger)
(U \otimes \mathbf{I}_\sM) (\mathbf{I}_\sE \otimes S),
\end{equation}
where $S$ is a unitary operation applied to Alice's qubits
and to the probe $\sM$, and it operates as follows:
\begin{equation}
S \ket{\bi}_\bb \ket{\bbm}_\sM = (-1)^{(\bi \oplus \bb) \cdot \bbm}
\ket{\bi \oplus \bbm}_\bb \ket{\bbm}_\sM.
\end{equation}
Intuitively, Eve first XORs Alice's bit values with a random string
$\bbm$ (kept by her), then applies her original attack,
and finally reverses the XOR with $\bbm$.
Full definition and explanations are available in Subsection~3.1
of~\cite{BBBMR06}.

In this paper we use several properties of the symmetrized
attack. First of all, the ``Basic Lemma of Symmetrization'' (Lemma~3.1
of~\cite{BBBMR06}) gives the expression for
$\ket{E^\sym_{\bi, \bj}\,\!'}_\bb$ (of the symmetrized attack)
as a function of $\ket{E'_{\bi, \bj}}_\bb$ (of the original attack):
\begin{equation}\label{symm_basic_lemma}
\ket{E^\sym_{\bi, \bj}\,\!'}_\bb = \frac{1}{\sqrt{2^N}}
\sum_{\bbm \in \mathbf{F}_2^N} (-1)^{(\bi \oplus \bj) \cdot \bbm}
\ket{E'_{\bi \oplus \bbm, \bj \oplus \bbm}}_\bb \ket{\bbm}_\sM.
\end{equation}

The second property we use, proved in Corollary~3.3
of~\cite{BBBMR06},
is the fact that the probabilities of error strings $\bc_\sI$ and
$\bc_\sT$ (if {\em not conditioning} on $\bi$) are not affected
by the symmetrization. Namely,
\begin{equation}\label{symm_probs_c}
\pr^\sym(\bc_\sI, \bc_\sT \mid \bb, \bs) =
\pr(\bc_\sI, \bc_\sT \mid \bb, \bs).
\end{equation}
This property is true for all basis strings $\bb$. In particular,
it is true for the basis string $\bb^0 \triangleq \bb \oplus \bs$
used in the hypothetical ``inverted-INFO-basis'' protocol defined in
Subsection~\ref{subsec_invb}; therefore,
\begin{equation}\label{symm_probs_c_invb}
\pr_\invb^\sym(\bc_\sI, \bc_\sT \mid \bb, \bs) =
\pr_\invb(\bc_\sI, \bc_\sT \mid \bb, \bs).
\end{equation}

The third property we use, proved in Lemma~3.8
of~\cite{BBBMR06},
is the fact that the probabilities for errors in the TEST bits
are not affected by the bases used for the INFO bits:
\begin{equation}\label{symm_prob_test_err}
\pr^\sym(\bj_\sT \mid \bi_\sT, \bb, \bs) =
\pr^\sym(\bj_\sT \mid \bi_\sT, \bb_\sT, \bs).
\end{equation}
In particular, since the only difference between the hypothetical
``inverted-INFO-basis'' protocol and the real protocol is
the basis string used for the {\em INFO bits}
($\overline{\bb_\sI}$ and $\bb_\sI$, respectively),
this property means that the probabilities of errors in the TEST bits
are identical for both protocols:
\begin{equation}\label{symm_prob_test_err_invb}
\pr^\sym_\invb(\bj_\sT \mid \bi_\sT, \bb, \bs) =
\pr^\sym(\bj_\sT \mid \bi_\sT, \bb, \bs).
\end{equation}

The fourth property we use, proved in Corollary~3.6
of~\cite{BBBMR06},
is the fact that the probability of any string of INFO bits $\bi_\sI$
is uniform (that is, $\frac{1}{2^n}$) even when conditioning on
the four parameters $\bi_\sT, \bj_\sT, \bb, \bs$
that are ultimately known to Eve
(we note that $\bj_\sT$ is affected by Eve's attack). Namely,
\begin{equation}\label{symm_prob_info_bits}
\pr^\sym(\bi_\sI \mid \bi_\sT, \bj_\sT, \bb, \bs) =
\frac{1}{2^n}.
\end{equation}
In particular, the same is true when conditioning on $P_\sC, P_\sK$
(because, as noted in Subsection~\ref{sec_joint},
the randomly chosen matrices $P_\sC, P_\sK$
are completely independent of $\bi, \bj, \bb, \bs$). Thus,
\begin{equation}\label{symm_prob_info_bits_PC_PK}
\pr^\sym(\bi_\sI \mid \bi_\sT, \bj_\sT, \bb, \bs, P_\sC, P_\sK) =
\pr^\sym(\bi_\sI \mid \bi_\sT, \bj_\sT, \bb, \bs) =
\frac{1}{2^n}.
\end{equation}
Given specific $P_\sC, P_\sK$, we can thus compute the probability
of each specific syndrome $\bxi$
(namely, the probability that for a specific value of
$\bxi \in \mathbf{F}_2^r$, the chosen value of $\bi_\sI$
satisfies $\bi_\sI P_\sC^\transp = \bxi$):
\begin{eqnarray}
\pr^\sym(\bxi \mid \bi_\sT, \bj_\sT, \bb, \bs, P_\sC, P_\sK)
&=& \sum_{\bi_\sI\ \mid\ \bi_\sI P_\sC^\transp = \bxi}
\pr^\sym(\bi_\sI \mid \bi_\sT, \bj_\sT, \bb, \bs, P_\sC, P_\sK)
\nonumber \\
&=& \sum_{\bi_\sI\ \mid\ \bi_\sI P_\sC^\transp = \bxi} \frac{1}{2^n}
= 2^{n-r} \cdot \frac{1}{2^n} = \frac{1}{2^r}
\label{symm_prob_info_bits_xi_1}
\end{eqnarray}
(because the matrix $P_\sC$ has rank $r$,
so for each specific syndrome $\bxi \in \mathbf{F}_2^r$
there are exactly $2^{n-r}$ values of $\bi_\sI$
satisfying $\bi_\sI P_\sC^\transp = \bxi$).
In particular,
\begin{eqnarray}
\pr^\sym(\bxi \mid \bi_\sT, \bj_\sT, \bb, \bs)
&=& \sum_{P_\sC,P_\sK}
\pr^\sym(P_\sC, P_\sK \mid \bi_\sT, \bj_\sT, \bb, \bs)
\cdot \pr^\sym(\bxi \mid \bi_\sT, \bj_\sT, \bb, \bs, P_\sC, P_\sK)
\nonumber \\
&=& \sum_{P_\sC,P_\sK} \pr^\sym(P_\sC, P_\sK
\mid \bi_\sT, \bj_\sT, \bb, \bs) \cdot \frac{1}{2^r}
= \frac{1}{2^r},\label{symm_prob_info_bits_xi_2}
\end{eqnarray}
so
\begin{equation}\label{symm_prob_info_bits_xi_3}
\pr^\sym(\bxi \mid \bi_\sT, \bj_\sT, \bb, \bs, P_\sC, P_\sK)
= \frac{1}{2^r}
= \pr^\sym(\bxi \mid \bi_\sT, \bj_\sT, \bb, \bs).
\end{equation}
In addition, we can find that for any $\bi_\sI, P_\sC, P_\sK$ and
$\bxi \triangleq \bi_\sI P_\sC^\transp$,
according to Equations~\eqref{symm_prob_info_bits_PC_PK}
and~\eqref{symm_prob_info_bits_xi_1}:
\begin{equation}\label{symm_prob_info_bits_iI}
\pr^\sym(\bi_\sI \mid \bi_\sT, \bj_\sT, \bb, \bs, \bxi, P_\sC, P_\sK)
= \frac{\pr^\sym(\bi_\sI \mid \bi_\sT, \bj_\sT, \bb, \bs, P_\sC,
P_\sK)}{\pr^\sym(\bxi \mid \bi_\sT, \bj_\sT, \bb, \bs, P_\sC, P_\sK)}
= \frac{1 / 2^n}{1 / 2^r} = \frac{1}{2^{n-r}}.
\end{equation}

\subsection{\label{subsec_ivd}Main Result: Connecting Information to Disturbance}
Our security proof of the generalized BB84 protocols is very
similar to the security proof of BB84 itself, that was detailed
in~\cite{BBBMR06}. Most parts of the proof are
not affected at all by the changes made to BB84 to get the
generalized BB84 protocols (changes detailed in
Section~\ref{generalized-bb84_description} of the current paper),
because these parts assume fixed strings $\bb,\bs$
and fixed matrices $P_\sC,P_\sK$;
therefore, in some parts of our current proofs,
we refer to results from~\cite{BBBMR06}.

However, in this Subsection we improve
a pivotal result of~\cite{BBBMR06},
resulting in a significant improvement of the allowed error rate
derived from our security proof.

Given the error-correction parity-check matrix $P_\sC$
(which corresponds to the error-correcting code $\sC$)
and the privacy amplification matrix $P_\sK$,
we denote the rows of $P_\sC$
as the vectors $v_1, \ldots, v_r$ in $\mathbf{F}_2^n$,
and the rows of $P_\sK$
as the vectors $v_{r+1}, \ldots, v_{r+m}$;
these vectors must all be linearly independent.
We choose $n-r-m$ additional vectors
$\{v_{r+m+1}, \ldots, v_n\}$ that extend $\{v_1, \ldots, v_{r+m}\}$
to a basis of $\mathbf{F}_2^n$.
We thus also denote, for any $1 \le r' \le r+m$,
\begin{eqnarray}
V_{r'} &\triangleq& \Span \lbrace v_1, \ldots, v_{r'} \rbrace, \\
V_{r'}^c &\triangleq& \Span \lbrace v_{r'+1}, \ldots, v_n \rbrace, \\
V_{r'}^\except &\triangleq& \Span \lbrace v_1, \ldots, v_{r'-1},
v_{r'+1}, \ldots, v_{r+m} \rbrace,
\end{eqnarray}
which means that $V_{r'}$ and $V_{r'}^c$ are two complementary vector
spaces satisfying $\mathbf{F}_2^n = V_{r'}^c \oplus V_{r'}$ (namely,
each vector $\bl \in \mathbf{F}_2^n$ has a unique representation
$\bl = \bx \oplus \by$ with $\bx \in V_{r'}^c, \by \in V_{r'}$),
and $V_{r'}^\except$ is the $(r+m-1)$-dimensional
vector space that spans all error correction and
privacy amplification vectors except $v_{r'}$.

For a $1$-bit final key $k \in \{0, 1\}$ (that is, for $m = 1$),
and given a symmetrized attack of Eve,
we define $\widehat{\rho}^\sym_k$
to be the state of Eve corresponding to the final key $k$,
given that she knows $\bxi$, averaging over the matrices $P_\sC, P_\sK$
(which are ultimately known to Eve, and are thus included in the state
as classical information that will eventually be available to Eve)
and the bitstring $\bi_\sI$ (which is not ultimately known to Eve,
and is forgotten by Alice and Bob---who will only remember the final
key $k$---and is thus not included as classical information). Thus,
\begin{equation}\label{rhohatk}
\widehat{\rho}^\sym_k \triangleq \frac{1}{2^{n-r-1}}\sum_{P_\sC, P_\sK,
\bi_\sI\,\big|{\scriptsize\begin{matrix}\bi_\sI P_\sC^\transp = \bxi\\
\bi_\sI \cdot v_{r+1} = k\end{matrix}}} \pr^\sym(P_\sC, P_\sK) \cdot
\left( \rho^{\bi_\sI} \right)_\sE^\sym \otimes \ket{P_\sC, P_\sK}_\sC
\bra{P_\sC, P_\sK}_\sC,
\end{equation}
where $\left( \rho^{\bi_\sI} \right)_\sE^\sym$,
as defined in Equation~\eqref{rho_Eve},
is Eve's state after the (symmetrized) attack, given that
Alice sent the INFO bitstring $\bi_\sI$ (and given the bitstrings
$\bi_\sT, \bj_\sT, \bb, \bs$, that are ultimately known to Eve).

We now prove the following pivotal result
(based on Lemma~4.5 and Proposition~4.6 of~\cite{BBBMR06},
which are proved in Subsection~4.4 and Appendix~D.2 of~\cite{BBBMR06}),
which is the heart of our security proof---linking the information
obtained by Eve to the disturbance she would cause
in the hypothetical ``inverted-INFO-basis'' protocol:
\begin{theorem}
In the case of a $1$-bit final key (i.e., $m=1$),
for any {\em symmetrized} attack
and any integer $0 \le t \le \frac{n}{2}$,
\begin{equation}\label{boundrho}
\frac{1}{2} \tr |\widehat{\rho}^\sym_0 - \widehat{\rho}^\sym_1| \le
2 \sqrt{\pr^\sym_\invb\left[|\bC_\sI| \ge t
\ \mid\ \bi_\sT, \bj_\sT, \bb, \bs \right]
+ 2^{n[H_2(t/n) - (n-r-1)/n]}},
\end{equation}
where $\bC_\sI$ is the random variable whose value equals
$\bc_\sI \triangleq \bi_\sI \oplus \bj_\sI$;
$\pr^\sym_\invb$ means that the probability is taken over
the hypothetical ``inverted-INFO-basis'' protocol
defined in Subsection~\ref{subsec_invb}
(to which Eve applies the same symmetrized attack
that she applies to the real protocol);
and $H_2(x) \triangleq -x \log_2(x) - (1-x) \log_2(1-x)$.
\end{theorem}
\begin{proof}
This proof is mainly based on Appendix~D.2 of~\cite{BBBMR06},
with some important improvements.

We first define the state $\widetilde{\rho}^\sym_k$,
by giving Eve the purification $\ket{\varphi^\sym_{\bi_\sI}}_\sE$
of $\left( \rho^{\bi_\sI} \right)_\sE^\sym$
(that was defined in Equation~\eqref{phi_Eve_purified})
instead of $\left( \rho^{\bi_\sI} \right)_\sE^\sym$ itself:
\begin{equation}\label{rhotildek}
\widetilde{\rho}^\sym_k \triangleq \frac{1}{2^{n-r-1}}\sum_{P_\sC,
P_\sK, \bi_\sI\,
\big|{\scriptsize\begin{matrix}\bi_\sI P_\sC^\transp = \bxi\\
\bi_\sI \cdot v_{r+1} = k\end{matrix}}} \pr^\sym(P_\sC, P_\sK) \cdot
\ket{\varphi^\sym_{\bi_\sI}}_\sE \bra{\varphi^\sym_{\bi_\sI}}_\sE
\otimes \ket{P_\sC, P_\sK}_\sC \bra{P_\sC, P_\sK}_\sC.
\end{equation}
The state $\widetilde{\rho}^\sym_k$
is a lift-up of $\widehat{\rho}^\sym_k$---namely,
$\widehat{\rho}^\sym_k$
is a partial trace of $\widetilde{\rho}^\sym_k$),
(Note that $\widetilde{\rho}^\sym_k$ was defined in Equation~(4.10)
of~\cite{BBBMR06}, but was denoted there as $\rho_k(v_{r+1}, \bxi)$.)
According to~\cite[Theorem~9.2 and page~407]{NCBook}, and using
the fact that $\widehat{\rho}^\sym_k$ is a partial trace of
$\widetilde{\rho}^\sym_k$, we find out that
\begin{equation}\label{hat_tilde_rho}
\frac{1}{2} \tr |\widehat{\rho}^\sym_0 - \widehat{\rho}^\sym_1| \le
\frac{1}{2} \tr |\widetilde{\rho}^\sym_0 - \widetilde{\rho}^\sym_1|.
\end{equation}

Now, for each $\bxi \in \mathbf{F}_2^r$,
we denote $\bi_\bxi \in \mathbf{F}_2^n$
as an arbitrary, fixed word having the syndrome $\bxi$
(namely, satisfying $\bi_\bxi P_\sC^\transp = \bxi$).
For any $\bi_\sI \in \mathbf{F}_2^n$ that has the same syndrome $\bxi$
(namely, that satisfies $\bi_\sI P_\sC^\transp = \bxi$),
it must hold that:
\begin{equation}
(\bi_\sI - \bi_\bxi) P_\sC^\transp = \bxi - \bxi = \mathbf{0}.
\end{equation}
Thus, because $V_r$ is the row space of $P_\sC$,
we have $(\bi_\sI - \bi_\bxi) \cdot \by = 0$
for any $\by \in V_r$.
Based on this observation, Equation~\eqref{phi_as_eta},
and the fact that
each vector $\bl \in \mathbf{F}_2^n$ has a unique representation
$\bl = \bx \oplus \by$ with $\bx \in V_r^c, \by \in V_r$, we find:
\begin{eqnarray}
\ket{\varphi_{\bi_\sI}^\sym}_\sE &=& \sum_{\bl \in \mathbf{F}_2^n}
(-1)^{\bi_\sI \cdot \bl} \ket{\eta_\bl^\sym}_\sE
= \sum_{\bx \in V_r^c, \by \in V_r} (-1)^{\bi_\sI \cdot (\bx \oplus
\by)} \ket{\eta_{\bx \oplus \by}^\sym}_\sE \nonumber \\
&=& \sum_{\bx \in V_r^c} (-1)^{\bi_\sI \cdot \bx} \sum_{\by \in V_r}
(-1)^{\bi_\sI \cdot \by} \ket{\eta_{\bx \oplus \by}^\sym}_\sE
= \sum_{\bx \in V_r^c} (-1)^{\bi_\sI \cdot \bx} \sum_{\by \in V_r}
(-1)^{\bi_\bxi \cdot \by} \ket{\eta_{\bx \oplus \by}^\sym}_\sE
\nonumber \\
&=& \sum_{\bx \in V_r^c}
(-1)^{\bi_\sI \cdot \bx} \ket{\eta'^{\,\sym}_\bx}_\sE,
\label{eq_phi_eta}
\end{eqnarray}
where we define, for any $\bx \in V_r^c$:
\begin{eqnarray}
\ket{\eta'^{\,\sym}_\bx}_\sE &\triangleq& \sum_{\by \in V_r}
(-1)^{\bi_\bxi \cdot \by} \ket{\eta_{\bx \oplus \by}^\sym}_\sE, \\
d'^{\,\sym}_\bx &\triangleq&
\sqrt{\braket{\eta'^{\,\sym}_\bx}{\eta'^{\,\sym}_\bx}_\sE}, \\
\ket{\widehat{\eta}'^{\,\sym}_\bx}_\sE &\triangleq&
\frac{\ket{\eta'^{\,\sym}_\bx}_\sE}{d'^{\,\sym}_\bx}.
\end{eqnarray}
(These quantities depend on $\bxi$,
but they do not depend on the entire $\bi_\sI$.)
According to Proposition~4.3 of~\cite{BBBMR06},
it holds for symmetrized attacks
that $\braket{\eta^\sym_\bl}{\eta^\sym_{\bl'}} = 0$
for any $\bl \ne \bl'$; therefore,
\begin{equation}\label{d'_from_d}
d'^{\,\sym}_\bx
= \sqrt{\braket{\eta'^{\,\sym}_\bx}{\eta'^{\,\sym}_\bx}_\sE}
= \sqrt{\sum_{\by, \by' \in V_r}
(-1)^{\bi_\bxi \cdot (\by \oplus \by')}
\braket{\eta^\sym_{\bx \oplus \by}}{\eta^\sym_{\bx \oplus \by'}}_\sE}
= \sqrt{\sum_{\by \in V_r} (d_{\bx \oplus \by}^\sym)^2},
\end{equation}
where the last equality is according to Equation~\eqref{phi_d}.

We can now use Equation~\eqref{eq_phi_eta}
to compute $\widetilde{\rho}^\sym_k$: (we notice that
the set $\{\bi_\sI \mid \bi_\sI P_\sC^\transp = \bxi\}$ is actually
the code coset $\bi_\bxi + \sC \triangleq
\{\bi_\bxi \oplus \ba \mid \ba \in \sC\}$
corresponding to the syndrome $\bxi$)
\begin{eqnarray}
\widetilde{\rho}^\sym_k &\triangleq& \frac{1}{2^{n-r-1}}\sum_{P_\sC,
P_\sK, \bi_\sI\,
\big|{\scriptsize\begin{matrix}\bi_\sI P_\sC^\transp = \bxi\\
\bi_\sI \cdot v_{r+1} = k\end{matrix}}} \pr^\sym(P_\sC, P_\sK) \cdot
\ket{\varphi^\sym_{\bi_\sI}}_\sE \bra{\varphi^\sym_{\bi_\sI}}_\sE
\otimes \ket{P_\sC, P_\sK}_\sC \bra{P_\sC, P_\sK}_\sC \nonumber \\
&=& \frac{1}{2^{n-r-1}} \sum_{P_\sC, P_\sK,
\ba \in \sC\ \mid\ (\bi_\bxi \oplus \ba) \cdot v_{r+1} = k}
\pr^\sym(P_\sC, P_\sK) \cdot
\ket{\varphi^\sym_{\bi_\bxi \oplus \ba}}_\sE
\bra{\varphi^\sym_{\bi_\bxi \oplus \ba}}_\sE
\otimes \ket{P_\sC, P_\sK}_\sC \bra{P_\sC, P_\sK}_\sC \nonumber \\
&=& \frac{1}{2^{n-r-1}} \sum_{P_\sC, P_\sK,
\ba \in \sC\ \mid\ (\bi_\bxi \oplus \ba) \cdot v_{r+1} = k}
\pr^\sym(P_\sC, P_\sK) \nonumber \\
&\cdot& \left[ \sum_{\bx, \bx' \in V_r^c}
(-1)^{(\bi_\bxi \oplus \ba) \cdot (\bx \oplus \bx')}
\ket{\eta'^{\,\sym}_\bx}_\sE \bra{\eta'^{\,\sym}_{\bx'}}_\sE \right]
\otimes \ket{P_\sC, P_\sK}_\sC \bra{P_\sC, P_\sK}_\sC.
\end{eqnarray}
Therefore, the difference between $\widetilde{\rho}^\sym_0$ and
$\widetilde{\rho}^\sym_1$ is:
\begin{eqnarray}
\widetilde{\rho}^\sym_0 - \widetilde{\rho}^\sym_1
&=& \sum_{k \in \{0,1\}} (-1)^k \cdot \widetilde{\rho}^\sym_k
\nonumber \\
&=& \frac{1}{2^{n-r-1}} \sum_{P_\sC, P_\sK, \ba \in \sC}
\pr^\sym(P_\sC, P_\sK) \cdot
(-1)^{(\bi_\bxi \oplus \ba) \cdot v_{r+1}} \nonumber \\
&\cdot& \sum_{\bx, \bx' \in V_r^c}
(-1)^{(\bi_\bxi \oplus \ba) \cdot (\bx \oplus \bx')}
\ket{\eta'^{\,\sym}_\bx}_\sE \bra{\eta'^{\,\sym}_{\bx'}}_\sE
\otimes \ket{P_\sC, P_\sK}_\sC \bra{P_\sC, P_\sK}_\sC \nonumber \\
&=& \frac{1}{2^{n-r-1}} \sum_{P_\sC, P_\sK, \bx, \bx' \in V_r^c}
\pr^\sym(P_\sC, P_\sK)  \nonumber \\
&\cdot& \left( \sum_{\ba \in \sC} (-1)^{(\bi_\bxi \oplus \ba) \cdot
(\bx \oplus \bx' \oplus v_{r+1})} \right)
\ket{\eta'^{\,\sym}_\bx}_\sE \bra{\eta'^{\,\sym}_{\bx'}}_\sE
\otimes \ket{P_\sC, P_\sK}_\sC \bra{P_\sC, P_\sK}_\sC \nonumber \\
&=& \frac{1}{2^{n-r-1}} \sum_{P_\sC, P_\sK, \bx, \bx' \in V_r^c}
\pr^\sym(P_\sC, P_\sK) \cdot
(-1)^{\bi_\bxi \cdot (\bx \oplus \bx' \oplus v_{r+1})} \nonumber \\
&\cdot& \left( \sum_{\ba \in \sC} (-1)^{\ba \cdot
(\bx \oplus \bx' \oplus v_{r+1})} \right)
\ket{\eta'^{\,\sym}_\bx}_\sE \bra{\eta'^{\,\sym}_{\bx'}}_\sE
\otimes \ket{P_\sC, P_\sK}_\sC \bra{P_\sC, P_\sK}_\sC \nonumber \\
&=& \frac{1}{2^{n-r-1}} \sum_{P_\sC, P_\sK, \bx, \bx' \in V_r^c}
\pr^\sym(P_\sC, P_\sK) \cdot
(-1)^{\bi_\bxi \cdot (\bx \oplus \bx' \oplus v_{r+1})}
\cdot \left( \sum_{\ba \in \sC} (-1)^{\ba \cdot
(\bx \oplus \bx' \oplus v_{r+1})} \right) \nonumber \\
&\cdot& d'^{\,\sym}_\bx d'^{\,\sym}_{\bx'}
\ket{\widehat{\eta}'^{\,\sym}_\bx}_\sE
\bra{\widehat{\eta}'^{\,\sym}_{\bx'}}_\sE
\otimes \ket{P_\sC, P_\sK}_\sC \bra{P_\sC, P_\sK}_\sC.
\end{eqnarray}
According to Lemma~D.1 of~\cite{BBBMR06},
assuming that $\sC$ is a linear code over $\mathbf{F}_2^n$,
for any $\bl_0 \in \mathbf{F}_2^n \setminus \sC^\perp$
(where $\sC^\perp$ is the dual code of $\sC$)
it holds that $\sum_{\ba \in \sC} (-1)^{\ba \cdot \bl_0} = 0$.

We notice that in our case $\sC^\perp = V_r \triangleq
\Span\{v_1, \ldots, v_r\}$, because the dual code $\sC^\perp$
is the row space of $P_\sC$
(see Subsubsection~\ref{subsubsec_prelims_ecc}).
The sum $\sum_{\ba \in \sC} (-1)^{\ba \cdot
(\bx \oplus \bx' \oplus v_{r+1})}$ above is thus non-zero
if and only if $\bx \oplus \bx' \oplus v_{r+1} \in V_r$,
and because it always holds that
$\bx \oplus \bx' \oplus v_{r+1} \in V_r^c$
(since $\bx, \bx', v_{r+1} \in V_r^c$), the sum is non-zero
if and only if $\bx \oplus \bx' \oplus v_{r+1} = \mathbf{0}$
(because $\mathbf{F}_2^n = V_r^c \oplus V_r$, so according
to the properties of direct sum, $V_r^c \cap V_r = \{\mathbf{0}\}$).
Therefore, the sum is non-zero
if and only if $\bx' = \bx \oplus v_{r+1}$ (in which case the sum is
$\sum_{\ba \in \sC} (-1)^{\ba \cdot
(\bx \oplus \bx' \oplus v_{r+1})}
= \sum_{\ba \in \sC} (-1)^{\ba \cdot \mathbf{0}}
= \sum_{\ba \in \sC} (-1)^0 = |\sC| = 2^k = 2^{n-r}$), so:
\begin{equation}
\widetilde{\rho}^\sym_0 - \widetilde{\rho}^\sym_1 =
2 \sum_{P_\sC, P_\sK, \bx \in V_r^c} \pr^\sym(P_\sC, P_\sK) \cdot
d'^{\,\sym}_\bx d'^{\,\sym}_{\bx \oplus v_{r+1}}
\ket{\widehat{\eta}'^{\,\sym}_\bx}_\sE
\bra{\widehat{\eta}'^{\,\sym}_{\bx \oplus v_{r+1}}}_\sE
\otimes \ket{P_\sC, P_\sK}_\sC \bra{P_\sC, P_\sK}_\sC.
\end{equation}

Now, according to the definitions of
$V_r^c \triangleq \Span \lbrace v_{r+1}, \ldots, v_n \rbrace$ and
$V_{r+1}^c \triangleq \Span \lbrace v_{r+2}, \ldots, v_n \rbrace$ and
the linear independence of $v_1, \ldots, v_n$, we get the fact
\begin{equation}\label{eq_vr_plus_1}
V_r^c = V_{r+1}^c \cup \{\bx \oplus v_{r+1} \mid \bx \in V_{r+1}^c\}
\end{equation}
(a {\em disjoint} union), from which we obtain the following result:
\begin{eqnarray}
\widetilde{\rho}^\sym_0 - \widetilde{\rho}^\sym_1 &=&
2 \sum_{P_\sC, P_\sK, \bx \in V_{r+1}^c} \pr^\sym(P_\sC, P_\sK)
\cdot d'^{\,\sym}_\bx d'^{\,\sym}_{\bx \oplus v_{r+1}} \nonumber \\
&\cdot& \left[ \ket{\widehat{\eta}'^{\,\sym}_\bx}_\sE
\bra{\widehat{\eta}'^{\,\sym}_{\bx \oplus v_{r+1}}}_\sE +
\ket{\widehat{\eta}'^{\,\sym}_{\bx \oplus v_{r+1}}}_\sE
\bra{\widehat{\eta}'^{\,\sym}_\bx}_\sE \right] \nonumber \\
&\otimes& \ket{P_\sC, P_\sK}_\sC \bra{P_\sC, P_\sK}_\sC.
\end{eqnarray}

Using Equation~\eqref{hat_tilde_rho}, we deduce:
\begin{eqnarray}
\frac{1}{2} \tr |\widehat{\rho}^\sym_0 - \widehat{\rho}^\sym_1| &\le&
\frac{1}{2} \tr |\widetilde{\rho}^\sym_0 - \widetilde{\rho}^\sym_1|
\nonumber \\
&=& \tr \left| \sum_{P_\sC, P_\sK, \bx \in V_{r+1}^c}
\pr^\sym(P_\sC, P_\sK) \cdot
d'^{\,\sym}_\bx d'^{\,\sym}_{\bx \oplus v_{r+1}} \right. \nonumber \\
&\cdot& \left. \left[ \ket{\widehat{\eta}'^{\,\sym}_\bx}_\sE
\bra{\widehat{\eta}'^{\,\sym}_{\bx \oplus v_{r+1}}}_\sE +
\ket{\widehat{\eta}'^{\,\sym}_{\bx \oplus v_{r+1}}}_\sE
\bra{\widehat{\eta}'^{\,\sym}_\bx}_\sE \right]
\otimes \ket{P_\sC, P_\sK}_\sC \bra{P_\sC, P_\sK}_\sC
\right| \nonumber \\
&\le& \sum_{P_\sC, P_\sK, \bx \in V_{r+1}^c} \pr^\sym(P_\sC, P_\sK)
\cdot d'^{\,\sym}_\bx d'^{\,\sym}_{\bx \oplus v_{r+1}} \nonumber \\
&\cdot& \tr \left| \left[ \ket{\widehat{\eta}'^{\,\sym}_\bx}_\sE
\bra{\widehat{\eta}'^{\,\sym}_{\bx \oplus v_{r+1}}}_\sE +
\ket{\widehat{\eta}'^{\,\sym}_{\bx \oplus v_{r+1}}}_\sE
\bra{\widehat{\eta}'^{\,\sym}_\bx} \right]_\sE
\otimes \ket{P_\sC, P_\sK}_\sC \bra{P_\sC, P_\sK}_\sC \right|
\nonumber \\
&\le& \sum_{P_\sC, P_\sK, \bx \in V_{r+1}^c} \pr^\sym(P_\sC, P_\sK)
\cdot 2 d'^{\,\sym}_\bx d'^{\,\sym}_{\bx \oplus v_{r+1}}.
\end{eqnarray}
We now use the general inequality $2xy \le \alpha x^2 +
\frac{y^2}{\alpha}$ for any $\alpha > 0$ and $x,y \in \mathbb{R}$
(derived in Appendix~D.2 of~\cite{BBBMR06}) to get:
\begin{eqnarray}
\frac{1}{2} \tr |\widehat{\rho}^\sym_0 - \widehat{\rho}^\sym_1|
&\le& \sum_{P_\sC, P_\sK, \bx \in V_{r+1}^c} \pr^\sym(P_\sC, P_\sK)
\cdot 2 d'^{\,\sym}_\bx d'^{\,\sym}_{\bx \oplus v_{r+1}} \nonumber \\
&=& \sum_{P_\sC, P_\sK, \bx \in V_{r+1}^c\ \mid\ d_\sH(\bx, V_r) \ge t}
\pr^\sym(P_\sC, P_\sK) \cdot
2 d'^{\,\sym}_\bx d'^{\,\sym}_{\bx \oplus v_{r+1}} \nonumber \\
&+& \sum_{P_\sC, P_\sK, \bx \in V_{r+1}^c\ \mid\ d_\sH(\bx, V_r) < t}
\pr^\sym(P_\sC, P_\sK) \cdot
2 d'^{\,\sym}_\bx d'^{\,\sym}_{\bx \oplus v_{r+1}} \nonumber \\
&\le& \sum_{P_\sC, P_\sK, \bx \in V_{r+1}^c\ \mid\ d_\sH(\bx, V_r) \ge
t} \pr^\sym(P_\sC, P_\sK)
\cdot \left[ \alpha (d'^{\,\sym}_{\bx \oplus v_{r+1}})^2
+ \frac{(d'^{\,\sym}_\bx)^2}{\alpha} \right] \nonumber \\
&+& \sum_{P_\sC, P_\sK, \bx \in V_{r+1}^c\ \mid\ d_\sH(\bx, V_r) < t}
\pr^\sym(P_\sC, P_\sK) \cdot \left[ \alpha (d'^{\,\sym}_\bx)^2
+ \frac{(d'^{\,\sym}_{\bx \oplus v_{r+1}})^2}{\alpha} \right]
\nonumber \\
&=& \alpha
\sum_{P_\sC, P_\sK, \bx \in V_{r+1}^c\ \mid\ d_\sH(\bx, V_r) \ge t}
\pr^\sym(P_\sC, P_\sK)
\cdot (d'^{\,\sym}_{\bx \oplus v_{r+1}})^2 \nonumber \\
&+& \alpha
\sum_{P_\sC, P_\sK, \bx \in V_{r+1}^c\ \mid\ d_\sH(\bx, V_r) < t}
\pr^\sym(P_\sC, P_\sK) \cdot (d'^{\,\sym}_\bx)^2 \nonumber \\
&+& \frac{1}{\alpha}
\sum_{P_\sC, P_\sK, \bx \in V_{r+1}^c\ \mid\ d_\sH(\bx, V_r) \ge t}
\pr^\sym(P_\sC, P_\sK) \cdot (d'^{\,\sym}_\bx)^2 \nonumber \\
&+& \frac{1}{\alpha}
\sum_{P_\sC, P_\sK, \bx \in V_{r+1}^c\ \mid\ d_\sH(\bx, V_r) < t}
\pr^\sym(P_\sC, P_\sK)
\cdot (d'^{\,\sym}_{\bx \oplus v_{r+1}})^2.\label{eq_inf_vs_dist_main}
\end{eqnarray}
The first two sums in Equation~\eqref{eq_inf_vs_dist_main}
are bounded as follows:
\begin{eqnarray}
&&\alpha
\sum_{P_\sC, P_\sK, \bx \in V_{r+1}^c\ \mid\ d_\sH(\bx, V_r) \ge t}
\pr^\sym(P_\sC, P_\sK)
\cdot (d'^{\,\sym}_{\bx \oplus v_{r+1}})^2 \nonumber \\
&+& \alpha
\sum_{P_\sC, P_\sK, \bx \in V_{r+1}^c\ \mid\ d_\sH(\bx, V_r) < t}
\pr^\sym(P_\sC, P_\sK)
\cdot (d'^{\,\sym}_\bx)^2 \nonumber \\
&\le& \alpha \sum_{P_\sC, P_\sK} \pr^\sym(P_\sC, P_\sK) \cdot
\left[ \sum_{\bx \in V_{r+1}^c} (d'^{\,\sym}_{\bx \oplus v_{r+1}})^2
+ \sum_{\bx \in V_{r+1}^c} (d'^{\,\sym}_\bx)^2 \right] \nonumber \\
&=& \alpha \sum_{P_\sC, P_\sK, \bx \in V_r^c} \pr^\sym(P_\sC, P_\sK)
\cdot (d'^{\,\sym}_\bx)^2,\label{eq_inf_vs_dist_1}
\end{eqnarray}
using the disjoint union
$V_r^c = V_{r+1}^c \cup \{\bx \oplus v_{r+1} \mid \bx \in V_{r+1}^c\}$
from Equation~\eqref{eq_vr_plus_1}.

In addition, the last two sums in Equation~\eqref{eq_inf_vs_dist_main}
are bounded as follows:
\begin{eqnarray}
&& \frac{1}{\alpha}
\sum_{P_\sC, P_\sK, \bx \in V_{r+1}^c\ \mid\ d_\sH(\bx, V_r) \ge t}
\pr^\sym(P_\sC, P_\sK) \cdot (d'^{\,\sym}_\bx)^2 \nonumber \\
&+& \frac{1}{\alpha}
\sum_{P_\sC, P_\sK, \bx \in V_{r+1}^c\ \mid\ d_\sH(\bx, V_r) < t}
\pr^\sym(P_\sC, P_\sK)
\cdot (d'^{\,\sym}_{\bx \oplus v_{r+1}})^2 \nonumber \\
&\le& \frac{1}{\alpha}
\sum_{P_\sC, P_\sK, \bx \in V_{r+1}^c\ \mid\ d_\sH(\bx, V_r) \ge t}
\pr^\sym(P_\sC, P_\sK) \cdot (d'^{\,\sym}_\bx)^2 \nonumber \\
&+& \frac{1}{\alpha}
\sum_{P_\sC, P_\sK, \bx \in V_{r+1}^c\ \mid\ d_\sH(\bx \oplus
v_{r+1}, V_r) \ge t} \pr^\sym(P_\sC, P_\sK)
\cdot (d'^{\,\sym}_{\bx \oplus v_{r+1}})^2 \nonumber \\
&+& \frac{1}{\alpha}
\sum_{P_\sC, P_\sK, \bx \in V_{r+1}^c\ \mid\ (d_\sH(\bx, V_r) < t)
\wedge (d_\sH(\bx \oplus v_{r+1}, V_r) < t)}
\pr^\sym(P_\sC, P_\sK) \cdot (d'^{\,\sym}_{\bx \oplus v_{r+1}})^2.
\label{eq_inf_vs_dist_2}
\end{eqnarray}
Using, again, the disjoint union
$V_r^c = V_{r+1}^c \cup \{\bx \oplus v_{r+1} \mid \bx \in V_{r+1}^c\}$
from Equation~\eqref{eq_vr_plus_1}
(for the first two sums in Equation~\eqref{eq_inf_vs_dist_2}),
and using the fact that
$\{\bx \oplus v_{r+1} \mid \bx \in V_{r+1}^c\} \subseteq V_r^c$
(for the third sum in Equation~\eqref{eq_inf_vs_dist_2}), we get:
\begin{eqnarray}
&\le& \frac{1}{\alpha}
\sum_{P_\sC, P_\sK, \bx \in V_r^c\ \mid\ d_\sH(\bx, V_r) \ge t}
\pr^\sym(P_\sC, P_\sK) \cdot (d'^{\,\sym}_\bx)^2 \nonumber \\
&+& \frac{1}{\alpha}
\sum_{P_\sC, P_\sK, \bx \in V_r^c\ \mid\ (d_\sH(\bx \oplus v_{r+1},
V_r) < t) \wedge (d_\sH(\bx, V_r) < t)}
\pr^\sym(P_\sC, P_\sK) \cdot (d'^{\,\sym}_\bx)^2.
\label{eq_inf_vs_dist_3}
\end{eqnarray}
Substituting Equations~\eqref{eq_inf_vs_dist_1}
and~\eqref{eq_inf_vs_dist_3} into Equation~\eqref{eq_inf_vs_dist_main},
and using the identity $d'^{\,\sym}_\bx
= \sqrt{\sum_{\by \in V_r} (d_{\bx \oplus \by}^\sym)^2}$
for any $\bx \in V_r^c$ from Equation~\eqref{d'_from_d}, we find:
\begin{eqnarray}
\frac{1}{2} \tr |\widehat{\rho}^\sym_0 - \widehat{\rho}^\sym_1|
&\le& \alpha \sum_{P_\sC, P_\sK, \bx \in V_r^c} \pr^\sym(P_\sC, P_\sK)
\cdot (d'^{\,\sym}_\bx)^2 \nonumber \\
&+& \frac{1}{\alpha}
\sum_{P_\sC, P_\sK, \bx \in V_r^c\ \mid\ d_\sH(\bx, V_r) \ge t}
\pr^\sym(P_\sC, P_\sK) \cdot (d'^{\,\sym}_\bx)^2 \nonumber \\
&+& \frac{1}{\alpha}
\sum_{P_\sC, P_\sK, \bx \in V_r^c\ \mid\ (d_\sH(\bx \oplus v_{r+1},
V_r) < t) \wedge (d_\sH(\bx, V_r) < t)}
\pr^\sym(P_\sC, P_\sK) \cdot (d'^{\,\sym}_\bx)^2 \nonumber \\
&=& \alpha \sum_{P_\sC, P_\sK, \bx \in V_r^c, \by \in V_r}
\pr^\sym(P_\sC, P_\sK) \cdot (d^\sym_{\bx \oplus \by})^2 \nonumber \\
&+& \frac{1}{\alpha} \sum_{P_\sC, P_\sK, \bx \in V_r^c, \by \in
V_r\ \mid\ d_\sH(\bx, V_r) \ge t}
\pr^\sym(P_\sC, P_\sK) \cdot (d^\sym_{\bx \oplus \by})^2 \nonumber \\
&+& \frac{1}{\alpha}
\sum_{P_\sC, P_\sK, \bx \in V_r^c, \by \in
V_r\ \mid\ (d_\sH(\bx \oplus v_{r+1}, V_r) < t) \wedge
(d_\sH(\bx, V_r) < t)} \pr^\sym(P_\sC, P_\sK) \nonumber \\
&\cdot& (d^\sym_{\bx \oplus \by})^2.\label{eq_inf_vs_dist_final}
\end{eqnarray}
We can now use the fact that
$\mathbf{F}_2^n = V_r^c \oplus V_r$ (namely,
each vector $\bl \in \mathbf{F}_2^n$ has a unique representation
$\bl = \bx \oplus \by$ with $\bx \in V_r^c, \by \in V_r$)
and the fact that $d^\sym_\bl$
is independent of $P_\sC,P_\sK$\footnote{This is correct
because $d^\sym_\bl$ depends on $\ket{\eta^\sym_\bl}_\sE,
\{\ket{\varphi^\sym_{\bi_\sI}}_\sE\}_{\bi_\sI \in \mathbf{F}_2^n},
\{\ket{E^\sym_{\bi_\sI,\bj_\sI}}_{\bb, \bs}\}_{\bi_\sI,
\bj_\sI \in \mathbf{F}_2^n}$,
which are all independent of $P_\sC,P_\sK$, because Alice
chooses $P_\sC,P_\sK$ uniformly at random,
and she does not use them and keeps them secret
until Eve's attack is over.}
for computing the three sums in Equation~\eqref{eq_inf_vs_dist_final}:
\begin{enumerate}
\item For the first sum in Equation~\eqref{eq_inf_vs_dist_final},
we find:
\begin{eqnarray}
&&\alpha \sum_{P_\sC, P_\sK, \bx \in V_r^c, \by \in V_r}
\pr^\sym(P_\sC, P_\sK) \cdot (d^\sym_{\bx \oplus \by})^2 \nonumber \\
&=& \alpha \sum_{P_\sC, P_\sK, \bl \in \mathbf{F}_2^n}
\pr^\sym(P_\sC, P_\sK) \cdot (d^\sym_\bl)^2 \nonumber \\
&=& \alpha \sum_{\bl \in \mathbf{F}_2^n} \left[ \sum_{P_\sC, P_\sK}
\pr^\sym(P_\sC, P_\sK) \right] \cdot (d^\sym_\bl)^2
= \alpha \sum_{\bl \in \mathbf{F}_2^n} (d^\sym_\bl)^2.
\label{eq_inf_vs_dist_final_1a}
\end{eqnarray}
According to Proposition~4.3 of~\cite{BBBMR06},
it holds for symmetrized attacks
that $\braket{\eta^\sym_\bl}{\eta^\sym_{\bl'}} = 0$
for any $\bl \ne \bl'$; therefore,
according to Equations~\eqref{norm_eta} and~\eqref{phi_as_eta},
the normalized state $\ket{\varphi^\sym_{\bi_\sI}}_\sE$ satisfies:
\begin{eqnarray}
1 &=& \braket{\varphi^\sym_{\bi_\sI}}{\varphi^\sym_{\bi_\sI}}_\sE
= \sum_{\bl, \bl' \in \mathbf{F}_2^n}
(-1)^{\bi_\sI \cdot (\bl \oplus \bl')}
\braket{\eta^\sym_\bl}{\eta^\sym_{\bl'}}_\sE \nonumber \\
&=& \sum_{\bl \in \mathbf{F}_2^n}
\braket{\eta^\sym_\bl}{\eta^\sym_\bl}_\sE
= \sum_{\bl \in \mathbf{F}_2^n} (d^\sym_\bl)^2.
\label{eq_inf_vs_dist_final_1b}
\end{eqnarray}
Substituting Equation~\eqref{eq_inf_vs_dist_final_1b} into
Equation~\eqref{eq_inf_vs_dist_final_1a},
we find the first sum in Equation~\eqref{eq_inf_vs_dist_final} to be:
\begin{equation}\label{eq_inf_vs_dist_final_1}
\alpha \sum_{P_\sC, P_\sK, \bx \in V_r^c, \by \in V_r}
\pr^\sym(P_\sC, P_\sK) \cdot (d^\sym_{\bx \oplus \by})^2
= \alpha \sum_{\bl \in \mathbf{F}_2^n} (d^\sym_\bl)^2
= \alpha.
\end{equation}
\item For the second sum in Equation~\eqref{eq_inf_vs_dist_final},
we find:
\begin{eqnarray}
&&\frac{1}{\alpha} \sum_{P_\sC, P_\sK, \bx \in V_r^c, \by \in
V_r\ \mid\ d_\sH(\bx, V_r) \ge t}
\pr^\sym(P_\sC, P_\sK) \cdot (d^\sym_{\bx \oplus \by})^2 \nonumber \\
&\le& \frac{1}{\alpha} \sum_{P_\sC, P_\sK, \bx \in V_r^c, \by \in
V_r\ \mid\ d_\sH(\bx, \by) \ge t}
\pr^\sym(P_\sC, P_\sK) \cdot (d^\sym_{\bx \oplus \by})^2 \nonumber \\
&=& \frac{1}{\alpha} \sum_{P_\sC, P_\sK,
\bl \in \mathbf{F}_2^n\ :\ |\bl| \ge t}
\pr^\sym(P_\sC, P_\sK) \cdot (d^\sym_\bl)^2 \nonumber \\
&=& \frac{1}{\alpha} \sum_{|\bl| \ge t}
\left[ \sum_{P_\sC, P_\sK} \pr^\sym(P_\sC, P_\sK) \right]
\cdot (d^\sym_\bl)^2 \nonumber \\
&=& \frac{1}{\alpha} \sum_{|\bl| \ge t} (d^\sym_\bl)^2
= \frac{1}{\alpha} \sum_{|\bc_\sI| \ge t}
(d^\sym_{\bc_\sI})^2.\label{eq_inf_vs_dist_final_2}
\end{eqnarray}
\item For the third sum in Equation~\eqref{eq_inf_vs_dist_final}, we
find: (based on the fact that for any $\bx \in V_r^c, \by \in V_r$ and
the corresponding $\bl \triangleq \bx \oplus \by \in \mathbf{F}_2^n$,
it holds that $d_\sH(\bx, V_r) = d_\sH(\bl, V_r)$ and
$d_\sH(\bx \oplus v_{r+1}, V_r) = d_\sH(\bl \oplus v_{r+1}, V_r)$)
\begin{eqnarray}
&&\frac{1}{\alpha} \sum_{P_\sC, P_\sK, \bx \in V_r^c, \by \in
V_r\ \mid\ (d_\sH(\bx \oplus v_{r+1}, V_r) < t) \wedge
(d_\sH(\bx, V_r) < t)} \pr^\sym(P_\sC, P_\sK)
\cdot (d^\sym_{\bx \oplus \by})^2 \nonumber \\
&=& \frac{1}{\alpha} \sum_{P_\sC, P_\sK, \bl \in
\mathbf{F}_2^n\ \mid\ (d_\sH(\bl \oplus v_{r+1}, V_r) < t) \wedge
(d_\sH(\bl, V_r) < t)}
\pr^\sym(P_\sC, P_\sK) \cdot (d^\sym_\bl)^2 \nonumber \\
&=& \frac{1}{\alpha} \sum_{\bl \in \mathbf{F}_2^n}
\left[ \sum_{P_\sC, P_\sK\ \mid\ (d_\sH(\bl \oplus v_{r+1}, V_r) < t)
\wedge (d_\sH(\bl, V_r) < t)} \pr^\sym(P_\sC, P_\sK) \right]
\cdot (d^\sym_\bl)^2 \nonumber \\
&=& \frac{1}{\alpha} \sum_{\bl \in \mathbf{F}_2^n}
\pr^\sym\left[(d_\sH(\bl \oplus v_{r+1}, V_r) < t) \wedge
(d_\sH(\bl, V_r) < t)\ \mid \ \bl\right]
\cdot (d^\sym_\bl)^2 \nonumber \\
&\le& \frac{1}{\alpha} \sum_{\bl \in \mathbf{F}_2^n}
\pr^\sym\left[d_\sH(\bl \oplus v_{r+1}, V_r) < t\ \mid \ \bl\right]
\cdot (d^\sym_\bl)^2 \nonumber \\
&=& \frac{1}{\alpha} \sum_{\bl \in \mathbf{F}_2^n}
\pr^\sym\left[\exists \by \in V_r :
|\bl \oplus v_{r+1} \oplus \by| < t\ \mid \ \bl\right]
\cdot (d^\sym_\bl)^2,\label{eq_inf_vs_dist_final_3a}
\end{eqnarray}
where the probability is actually taken over the randomly-chosen
matrices $P_\sC \in \mathbf{F}_2^{r \times n},
P_\sK \triangleq \left(v_{r+1}\right) \in \mathbf{F}_2^{1 \times n}$,
which constitute together an $(r+1) \times n$ generator matrix
for an $(r+1)$-dimensional random error-correcting code
over $\mathbf{F}_2^n$. This random code, in fact,
is $\sC^\perp_\ext \triangleq \Span\{\sC^\perp, v_{r+1}\}$,
where $\sC^\perp = V_r \triangleq
\Span\{v_1, \ldots, v_r\}$ is the $r$-dimensional dual code
(see Subsubsection~\ref{subsubsec_prelims_ecc})
of the error-correcting code $\sC$ belonging to the parity-check
matrix $P_\sC$, and $v_{r+1}$ is the only row of the matrix $P_\sK$.
Therefore:
\begin{equation}
\sC^\perp_\ext \triangleq \Span\{\sC^\perp, v_{r+1}\}
= \Span\{V_r, v_{r+1}\} = \Span\{v_1, \ldots, v_r, v_{r+1}\}.
\end{equation}
Now, applying Proposition~\ref{prop_ecc_random}
to this $(r+1)$-dimensional random code $\sC^\perp_\ext = V_{r+1}$
(the parameters for the
Proposition are our $n$, $k = r + 1$, our $0 \le t \le \frac{n}{2}$,
and our $\bl \in \mathbf{F}_2^n$),
we find that the probability that there exists a word
$\bz \triangleq \bl \oplus v_{r+1} \oplus \by \in \bl + \sC^\perp_\ext$
which satisfies $\bz \ne \bl$
(because $\{V_r, v_{r+1}\}$ are linearly independent,
so $v_{r+1} \in V_r^c \setminus \{\mathbf{0}\}$)
and $|\bz| \triangleq |\bl \oplus v_{r+1} \oplus \by| \le t$
is bounded by:
\begin{eqnarray}
\pr^\sym\left[\exists \by \in V_r :
|\bl \oplus v_{r+1} \oplus \by| < t\ \mid \ \bl\right] &\le&
\pr^\sym \left[ \exists \bz \in \bl + \sC^\perp_\ext
: \bz \ne \bl, |\bz| \le t\ \mid \ \bl \right] \nonumber \\
&\le& 2^{n[H_2(t/n) - (n-r-1)/n]}.\label{eq_inf_vs_dist_final_3b}
\end{eqnarray}
Substituting Equations~\eqref{eq_inf_vs_dist_final_3b}
and~\eqref{eq_inf_vs_dist_final_1b}
into Equation~\eqref{eq_inf_vs_dist_final_3a},
we find the third sum in Equation~\eqref{eq_inf_vs_dist_final} to be:
\begin{eqnarray}
&&\frac{1}{\alpha} \sum_{P_\sC, P_\sK, \bx \in V_r^c, \by \in
V_r\ \mid\ (d_\sH(\bx \oplus v_{r+1}, V_r) < t) \wedge
(d_\sH(\bx, V_r) < t)} \pr^\sym(P_\sC, P_\sK)
\cdot (d^\sym_{\bx \oplus \by})^2 \nonumber \\
&\le& \frac{1}{\alpha} \sum_{\bl \in \mathbf{F}_2^n}
\pr^\sym\left[\exists \by \in V_r :
|\bl \oplus v_{r+1} \oplus \by| < t\ \mid \ \bl\right]
\cdot (d^\sym_\bl)^2 \nonumber \\
&\le& \frac{1}{\alpha} \sum_{\bl \in \mathbf{F}_2^n}
2^{n[H_2(t/n) - (n-r-1)/n]} \cdot (d^\sym_\bl)^2 \nonumber \\
&=& \frac{1}{\alpha} \cdot 2^{n[H_2(t/n) - (n-r-1)/n]}
\sum_{\bl \in \mathbf{F}_2^n} (d^\sym_\bl)^2
= \frac{1}{\alpha} \cdot 2^{n[H_2(t/n) - (n-r-1)/n]}.
\label{eq_inf_vs_dist_final_3}
\end{eqnarray}
\end{enumerate}
Substituting all three Equations~\eqref{eq_inf_vs_dist_final_1},
\eqref{eq_inf_vs_dist_final_2}, and~\eqref{eq_inf_vs_dist_final_3}
into Equation~\eqref{eq_inf_vs_dist_final}, we find:
\begin{equation}
\frac{1}{2} \tr |\widehat{\rho}^\sym_0 - \widehat{\rho}^\sym_1|
\le \alpha
+ \frac{1}{\alpha} \cdot \sum_{|\bc_\sI| \ge t} (d^\sym_{\bc_\sI})^2
+ \frac{1}{\alpha} \cdot 2^{n[H_2(t/n) - (n-r-1)/n]}.
\end{equation}
Choosing $\alpha = \sqrt{\sum_{|\bc_\sI| \ge t}
(d^\sym_{\bc_\sI})^2 + 2^{n[H_2(t/n) - (n-r-1)/n]}}$, we obtain:
\begin{equation}
\frac{1}{2} \tr |\widehat{\rho}^\sym_0 - \widehat{\rho}^\sym_1| \le
2 \sqrt{\sum_{|\bc_\sI| \ge t}
(d^\sym_{\bc_\sI})^2 + 2^{n[H_2(t/n) - (n-r-1)/n]}}.
\end{equation}
According to Lemma~4.4 of~\cite{BBBMR06}, for any
$\bc_\sI \in \mathbf{F}_2^n$ it holds that:
\begin{equation}
(d^\sym_{\bc_\sI})^2 = \pr^\sym_\invb\left[\bC_\sI =
\bc_\sI\ \mid\ \bi_\sT, \bj_\sT, \bb, \bs \right].
\end{equation}
Therefore, we finally get:
\begin{equation}\label{boundrho2}
\frac{1}{2} \tr |\widehat{\rho}^\sym_0 - \widehat{\rho}^\sym_1|
\le 2 \sqrt{\pr^\sym_\invb\left[|\bC_\sI| \ge t
\ \mid\ \bi_\sT, \bj_\sT, \bb, \bs \right]
+ 2^{n[H_2(t/n) - (n-r-1)/n]}}.
\end{equation}
\end{proof}

\subsection{Bounding the Differences Between Eve's States}
So far, we have discussed a $1$-bit key.
We will now discuss a general $m$-bit key $\bk$.
We define $\widehat{\rho}^\sym_\bk$ to be the state of Eve
corresponding to the final key $\bk$, given that she knows $\bxi$:
\begin{equation}\label{rhohatk_multi}
\widehat{\rho}^\sym_\bk \triangleq \frac{1}{2^{n-r-m}}\sum_{P_\sC,
P_\sK,\bi_\sI\,\big|{\scriptsize\begin{matrix}\bi_\sI P_\sC^\transp =
\bxi\\ \bi_\sI P_\sK^\transp = \bk\end{matrix}}}
\pr^\sym(P_\sC, P_\sK) \cdot
\left( \rho^{\bi_\sI} \right)_\sE^\sym
\otimes \ket{P_\sC, P_\sK}_\sC \bra{P_\sC, P_\sK}_\sC.
\end{equation}
We note (for use in Subsection~\ref{sec_full_composability})
that if we substitute $(\rho^{\bi_\sI})^\sym$
from Equation~\eqref{rho_Eve}, we get
\begin{eqnarray}
\widehat{\rho}^\sym_\bk &=& \frac{1}{2^{n-r-m}}\sum_{P_\sC,P_\sK,
\bi_\sI,\bj_\sI\,
\big|{\scriptsize\begin{matrix}\bi_\sI P_\sC^\transp =
\bxi\\ \bi_\sI P_\sK^\transp = \bk\end{matrix}}}
\pr^\sym(P_\sC, P_\sK) \cdot
\pr^\sym(\bj_\sI \mid \bi_\sI, \bi_\sT, \bj_\sT, \bb, \bs) \nonumber \\
&\cdot& \left( \rho_{\bi_\sI,\bj_\sI}^{\bb, \bs} \right)_\sE^\sym
\otimes \ket{P_\sC, P_\sK}_\sC \bra{P_\sC, P_\sK}_\sC.
\label{rhohatk_multi2}
\end{eqnarray}

\begin{proposition}\label{probbound}
For any two keys $\bk,\bk'$ of $m$ bits,
any symmetrized attack,
and any integer $0 \le t \le \frac{n}{2}$,
\begin{equation}
\frac{1}{2} \tr |\widehat{\rho}^\sym_{\bk} -
\widehat{\rho}^\sym_{\bk'}|
\le 2m \sqrt{\pr^\sym_\invb\left[|\bC_\sI| \ge t
\ \mid\ \bi_\sT, \bj_\sT, \bb, \bs \right]
+ 2^{n[H_2(t/n) - (n-r-m)/n]}},\label{eqlemma}
\end{equation}
where $\bC_\sI$ is the random variable whose value equals
$\bc_\sI \triangleq \bi_\sI \oplus \bj_\sI$,
and $H_2(x) \triangleq -x \log_2(x) - (1-x) \log_2(1-x)$.
\end{proposition}
\begin{proof}
We define the key $\bk_j$, for $0 \le j \le m$, to consist of the first
$j$ bits of $\bk'$ and the last $m-j$ bits of $\bk$. This means that
$\bk_0 = \bk$, $\bk_m = \bk'$, and $\bk_{j-1}$ differs from $\bk_j$
at most on a single bit (the $j$-th bit).

First, we find a bound on $\frac{1}{2} \tr
|\widehat{\rho}^\sym_{\bk_{j-1}} - \widehat{\rho}^\sym_{\bk_j}|$:
since $\bk_{j-1}$ differs from $\bk_j$
at most on a single bit (the $j$-th bit, given by the formula
$\bi_\sI \cdot v_{r+j}$), we can use the same proof
that gave us Equation~\eqref{boundrho2}, attaching the other
(identical) $m - 1$ key bits to $\bxi$ of the original proof
(for the same parameter $t$); and we find out that
\begin{equation}
\frac{1}{2} \tr |\widehat{\rho}^\sym_{\bk_{j-1}} -
\widehat{\rho}^\sym_{\bk_j}| \le
2 \sqrt{\pr^\sym_\invb\left[|\bC_\sI| \ge t
\ \mid\ \bi_\sT, \bj_\sT, \bb, \bs \right]
+ 2^{n[H_2(t/n) - (n-r-m)/n]}}.\label{StatesDiff1}
\end{equation}
Now we use the triangle inequality for norms to find
\begin{align}
\frac{1}{2} \tr |\widehat{\rho}^\sym_{\bk} -
\widehat{\rho}^\sym_{\bk'}|
&= \frac{1}{2} \tr |\widehat{\rho}^\sym_{\bk_0} -
\widehat{\rho}^\sym_{\bk_m}|
\le \sum_{j = 1}^m \frac{1}{2} \tr |\widehat{\rho}^\sym_{\bk_{j-1}}
- \widehat{\rho}^\sym_{\bk_j}| \nonumber \\
&\le 2m \sqrt{\pr^\sym_\invb\left[|\bC_\sI| \ge t
\ \mid\ \bi_\sT, \bj_\sT, \bb, \bs \right]
+ 2^{n[H_2(t/n) - (n-r-m)/n]}}.\label{StatesDiff3}
\end{align}
\end{proof}

We would now like to bound the expected value
(namely, the average value) of the trace distance between
two states of Eve corresponding to two final keys.
However, we should take into account that if the test fails,
no final key is generated,
in which case we define the distance to be $0$.
We thus define the random variable $\Delta^\sym_\Eve(\bk,\bk')$
for any two final keys $\bk,\bk'$:
\begin{equation}
\Delta^\sym_\Eve(\bk, \bk' \mid \bi_\sT, \bj_\sT, \bb, \bs, \bxi)
\triangleq \begin{cases}
\frac{1}{2} \tr |\widehat{\rho}^\sym_{\bk} -
\widehat{\rho}^\sym_{\bk'}|
& \text{if $T(\bi_\sT \oplus \bj_\sT, \bb_\sT, \bs) = 1$} \\
0 & \text{otherwise} \end{cases}.\label{defipa}
\end{equation}

We need to bound the expected value
$\langle \Delta^\sym_\Eve(\bk, \bk') \rangle$,
that is given by:
\begin{align}
\langle \Delta^\sym_\Eve(\bk, \bk') \rangle
= \sum_{\scriptsize\begin{matrix}\bi_\sT, \bj_\sT
\in \mathbf{F}_2^{N-n},
\\ \bb \in B, \bs \in S_\bb,
\bxi \in \mathbf{F}_2^r\end{matrix}} \left[\right.&
\Delta^\sym_\Eve(\bk, \bk' \mid \bi_\sT, \bj_\sT, \bb, \bs, \bxi)
\nonumber \\
&\cdot \left.\pr^\sym(\bi_\sT, \bj_\sT, \bb, \bs, \bxi)\right].
\label{eqipa}
\end{align}
(In Subsection~\ref{sec_full_composability} we prove that this
expected value is indeed the quantity we need to bound
for proving fully composable security.)

\begin{theorem}\label{thmsecurity}
For any two final keys $\bk, \bk'$
and any integer $0 \le t \le \frac{n}{2}$,
\begin{equation}
\langle \Delta^\sym_\Eve(\bk, \bk')\rangle \le 2m
\sqrt{\pr_\invb\left[ \left( \frac{|\bC_\sI|}{n} \ge \frac{t}{n}
\right) \wedge \left( \bT = 1 \right) \right]
+ 2^{n[H_2(t/n) - (n-r-m)/n]}},\label{boundotherbasis}
\end{equation}
where $\frac{|\bC_\sI|}{n}$ is a random variable
whose value is the error rate on the INFO bits;
$\bT$ is a random variable whose value is $1$ if the test passes
and $0$ otherwise;
and $H_2(x) \triangleq -x \log_2(x) - (1-x) \log_2(1-x)$.
We note that the protocol considered for
the probability in the right-hand-side is
the hypothetical ``inverted-INFO-basis'' protocol
defined in Subsection~\ref{subsec_invb},
in which Alice and Bob use the basis string
$\bb^0 \triangleq \bb \oplus \bs$ instead of $\bb$;
we also note that the probability in the right-hand-side
is the probability for the original (non-symmetrized) attack.
\end{theorem}
\begin{proof}
We use convexity of $x^2$---namely,
the fact that for all $\{p_i\}_i$
satisfying $p_i \ge 0$ and $\sum_i p_i = 1$, it holds that
$\left( \sum_i p_i x_i \right)^2 \le \sum_i p_i x_i^2$.
We also use the fact that
\begin{eqnarray}\label{symm_invb_probs}
&&\pr^\sym_\invb(\bi_\sT, \bj_\sT, \bb, \bs) \nonumber \\
&=& \pr^\sym_\invb(\bi_\sT, \bb, \bs) \cdot
\pr^\sym_\invb(\bj_\sT \mid \bi_\sT, \bb, \bs) \nonumber \\
&=& \pr^\sym(\bi_\sT, \bb, \bs) \cdot
\pr^\sym(\bj_\sT \mid \bi_\sT, \bb, \bs) \nonumber \\
&=& \pr^\sym(\bi_\sT, \bj_\sT, \bb, \bs),
\end{eqnarray}
which is correct because $\bi_\sT, \bb, \bs$ are all {\em chosen}
in the same way in both the hypothetical ``inverted-INFO-basis''
protocol and the real protocol (even though
different basis strings are {\em used} in these protocols),
and because according to the third property of the symmetrized attack
(Equation~\eqref{symm_prob_test_err_invb}),
$\pr^\sym_\invb(\bj_\sT \mid \bi_\sT, \bb, \bs) =
\pr^\sym(\bj_\sT \mid \bi_\sT, \bb, \bs)$.

In addition, we use the result
\begin{eqnarray}\label{symm_invb_sec_c}
&&\pr^\sym_\invb\left[ \left( |\bC_\sI | \ge t \right)
\wedge \left( \bT = 1 \right) \mid \bb, \bs \right]
\nonumber \\
&=& \pr_\invb\left[ \left( |\bC_\sI | \ge t \right)
\wedge \left( \bT = 1 \right) \mid \bb, \bs \right],
\end{eqnarray}
which is correct because according to the second property
of the symmetrized attack (Equation~\eqref{symm_probs_c_invb}),
$\pr^\sym_\invb(\bc_\sI, \bc_\sT \mid \bb, \bs) =
\pr_\invb(\bc_\sI, \bc_\sT \mid \bb, \bs)$,
and because the random variable $\bT$ depends only
on the random variable $\bC_\sT$ and the parameters $\bb_\sT, \bs$.

We also use the result
\begin{equation}\label{symm_invb_b_s}
\pr^\sym_\invb(\bb, \bs) = \pr_\invb(\bb, \bs),
\end{equation}
which is correct because Alice's random choice of
$\bb, \bs$ is independent of Eve's attack.

We find out that:
\begin{align}
&\langle \Delta^\sym_\Eve(\bk, \bk')\rangle^2 &
\nonumber \\
&= \left[ \sum_{\bi_\sT, \bj_\sT, \bb, \bs, \bxi}
\Delta^\sym_\Eve(\bk, \bk' \mid \bi_\sT, \bj_\sT, \bb, \bs, \bxi)
\cdot \pr^\sym(\bi_\sT, \bj_\sT, \bb, \bs, \bxi)\right]^2
&\text{(by \eqref{eqipa})} \nonumber \\
&\le \sum_{\bi_\sT, \bj_\sT, \bb, \bs, \bxi} \left[
\Delta^\sym_\Eve(\bk, \bk' \mid \bi_\sT, \bj_\sT, \bb, \bs, \bxi)
\right]^2
\cdot \pr^\sym(\bi_\sT, \bj_\sT, \bb, \bs, \bxi)
\qquad &\text{\llap{(by convexity of $x^2$)}} \nonumber \\
&= \sum_{\bi_\sT, \bj_\sT, \bb, \bs, \bxi \mid \bT = 1}
\left( \frac{1}{2} \tr |\widehat{\rho}^\sym_{\bk} -
\widehat{\rho}^\sym_{\bk'}| \right)^2
\cdot \pr^\sym(\bi_\sT, \bj_\sT, \bb, \bs, \bxi)
&\text{(by \eqref{defipa})} \nonumber \\
&\le 4m^2 \cdot \sum_{\bi_\sT, \bj_\sT, \bb, \bs, \bxi \mid \bT = 1}
\left( \pr^\sym_\invb\left[|\bC_\sI| \ge t \mid
\bi_\sT, \bj_\sT, \bb, \bs \right] +
2^{n[H_2(t/n) - (n-r-m)/n]} \right) & \nonumber \\
&~\cdot~~\pr^\sym(\bi_\sT, \bj_\sT, \bb, \bs, \bxi)
&\text{(by \eqref{eqlemma})} \nonumber \\
&= 4m^2 \cdot \sum_{\bi_\sT, \bj_\sT, \bb, \bs \mid \bT = 1}
\left( \pr^\sym_\invb\left[ |\bC_\sI| \ge t
\mid \bi_\sT, \bj_\sT, \bb, \bs \right]
+ 2^{n[H_2(t/n) - (n-r-m)/n]} \right) & \nonumber \\
&~\cdot~~\pr^\sym(\bi_\sT, \bj_\sT, \bb, \bs) & \nonumber \\
&= 4m^2 \cdot \sum_{\bi_\sT, \bj_\sT, \bb, \bs \mid \bT = 1}
\left( \pr^\sym_\invb\left[ |\bC_\sI| \ge t
\mid \bi_\sT, \bj_\sT, \bb, \bs \right] \right. & \nonumber \\
&\left.~\cdot~~ \pr^\sym(\bi_\sT, \bj_\sT, \bb, \bs)
+ 2^{n[H_2(t/n) - (n-r-m)/n]}
\cdot \pr^\sym(\bi_\sT, \bj_\sT, \bb, \bs) \right) & \nonumber \\
&= 4m^2 \cdot \left( \sum_{\bi_\sT, \bj_\sT, \bb, \bs}
\pr^\sym_\invb\left[ \left(|\bC_\sI | \ge t \right) \wedge
\left( \bT = 1 \right)
\mid \bi_\sT, \bj_\sT, \bb, \bs \right] \right. & \nonumber \\
&\left.~\cdot~~\pr^\sym(\bi_\sT, \bj_\sT, \bb, \bs)
+ 2^{n[H_2(t/n) - (n-r-m)/n]}
\cdot \pr^\sym \left( \bT = 1 \right) \right) & \nonumber \\
&\le 4m^2 \cdot \left( \sum_{\bi_\sT, \bj_\sT, \bb, \bs}
\pr^\sym_\invb\left[ \left(|\bC_\sI | \ge t \right) \wedge
\left( \bT = 1 \right)
\mid \bi_\sT, \bj_\sT, \bb, \bs \right] \right. & \nonumber \\
&\left.~\cdot~~\pr^\sym(\bi_\sT, \bj_\sT, \bb, \bs)
+ 2^{n[H_2(t/n) - (n-r-m)/n]} \right) & \nonumber \\
&= 4m^2 \cdot \left( \sum_{\bi_\sT, \bj_\sT, \bb, \bs}
\pr^\sym_\invb\left[ \left(|\bC_\sI | \ge t \right) \wedge
\left( \bT = 1 \right) \mid \bi_\sT, \bj_\sT, \bb, \bs \right] \right.
& \nonumber \\
&\left.~\cdot~~\pr^\sym_\invb(\bi_\sT, \bj_\sT, \bb, \bs)
+ 2^{n[H_2(t/n) - (n-r-m)/n]} \right)
&\text{(by \eqref{symm_invb_probs})} \nonumber \\
&= 4m^2 \cdot \left( \sum_{\bb, \bs}
\pr^\sym_\invb\left[ \left(|\bC_\sI | \ge t \right) \wedge
\left( \bT = 1 \right) \mid \bb, \bs \right] \right. & \nonumber \\
&\left.~\cdot~~\pr^\sym_\invb(\bb, \bs)
+ 2^{n[H_2(t/n) - (n-r-m)/n]} \right) & \nonumber \\
&= 4m^2 \cdot \left( \sum_{\bb, \bs}
\pr_\invb\left[ \left(|\bC_\sI | \ge t \right) \wedge
\left( \bT = 1 \right) \mid \bb, \bs \right] \right. & \nonumber \\
&\left.~\cdot~~\pr_\invb(\bb, \bs)
+ 2^{n[H_2(t/n) - (n-r-m)/n]} \right) &\text{(by
\eqref{symm_invb_sec_c}--\eqref{symm_invb_b_s})} \nonumber \\
&= 4m^2 \cdot \left( \pr_\invb\left[ \left(|\bC_\sI | \ge t \right)
\wedge \left( \bT = 1 \right) \right]
+ 2^{n[H_2(t/n) - (n-r-m)/n]} \right). &
\end{align}
\end{proof}

\subsection{\label{sec_full_composability}Bound for Fully Composable
Security}
We now prove a crucial part of the claim that generalized BB84
protocols satisfy the definition of composable security
for a QKD protocol:
we derive an upper bound for the expression
$\frac{1}{2} \tr \left| \rho_{\mathrm{ABE}}
- \rho_\sU \otimes \rho_\sE \right|$,
where $\rho_{\mathrm{ABE}}$ is the actual joint state
of Alice, Bob, and Eve at the end of the protocol;
$\rho_\sU$ is an ideal (random, secret, and shared)
key distributed to Alice and Bob;
and $\rho_\sE$ is the partial trace of $\rho_{\mathrm{ABE}}$
over the system $\mathrm{AB}$.
In other words, we upper-bound the trace distance between the system
after the real QKD protocol and the system after an ideal key
distribution protocol (which first performs the real QKD protocol
and then magically distributes to Alice and Bob a random, secret,
and shared key).

The states $\rho_{\mathrm{ABE}}$ and $\rho_\sU$ are
\begin{eqnarray}
\rho_{\mathrm{ABE}} &=& \sum_{\bi, \bj, \bb, \bs, P_\sC, P_\sK
\mid \bT = 1} \pr\left( \bi, \bj, \bb, \bs, P_\sC, P_\sK \right) \cdot
\ket{\bk}_\sA \bra{\bk}_\sA \otimes
\ket{\bk^\sB}_\sB \bra{\bk^\sB}_\sB \nonumber \\
&\otimes& \left( \rho_{\bi_\sI,\bj_\sI}^{\bb, \bs} \right)_\sE \otimes
\ket{\bi_\sT, \bj_\sT, \bb, \bs, \bxi, P_\sC, P_\sK}_\sC
\bra{\bi_\sT, \bj_\sT, \bb, \bs, \bxi,
P_\sC, P_\sK}_\sC,\label{eq_rho_ABE_def} \\
\rho_\sU &=& \frac{1}{2^m} \sum_{\bk} \ket{\bk}_\sA \bra{\bk}_\sA
\otimes \ket{\bk}_\sB \bra{\bk}_\sB,\label{eq_rho_U_def}
\end{eqnarray}
where $\left( \rho_{\bi_\sI,\bj_\sI}^{\bb, \bs} \right)_\sE$
was defined in Equation~\eqref{rho_Eve_ij}
to be Eve's normalized quantum state if Alice chooses
$\bi_\sI, \bi_\sT, \bb, \bs$
and Bob measures $\bj_\sI, \bj_\sT$.
All the other states in
Equations~\eqref{eq_rho_ABE_def}--\eqref{eq_rho_U_def}
actually represent classical information:
subsystems $\sA$ and $\sB$ represent the final keys held by Alice
($\bk \triangleq \bi_\sI P_\sK^\transp$) and Bob (his key $\bk^\sB$ is
computed from $\bj_\sI$, $\bxi \triangleq \bi_\sI P_\sC^\transp$,
and $P_\sC, P_\sK$), respectively,
and subsystem $\sC$ represents the information published
in the unjammable classical channel during the protocol
(this information is known to Alice, Bob, and Eve)---namely,
$\bi_\sT, \bj_\sT$ (all the TEST bits),
$\bb$ (the basis string),
$\bs$ (the string representing the partition into INFO and TEST bits),
$\bxi \triangleq \bi_\sI P_\sC^\transp$ (the syndrome),
and $P_\sC, P_\sK$
(the error correction and privacy amplification matrices).

We note that in the definition of $\rho_{\mathrm{ABE}}$,
we sum only over the events in which the test is {\em passed}
(namely, in which the protocol is not aborted by Alice and Bob):
in such cases, Alice and Bob generate an $m$-bit key.
The cases in which the protocol aborts do not exist
in the sum---namely, they are represented by the zero operator,
which matches standard definitions of composable security
(see~\cite[Subsection 6.1.2]{renner_thesis08}).
Thus, $\rho_{\mathrm{ABE}}$ is a non-normalized state,
and $\tr(\rho_{\mathrm{ABE}})$ is the probability
that the test is passed.

To help us bound the trace distance, we define
the following intermediate state:
\begin{eqnarray}
\sigma_{\mathrm{ABE}} &\triangleq&
\sum_{\bi, \bj, \bb, \bs, P_\sC, P_\sK \mid \bT = 1}
\pr\left( \bi, \bj, \bb, \bs, P_\sC, P_\sK \right) \cdot
\ket{\bk}_\sA \bra{\bk}_\sA \otimes
\ket{\bk}_\sB \bra{\bk}_\sB \nonumber \\
&\otimes& \left( \rho_{\bi_\sI, \bj_\sI}^{\bb, \bs} \right)_\sE
\otimes \ket{\bi_\sT, \bj_\sT, \bb, \bs, \bxi, P_\sC, P_\sK}_\sC
\bra{\bi_\sT, \bj_\sT, \bb, \bs, \bxi,
P_\sC, P_\sK}_\sC.\label{eq_sigma_ABE_def}
\end{eqnarray}
This state is identical to $\rho_{\mathrm{ABE}}$, except that Bob holds
Alice's final key ($\bk$) instead of his own calculated final key
($\bk^\sB$). In particular, the similarity between
$\rho_{\mathrm{ABE}}$ and $\sigma_{\mathrm{ABE}}$ means that,
by definition,
$\rho_\sE \triangleq \tr_{\mathrm{AB}}
\left( \rho_{\mathrm{ABE}} \right)$
and $\sigma_\sE \triangleq
\tr_{\mathrm{AB}} \left( \sigma_{\mathrm{ABE}} \right)$
are identical---that is, $\rho_\sE = \sigma_\sE$.

A similar intermediate state $\sigma_{\mathrm{ABE}}^\sym$ is defined
in case Eve uses the symmetrized attack:
\begin{eqnarray}
\sigma_{\mathrm{ABE}}^\sym &\triangleq&
\sum_{\bi, \bj, \bb, \bs, P_\sC, P_\sK \mid \bT = 1}
\pr^\sym\left( \bi, \bj, \bb, \bs, P_\sC, P_\sK \right) \cdot
\ket{\bk}_\sA \bra{\bk}_\sA \otimes
\ket{\bk}_\sB \bra{\bk}_\sB \nonumber \\
&\otimes& \left( \rho_{\bi_\sI, \bj_\sI}^{\bb, \bs} \right)^\sym_\sE
\otimes \ket{\bi_\sT, \bj_\sT, \bb, \bs, \bxi, P_\sC, P_\sK}_\sC
\bra{\bi_\sT, \bj_\sT, \bb, \bs, \bxi,
P_\sC, P_\sK}_\sC.\label{eq_sigma_ABE_sym_def}
\end{eqnarray}

\begin{proposition}\label{prop_symm_sec}
For any symmetrized attack
and any integer $0 \le t \le \frac{n}{2}$, it holds that
\begin{align}
&\frac{1}{2} \tr \left| \sigma_{\mathrm{ABE}}^\sym -
\rho_\sU \otimes \sigma_\sE^\sym \right| \nonumber \\
&\le 2m \sqrt{\pr_\invb\left[ \left( \frac{|\bC_\sI|}{n} \ge
\frac{t}{n} \right) \wedge \left( \bT = 1 \right) \right]
+ 2^{n[H_2(t/n) - (n-r-m)/n]}},
\end{align}
for $\sigma_{\mathrm{ABE}}^\sym$ and $\rho_\sU$
defined in Equations~\eqref{eq_sigma_ABE_sym_def}
and~\eqref{eq_rho_U_def}, respectively;
the partial trace state
$\sigma_\sE^\sym \triangleq
\tr_{\mathrm{AB}} \left( \sigma_{\mathrm{ABE}}^\sym \right)$;
and $H_2(x) \triangleq -x \log_2(x) - (1-x) \log_2(1-x)$.
We note that the probability in the right-hand-side is the probability
for the original (non-symmetrized) attack.
\end{proposition}
\begin{proof}
We notice that only two terms in $\sigma_{\mathrm{ABE}}^\sym$
depend {\em directly} on $\bi_\sI, \bj_\sI$
(and not only on $\bk, \bxi$):
the probability $\pr^\sym\left( \bi, \bj, \bb, \bs,
P_\sC, P_\sK \right)$, and Eve's state
$\left( \rho_{\bi_\sI, \bj_\sI}^{\bb, \bs} \right)^\sym_\sE$.
For any $\bi, \bj, \bb, \bs, P_\sC, P_\sK$
and $\bxi \triangleq \bi_\sI P_\sC^\transp$,
the probability can be reformulated as
\begin{eqnarray}
\pr^\sym\left( \bi, \bj, \bb, \bs, P_\sC, P_\sK \right)
&=& \pr^\sym\left( \bi_\sT, \bj_\sT, \bb, \bs, P_\sC, P_\sK \right)
\cdot \pr^\sym\left( \bxi \mid \bi_\sT, \bj_\sT, \bb, \bs, P_\sC,
P_\sK \right) \nonumber \\
&\cdot& \pr^\sym\left( \bi_\sI \mid \bi_\sT, \bj_\sT, \bb, \bs, \bxi,
P_\sC, P_\sK \right)
\cdot \pr^\sym\left( \bj_\sI \mid \bi_\sI, \bi_\sT, \bj_\sT,
\bb, \bs, \bxi, P_\sC, P_\sK \right) \nonumber \\
&=& \pr^\sym\left( P_\sC, P_\sK \right)
\cdot \pr^\sym\left( \bi_\sT, \bj_\sT, \bb, \bs \right)
\cdot \pr^\sym\left( \bxi \mid \bi_\sT, \bj_\sT, \bb, \bs \right)
\nonumber \\
&\cdot& \frac{1}{2^{n-r}} \cdot \pr^\sym\left( \bj_\sI \mid
\bi_\sI, \bi_\sT, \bj_\sT, \bb, \bs \right) \nonumber \\
&=& \pr^\sym\left( P_\sC, P_\sK \right)
\cdot \pr^\sym\left( \bi_\sT, \bj_\sT, \bb, \bs, \bxi \right)
\nonumber \\
&\cdot& \frac{1}{2^m} \cdot \frac{1}{2^{n-r-m}}
\cdot \pr^\sym\left( \bj_\sI \mid
\bi_\sI, \bi_\sT, \bj_\sT, \bb, \bs \right).
\end{eqnarray}
(This is correct according to the conclusions
of the fourth property of the symmetrized
attack---Equations~\eqref{symm_prob_info_bits_xi_3}
and~\eqref{symm_prob_info_bits_iI};
and because as noted in Subsection~\ref{sec_joint},
the matrices $P_\sC, P_\sK$ are completely independent of
$\bi_\sI, \bi_\sT, \bj_\sI, \bj_\sT, \bb, \bs$, and therefore
$\pr^\sym\left( \bi_\sT, \bj_\sT, \bb, \bs, P_\sC, P_\sK \right)
= \pr^\sym\left( P_\sC, P_\sK \right)
\cdot \pr^\sym\left( \bi_\sT, \bj_\sT, \bb, \bs \right)$ and
$\pr^\sym\left( \bj_\sI \mid \bi_\sI, \bi_\sT, \bj_\sT,
\bb, \bs, \bxi, P_\sC, P_\sK \right) = \pr^\sym\left( \bj_\sI \mid
\bi_\sI, \bi_\sT, \bj_\sT, \bb, \bs \right)$.)

Therefore, the state $\sigma_{\mathrm{ABE}}^\sym$
from Equation~\eqref{eq_sigma_ABE_sym_def}
takes the following form:
\begin{eqnarray}
\sigma_{\mathrm{ABE}}^\sym &=& \frac{1}{2^m}
\sum_{\bk, \bi_\sT, \bj_\sT, \bb, \bs, \bxi \mid \bT = 1}
\pr^\sym\left( \bi_\sT, \bj_\sT, \bb, \bs, \bxi \right)
\cdot \ket{\bk}_\sA \bra{\bk}_\sA \otimes
\ket{\bk}_\sB \bra{\bk}_\sB \nonumber \\
&\otimes& \left[ \frac{1}{2^{n-r-m}} \sum_{P_\sC, P_\sK, \bi_\sI,
\bj_\sI \big|{\scriptsize\begin{matrix}\bi_\sI P_\sC^\transp = \bxi\\
\bi_\sI P_\sK^\transp = \bk\end{matrix}}}
\pr^\sym\left( P_\sC, P_\sK \right)
\cdot \pr^\sym\left( \bj_\sI \mid \bi_\sI, \bi_\sT, \bj_\sT, \bb, \bs
\right) \right. \nonumber \\
&\cdot& \left.
\left( \rho_{\bi_\sI, \bj_\sI}^{\bb, \bs} \right)^\sym_\sE
\otimes \ket{P_\sC, P_\sK}_\sC \bra{P_\sC, P_\sK}_\sC \right]
\otimes \ket{\bi_\sT, \bj_\sT, \bb, \bs, \bxi}_\sC
\bra{\bi_\sT, \bj_\sT, \bb, \bs, \bxi}_\sC \nonumber \\
&=& \frac{1}{2^m}
\sum_{\bk, \bi_\sT, \bj_\sT, \bb, \bs, \bxi \mid \bT = 1}
\pr^\sym\left( \bi_\sT, \bj_\sT, \bb, \bs, \bxi \right)
\cdot \ket{\bk}_\sA \bra{\bk}_\sA \otimes
\ket{\bk}_\sB \bra{\bk}_\sB \nonumber \\
&\otimes& \widehat{\rho}^\sym_\bk
\otimes \ket{\bi_\sT, \bj_\sT, \bb, \bs, \bxi}_\sC
\bra{\bi_\sT, \bj_\sT, \bb, \bs, \bxi}_\sC.
\end{eqnarray}
(This expression for $\widehat{\rho}^\sym_\bk$ was found
in Equation~\eqref{rhohatk_multi2}.)

The partial trace $\sigma_\sE^\sym \triangleq
\tr_{\mathrm{AB}} \left( \sigma_{\mathrm{ABE}}^\sym \right)$ is
\begin{eqnarray}
\sigma_\sE^\sym &=& \frac{1}{2^m}
\sum_{\bk, \bi_\sT, \bj_\sT, \bb, \bs, \bxi \mid \bT = 1}
\pr^\sym\left( \bi_\sT, \bj_\sT, \bb, \bs, \bxi \right)
\nonumber \\
&\cdot& \widehat{\rho}^\sym_\bk
\otimes \ket{\bi_\sT, \bj_\sT, \bb, \bs, \bxi}_\sC
\bra{\bi_\sT, \bj_\sT, \bb, \bs, \bxi}_\sC,
\end{eqnarray}
and the state $\rho_\sU \otimes \sigma_\sE^\sym$ is thus
(using Equation~\eqref{eq_rho_U_def})
\begin{eqnarray}
\rho_\sU \otimes \sigma_\sE^\sym &=& \frac{1}{2^{2m}}
\sum_{\bk, \bk', \bi_\sT, \bj_\sT, \bb, \bs, \bxi \mid \bT = 1}
\pr^\sym\left( \bi_\sT, \bj_\sT, \bb, \bs, \bxi \right)
\cdot \ket{\bk}_\sA \bra{\bk}_\sA \otimes
\ket{\bk}_\sB \bra{\bk}_\sB \nonumber \\
&\otimes& \widehat{\rho}^\sym_{\bk'}
\otimes \ket{\bi_\sT, \bj_\sT, \bb, \bs, \bxi}_\sC
\bra{\bi_\sT, \bj_\sT, \bb, \bs, \bxi}_\sC.
\end{eqnarray}
Since $\sigma_{\mathrm{ABE}}^\sym$
and $\rho_\sU \otimes \sigma_\sE^\sym$
are almost identical (except the difference between Eve's states:
$\widehat{\rho}^\sym_{\bk}$ and $\widehat{\rho}^\sym_{\bk'}$),
we can use the triangle inequality for norms,
the definitions of
$\Delta^\sym_\Eve(\bk, \bk' \mid \bi_\sT, \bj_\sT, \bb, \bs, \bxi)$
(Equation~\eqref{defipa})
and $\langle \Delta^\sym_\Eve(\bk, \bk') \rangle$
(Equation~\eqref{eqipa}),
and Theorem~\ref{thmsecurity} to get:
\begin{align}
&\frac{1}{2} \tr \left| \sigma_{\mathrm{ABE}}^\sym -
\rho_\sU \otimes \sigma_\sE^\sym \right| \nonumber \\
&\le \frac{1}{2^{2m}}
\sum_{\bk, \bk', \bi_\sT, \bj_\sT, \bb, \bs, \bxi \mid \bT = 1}
\pr^\sym\left( \bi_\sT, \bj_\sT, \bb, \bs, \bxi \right)
\cdot \frac{1}{2} \tr \left| \widehat{\rho}^\sym_{\bk}
- \widehat{\rho}^\sym_{\bk'} \right| \nonumber \\
&= \frac{1}{2^{2m}} \sum_{\bk, \bk', \bi_\sT, \bj_\sT, \bb, \bs, \bxi}
\pr^\sym\left( \bi_\sT, \bj_\sT, \bb, \bs, \bxi \right) \cdot
\Delta^\sym_\Eve(\bk, \bk' \mid \bi_\sT, \bj_\sT, \bb, \bs, \bxi)
\nonumber \\
&= \frac{1}{2^{2m}} \sum_{\bk, \bk'}
\langle \Delta^\sym_\Eve(\bk, \bk') \rangle
\nonumber \\
&\le 2m \sqrt{\pr_\invb\left[ \left( \frac{|\bC_\sI|}{n} \ge
\frac{t}{n} \right) \wedge \left( \bT = 1 \right) \right]
+ 2^{n[H_2(t/n) - (n-r-m)/n]}}.
\end{align}
\end{proof}

\begin{proposition}\label{prop_symm_general}
For any attack, it holds that
\begin{equation}
\frac{1}{2} \tr \left| \sigma_{\mathrm{ABE}} -
\rho_\sU \otimes \sigma_\sE \right| \le
\frac{1}{2} \tr \left| \sigma_{\mathrm{ABE}}^\sym -
\rho_\sU \otimes \sigma_\sE^\sym \right|,
\end{equation}
for $\sigma_{\mathrm{ABE}}$, $\sigma_{\mathrm{ABE}}^\sym$,
and $\rho_\sU$ defined in
Equations~\eqref{eq_sigma_ABE_def},~\eqref{eq_sigma_ABE_sym_def},
and~\eqref{eq_rho_U_def}, respectively,
and the partial trace states
$\sigma_\sE \triangleq \tr_{\mathrm{AB}}
\left( \sigma_{\mathrm{ABE}} \right)$
and $\sigma_\sE^\sym \triangleq
\tr_{\mathrm{AB}} \left( \sigma_{\mathrm{ABE}}^\sym \right)$.
\end{proposition}
\begin{proof}
First, we have to find an expression for
$\left( \rho_{\bi_\sI, \bj_\sI}^{\bb, \bs} \right)^\sym_\sE$.
According to Equation~\eqref{rho_Eve_ij},
\begin{equation}
\left( \rho_{\bi_\sI,\bj_\sI}^{\bb, \bs} \right)^\sym_\sE =
\frac{\left[ \ket{E^\sym_{\bi, \bj}\,\!'}_\bb
\bra{E^\sym_{\bi, \bj}\,\!'}_\bb \right]_\sE}
{\pr^\sym(\bj \mid \bi, \bb, \bs, P_\sC, P_\sK)},
\end{equation}
and according to the ``Basic Lemma of Symmetrization''
(see Equation~\eqref{symm_basic_lemma}),
\begin{equation}
\ket{E^\sym_{\bi, \bj}\,\!'}_\bb = \frac{1}{\sqrt{2^N}}
\sum_{\bbm \in \mathbf{F}_2^N} (-1)^{(\bi \oplus \bj) \cdot \bbm}
\ket{E'_{\bi \oplus \bbm, \bj \oplus \bbm}}_\bb \ket{\bbm}_\sM.
\end{equation}
Therefore,
\begin{equation}
\left( \rho_{\bi_\sI,\bj_\sI}^{\bb, \bs} \right)^\sym_\sE =
\frac{1}{2^N} \sum_{\bbm, \bbm' \in \mathbf{F}_2^N}
\frac{(-1)^{(\bi \oplus \bj) \cdot (\bbm \oplus \bbm')}
\left[ \ket{E'_{\bi \oplus \bbm, \bj \oplus \bbm}}_\bb
\bra{E'_{\bi \oplus \bbm', \bj \oplus \bbm'}}_\bb \otimes
\ket{\bbm}_\sM \bra{\bbm'}_\sM \right]_\sE}
{\pr^\sym(\bj \mid \bi, \bb, \bs, P_\sC, P_\sK)}.
\end{equation}

Substituting this result into Equation~\eqref{eq_sigma_ABE_sym_def},
we find:
\begin{eqnarray}
\sigma_{\mathrm{ABE}}^\sym &=& \frac{1}{2^N}
\sum_{\bi, \bj, \bb, \bs, P_\sC, P_\sK, \bbm, \bbm' \mid \bT = 1}
\pr^\sym\left( \bi, \bj, \bb, \bs, P_\sC, P_\sK \right) \cdot
\ket{\bk}_\sA \bra{\bk}_\sA \otimes
\ket{\bk}_\sB \bra{\bk}_\sB \nonumber \\
&\otimes& \frac{(-1)^{(\bi \oplus \bj) \cdot (\bbm \oplus \bbm')}
\left[ \ket{E'_{\bi \oplus \bbm, \bj \oplus \bbm}}_\bb
\bra{E'_{\bi \oplus \bbm', \bj \oplus \bbm'}}_\bb \otimes
\ket{\bbm}_\sM \bra{\bbm'}_\sM \right]_\sE}
{\pr^\sym(\bj \mid \bi, \bb, \bs, P_\sC, P_\sK)} \nonumber \\
&\otimes& \ket{\bi_\sT, \bj_\sT, \bb, \bs, \bxi, P_\sC, P_\sK}_\sC
\bra{\bi_\sT, \bj_\sT, \bb, \bs, \bxi, P_\sC, P_\sK}_\sC \nonumber \\
&=& \frac{1}{2^N} \sum_{\bi, \bj, \bb, \bs, P_\sC, P_\sK, \bbm, \bbm'
\mid \bT = 1}
\pr^\sym\left( \bi, \bb, \bs, P_\sC, P_\sK \right) \cdot
\ket{\bk}_\sA \bra{\bk}_\sA \otimes
\ket{\bk}_\sB \bra{\bk}_\sB \nonumber \\
&\otimes& (-1)^{(\bi \oplus \bj) \cdot (\bbm \oplus \bbm')}
\left[ \ket{E'_{\bi \oplus \bbm, \bj \oplus \bbm}}_\bb
\bra{E'_{\bi \oplus \bbm', \bj \oplus \bbm'}}_\bb \otimes
\ket{\bbm}_\sM \bra{\bbm'}_\sM \right]_\sE \nonumber \\
&\otimes& \ket{\bi_\sT, \bj_\sT, \bb, \bs, \bxi, P_\sC, P_\sK}_\sC
\bra{\bi_\sT, \bj_\sT, \bb, \bs, \bxi, P_\sC, P_\sK}_\sC.
\end{eqnarray}

We define a unitary operator $V$: given the states $\ket{\bbm}_\sM$
(held by Eve) and $\ket{\bs, P_\sC, P_\sK}_\sC$ (public information
sent over the classical channel),
the unitary operator $V$ performs a XOR of all
states in subsystems $\sA$, $\sB$, and $\sC$ with the relevant parts
of $\bbm$. Namely, if we define $\bbm_\sI$ and $\bbm_\sT$ as the
INFO bits and the TEST bits of $\bbm$, respectively
(which also depend on $\bs$, of course), and if we define
$\bk_\bbm \triangleq \bbm_\sI P_\sK^\transp$ and
$\bxi_\bbm \triangleq \bbm_\sI P_\sC^\transp$
(which also depend on $P_\sC, P_\sK$), then
\begin{eqnarray}
& V& \ket{\bk}_\sA \ket{\bk}_\sB
\left[ \ket{E} \ket{\bbm}_\sM \right]_\sE
\ket{\bi_\sT, \bj_\sT, \bb, \bs, \bxi, P_\sC, P_\sK}_\sC \nonumber \\
&=& \ket{\bk \oplus \bk_\bbm}_\sA \ket{\bk \oplus \bk_\bbm}_\sB
\left[ \ket{E} \ket{\bbm}_\sM \right]_\sE \nonumber \\
&\otimes& \ket{\bi_\sT \oplus \bbm_\sT, \bj_\sT \oplus \bbm_\sT, \bb,
\bs, \bxi \oplus \bxi_\bbm, P_\sC, P_\sK}_\sC.
\end{eqnarray}
Therefore (also using the fact that
$\pr^\sym\left( \bi, \bb, \bs, P_\sC, P_\sK \right)
= \pr\left( \bi, \bb, \bs, P_\sC, P_\sK \right)
= \pr\left( \bi \oplus \bbm, \bb, \bs, P_\sC, P_\sK \right)$),
\begin{eqnarray}
V \sigma_{\mathrm{ABE}}^\sym V^\dagger &=& \frac{1}{2^N}
\sum_{\bi, \bj, \bb, \bs, P_\sC, P_\sK, \bbm, \bbm' \mid \bT = 1}
\pr\left( \bi \oplus \bbm, \bb, \bs, P_\sC, P_\sK \right) \nonumber \\
&\cdot& \ket{\bk \oplus \bk_\bbm}_\sA \bra{\bk \oplus \bk_{\bbm'}}_\sA
\otimes \ket{\bk \oplus \bk_\bbm}_\sB \bra{\bk \oplus \bk_{\bbm'}}_\sB
\nonumber \\
&\otimes& (-1)^{(\bi \oplus \bj) \cdot (\bbm \oplus \bbm')}
\left[ \ket{E'_{\bi \oplus \bbm, \bj \oplus \bbm}}_\bb
\bra{E'_{\bi \oplus \bbm', \bj \oplus \bbm'}}_\bb \otimes
\ket{\bbm}_\sM \bra{\bbm'}_\sM \right]_\sE \nonumber \\
&\otimes& \ket{\bi_\sT \oplus \bbm_\sT, \bj_\sT \oplus \bbm_\sT,
\bb, \bs, \bxi \oplus \bxi_\bbm, P_\sC, P_\sK}_\sC \nonumber \\
&&\bra{\bi_\sT \oplus \bbm'_\sT, \bj_\sT \oplus \bbm'_\sT, \bb, \bs,
\bxi \oplus \bxi_{\bbm'}, P_\sC, P_\sK}_\sC.
\end{eqnarray}
Tracing out subsystem $\sM$ (which is part of Eve's probe), we get
\begin{eqnarray}
\tr_\sM \left[ V \sigma_{\mathrm{ABE}}^\sym V^\dagger \right]
&=& \frac{1}{2^N} \sum_{\bi, \bj, \bb, \bs, P_\sC, P_\sK, \bbm
\mid \bT = 1}
\pr\left( \bi \oplus \bbm, \bb, \bs, P_\sC, P_\sK \right) \nonumber \\
&\cdot& \ket{\bk \oplus \bk_\bbm}_\sA \bra{\bk \oplus \bk_\bbm}_\sA
\otimes \ket{\bk \oplus \bk_\bbm}_\sB \bra{\bk \oplus \bk_\bbm}_\sB
\nonumber \\
&\otimes& \left[ \ket{E'_{\bi \oplus \bbm, \bj \oplus \bbm}}_\bb
\bra{E'_{\bi \oplus \bbm, \bj \oplus \bbm}}_\bb \right]_\sE
\nonumber \\
&\otimes& \ket{\bi_\sT \oplus \bbm_\sT, \bj_\sT \oplus \bbm_\sT,
\bb, \bs, \bxi \oplus \bxi_\bbm, P_\sC, P_\sK}_\sC \nonumber \\
&&\bra{\bi_\sT \oplus \bbm_\sT, \bj_\sT \oplus \bbm_\sT, \bb, \bs,
\bxi \oplus \bxi_\bbm, P_\sC, P_\sK}_\sC.
\end{eqnarray}
Now we change the indexes:
we denote $\bi' \triangleq \bi \oplus \bbm$ and
$\bj' \triangleq \bj \oplus \bbm$ (for a fixed $\bbm$),
and according to the definitions, we immediately get
$\bi'_\sT = \bi_\sT \oplus \bbm_\sT$, ~
$\bj'_\sT = \bj_\sT \oplus \bbm_\sT$,~
$\bk' \triangleq \bi'_\sI P_\sK^\transp = (\bi_\sI \oplus \bbm_\sI)
P_\sK^\transp = \bk \oplus \bk_\bbm$, and similarly
$\bxi' \triangleq \bi'_\sI P_\sC^\transp = \bxi \oplus \bxi_\bbm$.
We also notice that
$\bT$ gets $(\bi_\sT \oplus \bj_\sT, \bb_\sT, \bs)$ as inputs,
and that they all stay the same
(because $\bi'_\sT \oplus \bj'_\sT =
(\bi_\sT \oplus \bbm_\sT) \oplus (\bj_\sT \oplus \bbm_\sT) =
\bi_\sT \oplus \bj_\sT$), and therefore the change of indexes does not
impact the condition $\bT = 1$. Therefore,
\begin{eqnarray}
\tr_\sM \left[ V \sigma_{\mathrm{ABE}}^\sym V^\dagger \right]
&=& \frac{1}{2^N} \sum_{\bi', \bj', \bb, \bs, P_\sC, P_\sK, \bbm
\mid \bT = 1}
\pr\left( \bi', \bb, \bs, P_\sC, P_\sK \right) \cdot
\ket{\bk'}_\sA \bra{\bk'}_\sA \otimes
\ket{\bk'}_\sB \bra{\bk'}_\sB \nonumber \\
&\otimes& \left[ \ket{E'_{\bi', \bj'}}_\bb
\bra{E'_{\bi', \bj'}}_\bb \right]_\sE \nonumber \\
&\otimes& \ket{\bi'_\sT, \bj'_\sT, \bb, \bs, \bxi', P_\sC, P_\sK}_\sC
\bra{\bi'_\sT, \bj'_\sT, \bb, \bs, \bxi', P_\sC, P_\sK}_\sC.
\end{eqnarray}
Using the relation $\left( \rho_{\bi_\sI,\bj_\sI}^{\bb, \bs}
\right)_\sE = \frac{\left[ \ket{E'_{\bi, \bj}}_\bb
\bra{E'_{\bi, \bj}}_\bb \right]_\sE}
{\pr(\bj \mid \bi, \bb, \bs, P_\sC, P_\sK)}$
from Equation~\eqref{rho_Eve_ij}, we get
\begin{eqnarray}
\tr_\sM \left[ V \sigma_{\mathrm{ABE}}^\sym V^\dagger \right]
&=& \sum_{\bi', \bj', \bb, \bs, P_\sC, P_\sK \mid \bT = 1}
\pr\left( \bi', \bb, \bs, P_\sC, P_\sK \right) \cdot
\ket{\bk'}_\sA \bra{\bk'}_\sA \otimes
\ket{\bk'}_\sB \bra{\bk'}_\sB \nonumber \\
&\otimes& \pr(\bj' \mid \bi', \bb, \bs, P_\sC, P_\sK)
\left( \rho_{\bi'_\sI,\bj'_\sI}^{\bb, \bs} \right)_\sE \nonumber \\
&\otimes& \ket{\bi'_\sT, \bj'_\sT, \bb, \bs, \bxi', P_\sC, P_\sK}_\sC
\bra{\bi'_\sT, \bj'_\sT, \bb, \bs, \bxi', P_\sC, P_\sK}_\sC
\nonumber \\
&=& \sum_{\bi', \bj', \bb, \bs, P_\sC, P_\sK \mid \bT = 1}
\pr\left( \bi', \bj', \bb, \bs, P_\sC, P_\sK \right) \cdot
\ket{\bk'}_\sA \bra{\bk'}_\sA \otimes
\ket{\bk'}_\sB \bra{\bk'}_\sB \nonumber \\
&\otimes& \left( \rho_{\bi'_\sI,\bj'_\sI}^{\bb, \bs} \right)_\sE
\otimes \ket{\bi'_\sT, \bj'_\sT, \bb, \bs, \bxi', P_\sC, P_\sK}_\sC
\bra{\bi'_\sT, \bj'_\sT, \bb, \bs, \bxi', P_\sC, P_\sK}_\sC
\nonumber \\
&=& \sigma_{\mathrm{ABE}},
\end{eqnarray}
where the final equality is according to the definition
(Equation~\eqref{eq_sigma_ABE_def}).

To sum up, we get the result
$\sigma_{\mathrm{ABE}} = \tr_\sM \left[ V
\sigma_{\mathrm{ABE}}^\sym V^\dagger \right]$;
a very similar proof gives us the result
$\rho_\sU \otimes \sigma_\sE = \tr_\sM \left[ V
\left( \rho_\sU \otimes \sigma_\sE^\sym \right) V^\dagger \right]$.
Since the trace distance is preserved under
unitary operators~\cite[Equation~(9.21) in page~404]{NCBook}
and does not increase under partial trace~\cite[Theorem~9.2 and
page~407]{NCBook}, we get
\begin{eqnarray}
\frac{1}{2} \tr \left| \sigma_{\mathrm{ABE}} -
\rho_\sU \otimes \sigma_\sE \right|
&=& \frac{1}{2} \tr \left| \tr_\sM \left[ V
\left( \sigma_{\mathrm{ABE}}^\sym -
\rho_\sU \otimes \sigma_\sE^\sym \right) V^\dagger \right] \right|
\nonumber \\
&\le& \frac{1}{2} \tr \left| V \left( \sigma_{\mathrm{ABE}}^\sym -
\rho_\sU \otimes \sigma_\sE^\sym \right) V^\dagger \right| \nonumber \\
&=& \frac{1}{2} \tr \left| \sigma_{\mathrm{ABE}}^\sym -
\rho_\sU \otimes \sigma_\sE^\sym \right|.
\end{eqnarray}
\end{proof}

\begin{proposition}\label{prop_rel}
For any attack,
\begin{equation}
\frac{1}{2} \tr \left| \rho_{\mathrm{ABE}}
- \sigma_{\mathrm{ABE}} \right|
\le \pr\left[ \left( \bk \ne \bk^\sB \right)
\wedge \left( \bT = 1 \right) \right],
\end{equation}
where $\rho_{\mathrm{ABE}}$ and $\sigma_{\mathrm{ABE}}$ were
defined in Equations~\eqref{eq_rho_ABE_def}
and~\eqref{eq_sigma_ABE_def}, respectively;
$\bk$ is the final key computed by Alice;
and $\bk^\sB$ is the final key computed by Bob.
\end{proposition}
\begin{proof}
\begin{eqnarray}
\rho_{\mathrm{ABE}} - \sigma_{\mathrm{ABE}} &=&
\sum_{\bi, \bj, \bb, \bs, P_\sC, P_\sK \mid \bT = 1}
\pr\left( \bi, \bj, \bb, \bs, P_\sC, P_\sK \right) \nonumber \\
&\cdot& \ket{\bk}_\sA \bra{\bk}_\sA
\otimes \left[ \ket{\bk^\sB}_\sB \bra{\bk^\sB}_\sB -
\ket{\bk}_\sB \bra{\bk}_\sB \right]
\otimes \left( \rho_{\bi_\sI, \bj_\sI}^{\bb, \bs} \right)_\sE
\nonumber \\
&\otimes& \ket{\bi_\sT, \bj_\sT, \bb, \bs, \bxi, P_\sC, P_\sK}_\sC
\bra{\bi_\sT, \bj_\sT, \bb, \bs, \bxi, P_\sC, P_\sK}_\sC \nonumber \\
&=& \pr\left[ \left( \bk \ne \bk^\sB \right)
\wedge \left( \bT = 1 \right) \right] \nonumber \\
&\cdot& \sum_{\bi, \bj, \bb, \bs, P_\sC, P_\sK}
\pr\left[ \bi, \bj, \bb, \bs, P_\sC, P_\sK \mid
\left( \bk \ne \bk^\sB \right)
\wedge \left( \bT = 1 \right) \right] \nonumber \\
&\cdot& \ket{\bk}_\sA \bra{\bk}_\sA
\otimes \left[ \ket{\bk^\sB}_\sB \bra{\bk^\sB}_\sB -
\ket{\bk}_\sB \bra{\bk}_\sB \right]
\otimes \left( \rho_{\bi_\sI, \bj_\sI}^{\bb, \bs} \right)_\sE
\nonumber \\
&\otimes& \ket{\bi_\sT, \bj_\sT, \bb, \bs, \bxi, P_\sC, P_\sK}_\sC
\bra{\bi_\sT, \bj_\sT, \bb, \bs, \bxi, P_\sC, P_\sK}_\sC.
\end{eqnarray}
The trace distance between any two normalized
states is bounded by $1$. Therefore,
\begin{equation}
\frac{1}{2} \tr \left| \rho_{\mathrm{ABE}}
- \sigma_{\mathrm{ABE}} \right|
\le \pr\left[ \left( \bk \ne \bk^\sB \right)
\wedge \left( \bT = 1 \right) \right].
\end{equation}
\end{proof}

\begin{corollary}\label{corollary_compos}
For any attack
and any integer $0 \le t \le \frac{n}{2}$,
\begin{align}
&\frac{1}{2} \tr \left| \rho_{\mathrm{ABE}}
- \rho_\sU \otimes \rho_\sE \right| \nonumber \\
&\le \pr\left[ \left( \bk \ne \bk^\sB \right)
\wedge \left( \bT = 1 \right) \right] \nonumber \\
&+ 2m \sqrt{\pr_\invb\left[ \left( \frac{|\bC_\sI|}{n} \ge
\frac{t}{n} \right) \wedge \left( \bT = 1 \right) \right]
+ 2^{n[H_2(t/n) - (n-r-m)/n]}},
\end{align}
for $\rho_{\mathrm{ABE}}$ and $\rho_\sU$
defined in Equations~\eqref{eq_rho_ABE_def}
and~\eqref{eq_rho_U_def}, respectively;
the partial trace state
$\rho_\sE \triangleq \tr_{\mathrm{AB}}
\left( \rho_{\mathrm{ABE}} \right)$;
and $H_2(x) \triangleq -x \log_2(x) - (1-x) \log_2(1-x)$.
\end{corollary}
\begin{proof}
Using Propositions~\ref{prop_symm_sec}, \ref{prop_symm_general},
and~\ref{prop_rel} and the fact that $\rho_\sE = \sigma_\sE$,
we get:
\begin{align}
&\frac{1}{2} \tr \left| \rho_{\mathrm{ABE}}
- \rho_\sU \otimes \rho_\sE \right| \nonumber \\
&\le \frac{1}{2} \tr \left| \rho_{\mathrm{ABE}}
- \sigma_{\mathrm{ABE}} \right|
+ \frac{1}{2} \tr \left| \sigma_{\mathrm{ABE}}
- \rho_\sU \otimes \sigma_\sE \right| \nonumber \\
&\le \frac{1}{2} \tr \left| \rho_{\mathrm{ABE}}
- \sigma_{\mathrm{ABE}} \right|
+ \frac{1}{2} \tr \left| \sigma_{\mathrm{ABE}}^\sym
- \rho_\sU \otimes \sigma_\sE^\sym \right| \nonumber \\
&\le \pr\left[ \left( \bk \ne \bk^\sB \right)
\wedge \left( \bT = 1 \right) \right] \nonumber \\
&+ 2m \sqrt{\pr_\invb\left[ \left( \frac{|\bC_\sI|}{n} \ge
\frac{t}{n} \right) \wedge \left( \bT = 1 \right) \right]
+ 2^{n[H_2(t/n) - (n-r-m)/n]}}.
\end{align}
\end{proof}

We have thus found an upper bound for the expression
$\frac{1}{2} \tr \left| \rho_{\mathrm{ABE}}
- \rho_\sU \otimes \rho_\sE \right|$.
In Section~\ref{sec_protocols} we prove this upper bound to be
exponentially small in $n$ for specific protocols.

\section{\label{sec_protocols}Full Security Proofs for Specific
Protocols}
Below we prove full security for specific important examples
of generalized BB84 protocols.

\subsection{Hoeffding's Theorem}
The rest of our security proof consists mainly of applications of
the following Theorem, proven by Hoeffding in~\cite[Section 6]{Ho63}:
\begin{theorem}[Hoeffding's Theorem]\label{theorem_hoeffding}
Let $X_1, \ldots, X_n$ be a random sample without replacement
taken from a population $c_1, \ldots, c_N$ such that
$a \le c_j \le b$ for all $1 \le j \le N$.
(That is, each $X_i$ gets the value of a random $c_j$,
such that the same $j$ is never chosen for two different
variables $X_i,X_{i'}$.)
If we denote $\overline{X} \triangleq \frac{X_1 + ... + X_n}{n}$ and
$\mu \triangleq E[\overline{X}]$ (namely, $\mu$ is the expected value
of $\overline{X}$), then:
\begin{enumerate}
\item For any $\epsilon > 0$,
\begin{equation}
\pr\left[ \overline{X} - \mu \ge \epsilon \right] \le
e^{-\frac{2n\epsilon^2}{(b - a)^2}}.
\end{equation}
\item $\mu = \frac{1}{N} \sum_{i = 1}^N c_i$: namely, the expected
value of $\overline{X}$ is the average value of the population.
\end{enumerate}
\end{theorem}

The following Corollary of Hoeffding's theorem
is useful for proving security:
\begin{corollary}\label{corollary_hoeffding}
Given an $(n + n')$-bit string $\bc = c_1 \ldots c_{n + n'}$,
assume we randomly and uniformly choose
a partition of $\bc$ into two substrings,
$\bc_\sA$ of length $n$ and $\bc_\sB$ of length $n'$.
(Formally, this is a random partition of the index set
$\{1, \ldots, n + n'\}$ into two disjoint sets $\sA,\sB$
satisfying $|\sA| = n$, $|\sB| = n'$,
and $\sA \cup \sB = \{1, \ldots, n + n'\}$.)
Then, for any $p > 0$ and $\epsilon > 0$,
\begin{equation}
\pr\left[ \left( \frac{|\bC_\sA|}{n} \ge p + \epsilon \right)
\wedge \left( \frac{|\bC_\sB|}{n'} \le p \right) \right]
\le e^{-2 \left( \frac{n'}{n + n'} \right)^2 n \epsilon^2},
\end{equation}
where $\bC_\sA,\bC_\sB$ are random variables
whose values equal to $\bc_\sA$ and $\bc_\sB$, respectively,
and $|\bC_\sA|,|\bC_\sB|$ are the corresponding numbers
of $1$-bits in these bitstrings.
\end{corollary}
\begin{proof}
The random and uniform partition of $\bc$ into two substrings,
$\bc_\sA$ of length $n$ and $\bc_\sB$ of length $n'$,
is actually a sample of size $n$ without replacement from the
population $c_1, \ldots, c_{n + n'} \in \{0, 1\}$.
(The sampled $n$ bits are the bits of $\bc_\sA$,
while the other $n'$ bits are the bits of $\bc_\sB$.)
Therefore, we can apply Hoeffding's theorem
(Theorem~\ref{theorem_hoeffding}) to this sampling.

Let $\overline{X}$ be the average of the sample,
and let $\mu$ be the expected value of $\overline{X}$
(so, according to Theorem~\ref{theorem_hoeffding},
$\mu$ is the average value of the population), then
\begin{eqnarray}
\overline{X} &=& \frac{|\bC_\sA|}{n}, \\
\mu &=& \frac{|\bC_\sA| + |\bC_\sB|}{n + n'}.
\end{eqnarray}
The condition $\frac{|\bC_\sB|}{n'} \le p$ is thus equivalent to
$(n + n') \mu - n \overline{X} \le n' \cdot p$, which is equivalent
to $n \cdot (\overline{X} - \mu) \ge n' \cdot (\mu - p)$.
Therefore, the conditions
$\left( \frac{|\bC_\sA|}{n} \ge p + \epsilon \right)$ and
$\left( \frac{|\bC_\sB|}{n'} \le p \right)$ rewrite to
\begin{equation}
\left( \overline{X} - \mu \ge \epsilon + p - \mu \right)
\wedge \left( \frac{n}{n'} \cdot (\overline{X} - \mu)
\ge \mu - p \right).
\end{equation}
This condition implies
$\left( 1 + \frac{n}{n'} \right)(\overline{X} - \mu) \ge \epsilon$,
which is equivalent to $\overline{X} - \mu \ge
\frac{n'}{n + n'} \epsilon$.
Using Hoeffding's theorem (Theorem~\ref{theorem_hoeffding})
with $a = 0, b = 1$, we get
\begin{equation}
\pr\left[ \left( \frac{|\bC_\sA|}{n} \ge p + \epsilon \right)
\wedge \left( \frac{|\bC_\sB|}{n'} \le p \right) \right]
\le \pr\left[ \overline{X} - \mu \ge \frac{n'}{n + n'} \epsilon \right]
\le e^{-2 \left( \frac{n'}{n + n'} \right)^2 n \epsilon^2}.
\end{equation}
\end{proof}

We are allowed to use Corollary~\ref{corollary_hoeffding} for comparing
the error rates in different sets of qubits (e.g., INFO and TEST bits),
under the condition that the random and uniform sampling occurs
only {\em after} the qubits are sent by Alice and measured by Bob.
In other words, the sampling cannot affect the bases in which
the qubits are sent and measured, and it cannot affect Eve's attack.
This means, in particular, that Alice's choice of $\bs$ must be
independent of $\bi$ and kept in secret from Bob and Eve
until Bob completes his measurement.
(Some independence between $\bs$ and $\bb$ must also be assumed,
but we discuss this condition separately for each protocol.)

Similar uses of Hoeffding's theorem for proving security of QKD
are available in~\cite{BBBMR06,BGM09}.

We also use another Theorem,
proven by Hoeffding in~\cite[Section~2, Theorem~1]{Ho63}:
\begin{theorem}\label{theorem_hoeffding_indep}
Let $X_1, \ldots, X_N$ be independent random variables
with finite first and second moments,
such that $0 \le X_i \le 1$ for all $1 \le i \le N$.
If $\overline{X} \triangleq \frac{X_1 + ... + X_N}{N}$ and
$\mu \triangleq E[\overline{X}]$ is the expected value
of $\overline{X}$, then for any $\epsilon > 0$,
\begin{equation}
\pr\left[ \overline{X} - \mu \ge \epsilon \right] \le
e^{-2N\epsilon^2},
\end{equation}
and, in a similar way (see~\cite[Section~1]{Ho63}),
\begin{equation}
\pr\left[ \mu - \overline{X} \ge \epsilon \right] \le
e^{-2N\epsilon^2}.
\end{equation}
\end{theorem}

We will use the following Corollary of
Theorem~\ref{theorem_hoeffding_indep} for proving
the security of the ``efficient BB84'' protocol,
described in Subsection~\ref{subsec_eff_bb84}:
\begin{corollary}\label{corollary_hoeffding_indep}
Let $0 \le p \le 1$ be a parameter,
and let $\bb = b_1 \ldots b_N$ be an $N$-bit string,
such that each $b_i$ is chosen probabilistically and independently
out of $\{0, 1\}$, with $\pr(b_i = 0) = p$ and $\pr(b_i = 1) = 1 - p$.
Then:
\begin{eqnarray}
\pr\left( |\bb| \le \frac{(1 - p) N}{2} \right) &\le&
e^{-\frac{1}{2} N (1 - p)^2}, \\
\pr\left( |\overline{\bb}| \le \frac{p N}{2} \right) &\le&
e^{-\frac{1}{2} N p^2}.
\end{eqnarray}
\end{corollary}
\begin{proof}
Let us define $X_i = b_i$ for all $1 \le i \le N$.
Then $X_i$ are independent random variables
with finite first and second moments,
such that $0 \le X_i \le 1$ for all $1 \le i \le N$
and $\mu \triangleq E[\overline{X}] = 1 - p$.
Therefore, using Theorem~\ref{theorem_hoeffding_indep},
we get the two following results:
\begin{eqnarray}
\pr\left[ (1 - p) - \overline{X} \ge \frac{1 - p}{2} \right] &\le&
e^{-\frac{1}{2} N (1 - p)^2}, \\
\pr\left[ \overline{X} - (1 - p) \ge \frac{p}{2} \right] &\le&
e^{-\frac{1}{2} N p^2}.
\end{eqnarray}

We notice that $\overline{X} = \frac{|\bb|}{N}
= 1 - \frac{|\overline{\bb}|}{N}$.
Substituting this result, we get
\begin{eqnarray}
\pr\left[ -\frac{|\bb|}{N} \ge -\frac{1 - p}{2} \right]
&\le& e^{-\frac{1}{2} N (1 - p)^2}, \\
\pr\left[ 1 - \frac{|\overline{\bb}|}{N} - 1 \ge -\frac{p}{2} \right]
&\le& e^{-\frac{1}{2} N p^2},
\end{eqnarray}
and, therefore,
\begin{eqnarray}
\pr\left[ |\bb| \le \frac{(1 - p) N}{2} \right] &\le&
e^{-\frac{1}{2} N (1 - p)^2}, \\
\pr\left[ |\overline{\bb}| \le \frac{p N}{2} \right] &\le&
e^{-\frac{1}{2} N p^2}.
\end{eqnarray}
\end{proof}

\subsection{\label{subsec_bb84_info_z}The BB84-INFO-$z$ Protocol}
In the BB84-INFO-$z$ protocol, all INFO bits are sent by Alice
in the $z$ basis, while the TEST bits are sent in both the $z$
and the $x$ bases.
This means that $\bb$ and $\bs$ together define a random partition
of the set of indexes $\lbrace 1, 2, \ldots, N \rbrace$
into three disjoint sets:
\begin{itemize}
\item $\sI$ (INFO bits, where $s_j = 1$ and $b_j = 0$) of size $n$;
\item $\sT_\sZ$ (TEST-Z bits, where $s_j = 0$ and $b_j = 0$)
of size $n_z$; and
\item $\sT_\sX$ (TEST-X bits, where $s_j = 0$ and $b_j = 1$)
of size $n_x$.
\end{itemize}

Formally, Alice and Bob agree on parameters $n, n_z, n_x$
(such that $N = n + n_z + n_x$),
and we choose $B = \{\bb \in \mathbf{F}^{N}_2 \mid |\bb| = n_x\}$
and $S_\bb = \{\bs \in \mathbf{F}^{N}_2 \mid
(|\bs| = n) \wedge (|\bs \oplus \bb| = n + n_x)\}$
(namely, $\bs \in S_\bb$ if it consists of $n$ 1-bits
that do not overlap with the $n_x$ 1-bits of $\bb$)
for all $\bb \in B$. The probability distributions $\pr(\bb)$ and
$\pr(\bs \mid \bb)$ are all uniform---namely, $\pr(\bb, \bs)$
is identical for all $\bb \in B$ and $\bs \in S_\bb$.

Alice and Bob also agree on error rate thresholds, $p_{a,z}$
and $p_{a,x}$ (for the TEST-Z and TEST-X bits, respectively).
The testing function $T$ is defined as follows:
\begin{equation}
T(\bi_\sT \oplus \bj_\sT, \bb_\sT, \bs) = 1 ~ \Leftrightarrow ~
\left( |\bi_{\sT_\sZ} \oplus \bj_{\sT_\sZ}| \le n_z \cdot p_{a,z}
\right) \wedge
\left( |\bi_{\sT_\sX} \oplus \bj_{\sT_\sX}| \le n_x \cdot p_{a,x}
\right).
\end{equation}
Namely, the test passes if and only if
the error rate on the TEST-Z bits is at most $p_{a,z}$ {\em and}
the error rate on the TEST-X bits is at most $p_{a,x}$.

From Corollary~\ref{corollary_compos}
we get the following bound for any integer $0 \le t \le \frac{n}{2}$:
\begin{align}
&\frac{1}{2} \tr \left| \rho_{\mathrm{ABE}}
- \rho_\sU \otimes \rho_\sE \right| \nonumber \\
&\le \pr\left[ \left( \bk \ne \bk^\sB \right)
\wedge \left( \bT = 1 \right) \right] \nonumber \\
&+ 2m \sqrt{\pr_\invb\left[ \left( \frac{|\bC_\sI|}{n} \ge
\frac{t}{n} \right) \wedge \left( \bT = 1 \right) \right]
+ 2^{n[H_2(t/n) - (n-r-m)/n]}}.
\end{align}
Below we prove the two probabilities
in the right-hand-side to be exponentially small in $n$:

\begin{theorem}\label{thmsecurity1}
For any $n, n_z, n_x > 0$, $p_{a,z}, p_{a,x} > 0$,
and $\epsilon_\secur > 0$
such that $t \triangleq n (p_{a,x} + \epsilon_\secur)$ is an integer
and $p_{a,x} + \epsilon_\secur \le \frac{1}{2}$,
it holds for the BB84-INFO-$z$ protocol that
\begin{equation}\label{EveInfoBound2}
\pr_\invb\left[ \left( \frac{|\bC_\sI|}{n} \ge \frac{t}{n}
\right) \wedge \left( \bT = 1 \right) \right]
\le e^{-2 \left( \frac{n_x}{n + n_x} \right)^2 n \epsilon_\secur^2}.
\end{equation}
\end{theorem}
\begin{proof}
Because $\frac{t}{n} \triangleq p_{a,x} + \epsilon_\secur$,
it holds that
\begin{eqnarray}
&&\pr_\invb\left[ \left( \frac{|\bC_\sI|}{n} \ge \frac{t}{n} \right)
\wedge \left( \bT = 1 \right) \right] \nonumber \\
&=& \pr_\invb\left[ \left( \frac{|\bC_\sI|}{n} \ge
p_{a,x} + \epsilon_\secur \right)
\wedge \left( \frac{|\bC_{\sT_\sZ}|}{n_z} \le p_{a,z} \right)
\wedge \left( \frac{|\bC_{\sT_\sX}|}{n_x} \le p_{a,x} \right) \right]
\nonumber \\
&\le& \pr_\invb\left[ \left( \frac{|\bC_\sI|}{n} \ge
p_{a,x} + \epsilon_\secur\right) \wedge
\left( \frac{|\bC_{\sT_\sX}|}{n_x} \le p_{a,x} \right) \right].
\label{PcPkInequality}
\end{eqnarray}

In the hypothetical ``inverted-INFO-basis'' protocol,
the INFO and TEST-X bits are sent and measured in the $x$ basis,
while the TEST-Z bits are sent and measured in the $z$ basis.
Therefore, the random and uniform sampling of the $n$ INFO bits
out of the $n + n_x$ bits sent in the $x$ basis
(assuming that the TEST-Z bits have already been chosen)
does not affect the bases in the hypothetical protocol.
This means that we can apply Corollary~\ref{corollary_hoeffding}
to this sampling, and we get
\begin{equation}
\pr_\invb\left[ \left( \frac{|\bC_\sI|}{n} \ge
p_{a,x} + \epsilon_\secur \right) \wedge
\left( \frac{|\bC_{\sT_\sX}|}{n_x} \le p_{a,x} \right) \right]
\le e^{-2 \left( \frac{n_x}{n + n_x} \right)^2
n \epsilon_\secur^2}.
\end{equation}
\end{proof}

\begin{theorem}\label{thmsecurity2}
For any $n, n_z, n_x > 0$, $p_{a,z}, p_{a,x} > 0$,
and $\epsilon_\rel > 0$
such that $t_\rel \triangleq n (p_{a,z} + \epsilon_\rel)$ is an integer
and $p_{a,z} + \epsilon_\rel \le \frac{1}{2}$,
it holds for the BB84-INFO-$z$ protocol that
\begin{equation}
\pr\left[ \left( \bk \ne \bk^\sB \right)
\wedge \left( \bT = 1 \right) \right]
\le e^{-2 \left( \frac{n_z}{n + n_z}\right)^2 n \epsilon_\rel^2}
+ 2^{n[H_2(p_{a,z} + \epsilon_\rel) - r/n]},
\end{equation}
where $H_2(x) \triangleq -x \log_2(x) - (1-x) \log_2(1-x)$.
\end{theorem}
\begin{proof}
If $\bk \ne \bk^\sB$, Alice and Bob have different final keys,
and this means that the error correction stage did not succeed.
The $r \times n$ parity-check matrix $P_\sC$
is chosen uniformly at random
and is completely independent of the error word
$\bc_\sI \triangleq \bi_\sI \oplus \bj_\sI$
(because Alice sends it to Bob only after Eve's attack); therefore,
applying Corollary~\ref{corollary_ecc_decoding},
we can find that if there are at most
$t_\rel \triangleq n (p_{a,z} + \epsilon_\rel)$ errors
in the error word $\bc_\sI$, then the probability a random
error-correcting code fails to correct the error word is at most
$2^{n[H_2(t_\rel/n) - r/n]} =
2^{n[H_2(p_{a,z} + \epsilon_\rel) - r/n]}$.

Therefore, failure of the error correction stage ($\bk \ne \bk^\sB$)
must mean that either there are more than $t_\rel$ errors
in the INFO bits (namely, $\frac{|\bC_\sI|}{n} \ge \frac{t_\rel}{n}
\triangleq p_{a,z} + \epsilon_\rel$)
or there are at most $t_\rel$ errors in the INFO bits
but the error-correcting code failed in correcting the error string.
Therefore,
\begin{eqnarray}
\pr\left[ \left( \bk \ne \bk^\sB \right)
\wedge \left( \bT = 1 \right) \right]
&=& \pr\left[ \left( \bk \ne \bk^\sB \right)
\wedge \left( \frac{|\bC_{\sT_\sZ}|}{n_z} \le p_{a,z} \right)
\wedge \left( \frac{|\bC_{\sT_\sX}|}{n_x} \le p_{a,x} \right) \right]
\nonumber \\
&\le& \pr\left[ \left( \bk \ne \bk^\sB \right)
\wedge \left( \frac{|\bC_{\sT_\sZ}|}{n_z} \le p_{a,z} \right) \right]
\nonumber \\
&\le& \pr\left[ \left( \frac{|\bC_\sI|}{n} \ge
p_{a,z} + \epsilon_\rel \right)
\wedge \left( \frac{|\bC_{\sT_\sZ}|}{n_z} \le p_{a,z} \right) \right]
\nonumber \\
&+& 2^{n[H_2(p_{a,z} + \epsilon_\rel) - r/n]}.
\end{eqnarray}

In the real protocol,
the INFO and TEST-Z bits are sent and measured in the $z$ basis,
while the TEST-X bits are sent and measured in the $x$ basis.
Therefore, the random and uniform sampling of the $n$ INFO bits
out of the $n + n_z$ bits sent in the $z$ basis
(assuming that the TEST-X bits have already been chosen)
does not affect the bases in the real protocol.
This means that we can apply Corollary~\ref{corollary_hoeffding}
to this sampling, and we get
\begin{equation}
\pr\left[ \left( \frac{|\bC_\sI|}{n} \ge p_{a,z} + \epsilon_\rel \right)
\wedge \left( \frac{|\bC_{\sT_\sZ}|}{n_z} \le p_{a,z} \right) \right]
\le e^{-2 \left( \frac{n_z}{n + n_z} \right)^2 n \epsilon_\rel^2}.
\end{equation}
We thus conclude:
\begin{equation}
\pr\left[ \left( \bk \ne \bk^\sB \right)
\wedge \left( \bT = 1 \right) \right]
\le e^{-2 \left( \frac{n_z}{n + n_z} \right)^2 n \epsilon_\rel^2}
+ 2^{n[H_2(p_{a,z} + \epsilon_\rel) - r/n]}.
\end{equation}
\end{proof}

Combining the results of
Corollary~\ref{corollary_compos}, Theorem~\ref{thmsecurity1},
and Theorem~\ref{thmsecurity2}, we get:
\begin{corollary}\label{corsecurity}
For any $n, n_z, n_x > 0$, $p_{a,z}, p_{a,x} > 0$,
and $\epsilon_\secur, \epsilon_\rel > 0$
such that $n (p_{a,x} + \epsilon_\secur)$
and $n (p_{a,z} + \epsilon_\rel)$ are both integers,
$p_{a,x} + \epsilon_\secur \le \frac{1}{2}$,
and $p_{a,z} + \epsilon_\rel \le \frac{1}{2}$,
it holds for the BB84-INFO-$z$ protocol that
\begin{align}
&\frac{1}{2} \tr \left| \rho_{\mathrm{ABE}}
- \rho_\sU \otimes \rho_\sE \right| \nonumber \\
&\le e^{-2 \left( \frac{n_z}{n + n_z} \right)^2 n \epsilon_\rel^2}
+ 2^{n[H_2(p_{a,z} + \epsilon_\rel) - r/n]} \nonumber \\
&+ 2m
\sqrt{e^{-2 \left( \frac{n_x}{n + n_x} \right)^2 n \epsilon_\secur^2}
+ 2^{n[H_2(p_{a,x} + \epsilon_\secur) - (n-r-m)/n]}},
\end{align}
where $H_2(x) \triangleq -x \log_2(x) - (1-x) \log_2(1-x)$.
\end{corollary}
This bound is exponentially small in $n$
under the two following conditions:
\begin{eqnarray}
H_2(p_{a,z} + \epsilon_\rel) &<& \frac{r}{n}, \\
H_2(p_{a,x} + \epsilon_\secur) &<& \frac{n-r-m}{n},
\end{eqnarray}
which are equivalent to:
\begin{eqnarray}
H_2(p_{a,z} + \epsilon_\rel) &<& \frac{r}{n}, \\
H_2(p_{a,x} + \epsilon_\secur) + H_2(p_{a,z} + \epsilon_\rel)
&<& \frac{n-m}{n} = 1 - \frac{m}{n}.
\end{eqnarray}
We can thus obtain the following upper bound on the bit-rate:
\begin{equation}
R_{\mathrm{secret}} \triangleq \frac{m}{n}
< 1 - H_2(p_{a,x} + \epsilon_\secur)
- H_2(p_{a,z} + \epsilon_\rel).
\end{equation}

To get the asymptotic error rate thresholds, we require
$R_{\mathrm{secret}} > 0$, and we get the asymptotic condition
\begin{equation}
H_2(p_{a,x}) + H_2(p_{a,z}) < 1.
\end{equation}
The secure asymptotic error rate thresholds zone is shown in
Figure~\ref{fig:security} (it is below the curve).
Note the trade-off between the error rate thresholds $p_{a,z}$ and
$p_{a,x}$. Also note that in the case of $p_{a,z} = p_{a,x}$,
we get the same threshold as in BB84, which is $11\%$.

\begin{figure}
\includegraphics[width=\linewidth]{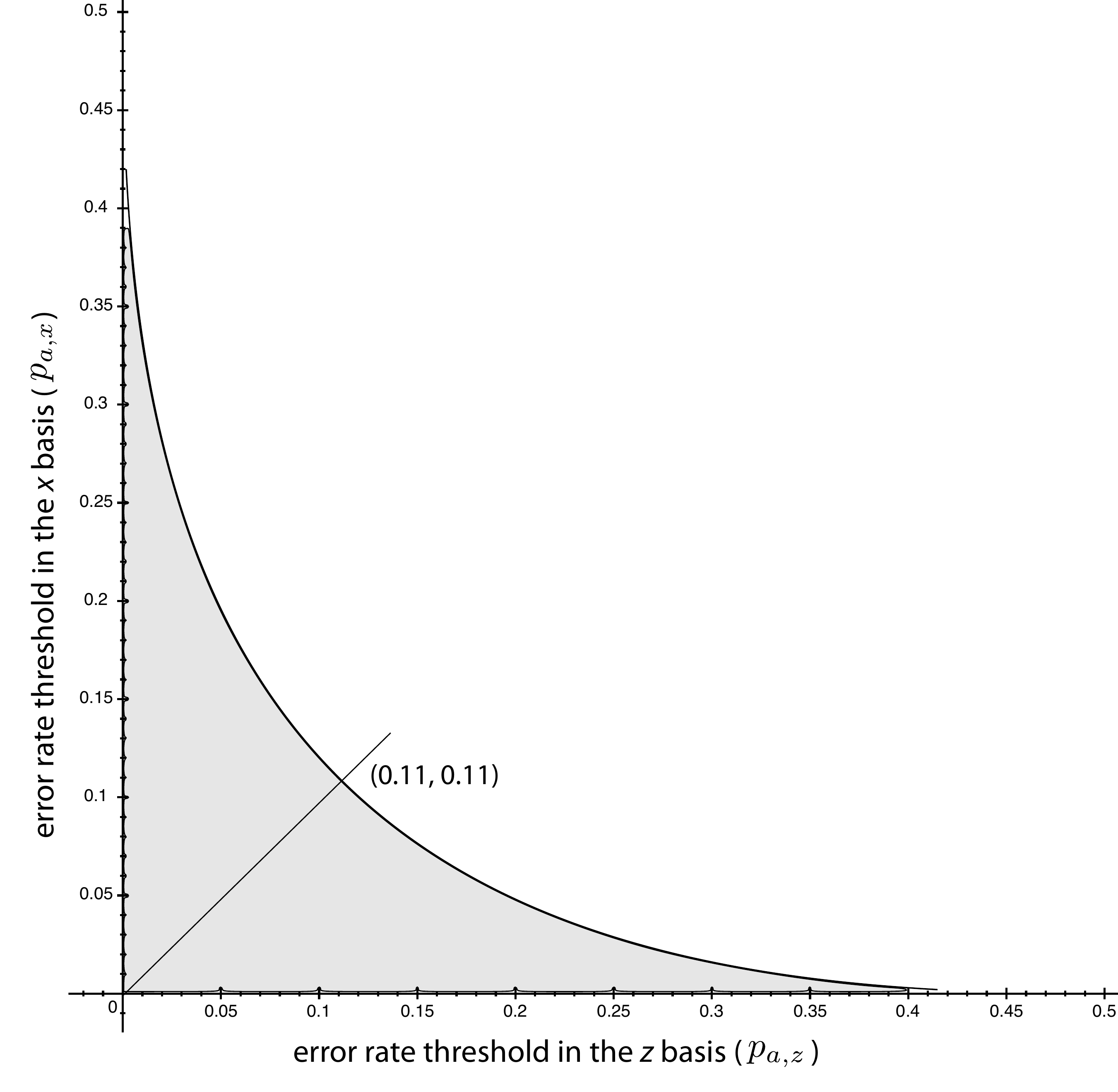}
\caption{\label{fig:security}\textbf{The secure asymptotic error rates
zone for BB84-INFO-$z$} (below the curve)}
\end{figure}

\subsection{\label{subsec_bb84}The Standard BB84 Protocol}
In the standard BB84 protocol, the strings $\bb$ and $\bs$ are chosen
randomly (except that we demand $|\bs| = n$) and independently,
and $N = 2n$. In other words, there are $n$ INFO bits
and $n$ TEST bits (chosen randomly), and for each one of them,
the basis ($z$ or $x$) is chosen randomly and independently.

Formally, in BB84, we choose $N = 2n$, $B = \mathbf{F}^{N}_2$,
and $S_\bb = \{\bs \in \mathbf{F}^{N}_2 \mid |\bs| = n\}$
for all $\bb \in B$. The probability distributions $\pr(\bb)$ and
$\pr(\bs \mid \bb) = \pr(\bs)$ are all uniform---namely,
$\pr(\bb, \bs)$ is identical for all $\bb \in B$ and $\bs \in S_\bb$.

Given the parameter $p_a$ agreed by Alice and Bob,
the testing function $T$ is
\begin{equation}
T(\bi_\sT \oplus \bj_\sT, \bb_\sT, \bs) = 1 ~ \Leftrightarrow ~
|\bi_\sT \oplus \bj_\sT| \le n \cdot p_a.
\end{equation}
Namely, the test passes if and only if the error rate
on the TEST bits is at most $p_a$.

\begin{proposition}\label{prop_bb84_prob}
In the standard BB84 protocol,
for any integer $0 \le t \le \frac{n}{2}$,
\begin{equation}
\pr_\invb\left[ \left( \frac{|\bC_\sI|}{n} \ge \frac{t}{n}
\right) \wedge \left( \bT = 1 \right) \right] =
\pr\left[ \left( \frac{|\bC_\sI|}{n} \ge \frac{t}{n} \right)
\wedge \left( \bT = 1 \right) \right].
\end{equation}
\end{proposition}
\begin{proof}
\begin{eqnarray}
&&\pr_\invb\left[ \left( \frac{|\bC_\sI|}{n} \ge \frac{t}{n}
\right) \wedge \left( \bT = 1 \right) \right] \nonumber \\
&=& \sum_{\bb, \bs} \pr_\invb\left[ \left( \frac{|\bC_\sI|}{n} \ge
\frac{t}{n} \right) \wedge
\left( \bT = 1 \right) \mid \bb, \bs \right] \cdot
\pr(\bb, \bs) \nonumber \\
&=& \sum_{\bb, \bs} \pr\left[ \left( \frac{|\bC_\sI|}{n} \ge
\frac{t}{n} \right) \wedge
\left( \bT = 1 \right) \mid \bb^0, \bs \right] \cdot
\pr(\bb^0, \bs) \nonumber \\
&=& \pr\left[ \left( \frac{|\bC_\sI|}{n} \ge \frac{t}{n} \right)
\wedge \left( \bT = 1 \right) \right]
\end{eqnarray}
(where $\bb^0 \triangleq \bb \oplus \bs$). 
\end{proof}

The security of the standard BB84 protocol is now easily obtained:
\begin{theorem}\label{thmsecurity_bb84}
For any $n > 0$, $p_a > 0$,
and $\epsilon_\secur, \epsilon_\rel > 0$
such that $n (p_a + \epsilon_\secur)$
and $n (p_a + \epsilon_\rel)$ are both integers,
$p_a + \epsilon_\secur \le \frac{1}{2}$,
and $p_a + \epsilon_\rel \le \frac{1}{2}$,
it holds for the standard BB84 protocol that
\begin{eqnarray}
\frac{1}{2} \tr \left| \rho_{\mathrm{ABE}}
- \rho_\sU \otimes \rho_\sE \right|
&\le& e^{-\frac{1}{2} n \epsilon_\rel^2}
+ 2^{n[H_2(p_a + \epsilon_\rel) - r/n]} \nonumber \\
&+& 2m \sqrt{e^{-\frac{1}{2} n \epsilon_\secur^2}
+ 2^{n[H_2(p_a + \epsilon_\secur) - (n-r-m)/n]}},
\end{eqnarray}
where $H_2(x) \triangleq -x \log_2(x) - (1-x) \log_2(1-x)$.
\end{theorem}
\begin{proof}
Using Corollary~\ref{corollary_compos} and
Proposition~\ref{prop_bb84_prob},
we get the following bound for the standard BB84 protocol,
choosing the integer $t \triangleq n (p_a + \epsilon_\secur)$:
\begin{align}
&\frac{1}{2} \tr \left| \rho_{\mathrm{ABE}}
- \rho_\sU \otimes \rho_\sE \right| \nonumber \\
&\le \pr\left[ \left( \bk \ne \bk^\sB \right)
\wedge \left( \bT = 1 \right) \right] \nonumber \\
&+ 2m \sqrt{\pr\left[ \left( \frac{|\bC_\sI|}{n} \ge
\frac{t}{n} \right) \wedge \left( \bT = 1 \right) \right]
+ 2^{n[H_2(t/n) - (n-r-m)/n]}} \nonumber \\
&= \pr\left[ \left( \bk \ne \bk^\sB \right)
\wedge \left( \frac{|\bC_\sT|}{n} \le p_a \right) \right] \nonumber \\
&+ 2m \sqrt{\pr\left[ \left( \frac{|\bC_\sI|}{n} \ge
p_a + \epsilon_\secur \right) \wedge
\left( \frac{|\bC_\sT|}{n} \le p_a \right) \right]
+ 2^{n[H_2(p_a + \epsilon_\secur) - (n-r-m)/n]}}.\label{bb84_proof1}
\end{align}

Because the event $\bk \ne \bk^\sB$
implies that either the error rate on the INFO bits is higher than
$p_a + \epsilon_\rel$ or the error correction step fails
(which, according to Corollary~\ref{corollary_ecc_decoding},
happens with probability at most
$2^{n[H_2(p_a + \epsilon_\rel) - r/n]}$;
see the similar proof of Theorem~\ref{thmsecurity2} for details),
we get:
\begin{eqnarray}
\pr\left[ \left( \bk \ne \bk^\sB \right) \wedge
\left( \frac{|\bC_\sT|}{n} \le p_a \right) \right]
&\le& \pr\left[
\left( \frac{|\bC_\sI|}{n} \ge p_a + \epsilon_\rel \right)
\wedge \left( \frac{|\bC_\sT|}{n} \le p_a \right) \right] \nonumber \\
&+& 2^{n[H_2(p_a + \epsilon_\rel) - r/n]}.\label{bb84_proof2}
\end{eqnarray}

All bits in the protocol are sent in random and independent bases.
Therefore, the random and uniform sampling of the $n$ INFO bits out of
the $N = 2n$ bits does not affect the bases (in the real protocol).
We can thus apply Corollary~\ref{corollary_hoeffding}
to this sampling, and we get
\begin{eqnarray}
\pr\left[ \left( \frac{|\bC_\sI|}{n} \ge p_a + \epsilon_\rel \right)
\wedge \left( \frac{|\bC_\sT|}{n} \le p_a \right) \right] &\le&
e^{-\frac{1}{2} n\epsilon_\rel^2},\label{bb84_proof3a} \\
\pr\left[ \left( \frac{|\bC_\sI|}{n} \ge p_a + \epsilon_\secur \right)
\wedge \left( \frac{|\bC_\sT|}{n} \le p_a \right) \right] &\le&
e^{-\frac{1}{2} n\epsilon_\secur^2}.\label{bb84_proof3b}
\end{eqnarray}

Combining Equations~\eqref{bb84_proof1}--\eqref{bb84_proof3b},
we get
\begin{eqnarray}
\frac{1}{2} \tr \left| \rho_{\mathrm{ABE}}
- \rho_\sU \otimes \rho_\sE \right|
&\le& e^{-\frac{1}{2} n \epsilon_\rel^2}
+ 2^{n[H_2(p_a + \epsilon_\rel) - r/n]} \nonumber \\
&+& 2m \sqrt{e^{-\frac{1}{2} n \epsilon_\secur^2}
+ 2^{n[H_2(p_a + \epsilon_\secur) - (n-r-m)/n]}}.
\end{eqnarray}
\end{proof}
We point out that this is a finite-key bound, matching pretty closely
the state-of-the-art for security of BB84 (e.g.,~\cite{TLGR12,TL17}),
except the superfluous ``$m$'' factor (resulting from Proposition~\ref{probbound})
that we intend to remove in a future version of this paper.

For finding the asymptotic error rate,
we note that this bound is exponentially small in $n$
under the two following conditions:
\begin{eqnarray}
H_2(p_a + \epsilon_\rel) &<& \frac{r}{n}, \\
H_2(p_a + \epsilon_\secur) &<& \frac{n-r-m}{n},
\end{eqnarray}
which are equivalent to:
\begin{eqnarray}
H_2(p_a + \epsilon_\rel) &<& \frac{r}{n}, \\
H_2(p_a + \epsilon_\secur) + H_2(p_a + \epsilon_\rel)
&<& \frac{n-m}{n} = 1 - \frac{m}{n}.
\end{eqnarray}
We can thus obtain the following upper bound on the bit-rate:
\begin{equation}
R_{\mathrm{secret}} \triangleq \frac{m}{n}
< 1 - H_2(p_a + \epsilon_\secur)
- H_2(p_a + \epsilon_\rel).
\end{equation}

To get the asymptotic error rate threshold, we require
$R_{\mathrm{secret}} > 0$, and we get the asymptotic condition
\begin{equation}
2 H_2(p_a) < 1.
\end{equation}
This condition gives an asymptotic error rate threshold of $11\%$.

\subsection{\label{subsec_eff_bb84}The ``Efficient BB84'' Protocol}
In the ``efficient BB84'' protocol (suggested by~\cite{LCA05}),
the bitstring $\bb$ is chosen probabilistically,
but {\em not uniformly}:
each qubit is sent in the $z$ basis with probability $p$
(and in the $x$ basis with probability $1 - p$),
where $0 < p \le \frac{1}{2}$. Then, the bitstring $\bs$ is chosen
such that there are $n_z$ TEST-Z bits and $n_x$ TEST-X bits.
In other words, as in BB84-INFO-$z$, the strings $\bb$ and $\bs$
together define a random partition of the set of indexes
$\lbrace 1, 2, \ldots, N \rbrace$ into three disjoint sets:
\begin{itemize}
\item $\sI$ (INFO bits, where $s_j = 1$) of size $n$. However,
{\em unlike} BB84-INFO-$z$, this set consists of both $z$ qubits
and $x$ qubits; therefore, it can be divided to two disjoint subsets:
\begin{itemize}
\item $\sI_\sZ$ (INFO-Z bits, where $s_j = 1$ and $b_j = 0$); and
\item $\sI_\sX$ (INFO-X bits, where $s_j = 1$ and $b_j = 1$).
\end{itemize}
\item $\sT_\sZ$ (TEST-Z bits, where $s_j = 0$ and $b_j = 0$)
of size $n_z$; and
\item $\sT_\sX$ (TEST-X bits, where $s_j = 0$ and $b_j = 1$)
of size $n_x$.
\end{itemize}

Formally, in ``efficient BB84'', Alice and Bob agree on parameters
$n, n_z, n_x$ (such that $N = n + n_z + n_x$)
and on a parameter $0 < p \le \frac{1}{2}$,
and we choose $B = \mathbf{F}^{N}_2$
and $S_\bb = \{\bs \in \mathbf{F}^{N}_2 \mid (|\bs| = n) \wedge
(|\overline{\bs} \wedge \bb| = n_x)\}$ for all $\bb \in B$
(namely, it is required that there are $n$ INFO bits,
$n_z$ TEST-Z bits, and $n_x$ TEST-X bits).
This time, the probability distribution
$\pr(\bb)$ is {\em not} uniform: it holds that
$\pr(\bb) = (1 - p)^{|\bb|} \cdot p^{N - |\bb|}$,
because the probability of each bit to be in the $x$ basis is $1 - p$.
On the other hand, the probability distribution $\pr(\bs \mid \bb)$
is uniform.

\begin{remark}\label{eff_bb84_remark}
A subtle point is that for some values $\bb \in \mathbf{F}^{N}_2$
(for example, for $\bb = 00\ldots 0$), the set $S_\bb$ is empty:
no $\bs$ can be agreed by Alice and Bob for such values of $\bb$.
In that case, as assumed in~\cite[Section 4.3]{LCA05},
the protocol aborts,
and different values of $\bb$ and $\bs$ are randomly chosen;
this is equivalent to assuming that Alice is not allowed
to choose those values of $\bb$.
Therefore, to be more precise, we must re-define
\begin{equation}
B = \{\bb \in \mathbf{F}^{N}_2 \mid S_\bb \ne \emptyset\}
= \{\bb \in \mathbf{F}^{N}_2 \mid (|\bb| \ge n_x) \wedge
(|\overline{\bb}| \ge n_z)\},
\end{equation}
and we must normalize the probabilities by defining
$\pr_0(\bb) \triangleq (1 - p)^{|\bb|} \cdot p^{N - |\bb|}$
(the original probability of each $\bb$),
$N_p \triangleq \sum_{\bb \in B} \pr_0(\bb)$
(the sum of all the original probabilities
for all the allowed values of $\bb \in B$),
and then the real probability of each $\bb \in B$ is
\begin{equation}
\pr(\bb) = \frac{\pr_0(\bb)}{N_p}
= \frac{(1 - p)^{|\bb|} \cdot p^{N - |\bb|}}{N_p}.
\end{equation}
This guarantees that the sum of probabilities of all the allowed values
$\bb \in B$ is $1$.
\end{remark}

Alice and Bob also agree on an error rate threshold, $p_a$
(applied {\em both} to the TEST-Z bits and to the TEST-X bits).
The testing function $T$ is defined as follows:
\begin{equation}
T(\bi_\sT \oplus \bj_\sT, \bb_\sT, \bs) = 1 ~ \Leftrightarrow ~
\left( |\bi_{\sT_\sZ} \oplus \bj_{\sT_\sZ}| \le n_z \cdot p_a \right) \wedge
\left( |\bi_{\sT_\sX} \oplus \bj_{\sT_\sX}| \le n_x \cdot p_a \right).
\end{equation}
Namely, the test passes if and only if
the error rate on the TEST-Z bits is at most $p_a$ {\em and}
the error rate on the TEST-X bits is at most $p_a$.

In this security proof, instead of analyzing all the INFO bits
together, we analyze the INFO-Z and the INFO-X bits separately.
We define the following random variables:
\begin{itemize}
\item $\bC_{\sI_\sZ}$ and $\bC_{\sI_\sX}$ are the random variables
corresponding to the error strings on the INFO-Z bits and
on the INFO-X bits, respectively.
\item $N_{\sI_\sZ}$ and $N_{\sI_\sX}$ are random variables equal
to the numbers of INFO-Z and INFO-X bits, respectively.
(We note that the parameters $n, n_z, n_x$ are deterministically chosen
by Alice and Bob, while $N_{\sI_\sZ}$ and $N_{\sI_\sX}$ are determined by
the probabilistic choice of $\bb$. We also note that, necessarily,
$n = N_{\sI_\sZ} + N_{\sI_\sX}$.)
\end{itemize}

\begin{proposition}\label{eff_bb84_probs}
For any $\epsilon > 0$,
\begin{eqnarray}
\pr\left[ \left( \frac{|\bC_\sI|}{n} \ge p_a + \epsilon \right)
\wedge \left( \bT = 1 \right) \right]
&\le& \pr\left[ \left( \frac{|\bC_{\sI_\sZ}|}{N_{\sI_\sZ}} \ge p_a +
\epsilon \right) \wedge \left( \bT = 1 \right) \right] \nonumber \\
&+& \pr\left[ \left( \frac{|\bC_{\sI_\sX}|}{N_{\sI_\sX}} \ge p_a + \epsilon
\right) \wedge \left( \bT = 1 \right) \right].
\label{eff_bb84_probs_real}
\end{eqnarray}

Equation~\eqref{eff_bb84_probs_real} similarly applies to the
hypothetical ``inverted-INFO-basis'' protocol, too
(namely, it applies even if $\pr$ is replaced by $\pr_\invb$).
\end{proposition}
\begin{proof}
We observe that if the error rate on all the INFO bits together
is at least $p_a + \epsilon$,
then at least one of the error rates (on the INFO-Z bits or
on the INFO-X bits) must be at least $p_a + \epsilon$.
(Equivalently, if both error rates on the INFO-Z bits and
on the INFO-X bits are lower than $p_a + \epsilon$, then the error rate
on the INFO bits is necessarily lower than $p_a + \epsilon$.) Namely,
\begin{equation}
\left( \frac{|\bC_\sI|}{n} \ge p_a + \epsilon \right) ~\Rightarrow~
\left( \frac{|\bC_{\sI_\sZ}|}{N_{\sI_\sZ}} \ge p_a + \epsilon \right) \vee
\left( \frac{|\bC_{\sI_\sX}|}{N_{\sI_\sX}} \ge p_a + \epsilon \right).
\end{equation}
In particular, the corresponding probabilities satisfy
\begin{eqnarray}
\pr\left[ \left( \frac{|\bC_\sI|}{n} \ge p_a + \epsilon \right)
\wedge \left( \bT = 1 \right) \right]
&\le& \pr\left[ \left( \frac{|\bC_{\sI_\sZ}|}{N_{\sI_\sZ}} \ge p_a + \epsilon
\right) \wedge \left( \bT = 1 \right) \right] \nonumber \\
&+& \pr\left[ \left( \frac{|\bC_{\sI_\sX}|}{N_{\sI_\sX}} \ge p_a + \epsilon
\right) \wedge \left( \bT = 1 \right) \right].
\end{eqnarray}
This result applies both to the real protocol
and to the hypothetical ``inverted-INFO-basis'' protocol. 
\end{proof}

\begin{proposition}\label{eff_bb84_nI}
For any $\epsilon > 0$ and $\delta > 0$,
\begin{eqnarray}
&&\pr\left[ \left( \frac{|\bC_{\sI_\sZ}|}{N_{\sI_\sZ}} \ge p_a + \epsilon
\right) \wedge \left( \bT = 1 \right) \right] \nonumber \\
&\le& \pr\left( N_{\sI_\sZ} \le \delta \right)
+ \max_{\delta \le t_z \le n}
\pr\left[ \left( \frac{|\bC_{\sI_\sZ}|}{t_z} \ge p_a + \epsilon
\right) \wedge \left( \bT = 1 \right) \mid N_{\sI_\sZ} = t_z \right],
\label{eff_bb84_nIZ}
\end{eqnarray}
and
\begin{eqnarray}
&&\pr\left[ \left( \frac{|\bC_{\sI_\sX}|}{N_{\sI_\sX}} \ge p_a + \epsilon
\right) \wedge \left( \bT = 1 \right) \right] \nonumber \\
&\le& \pr\left( N_{\sI_\sX} \le \delta \right)
+ \max_{\delta \le t_x \le n}
\pr\left[ \left( \frac{|\bC_{\sI_\sX}|}{t_x} \ge p_a + \epsilon
\right) \wedge \left( \bT = 1 \right) \mid N_{\sI_\sX} = t_x \right].
\label{eff_bb84_nIX}
\end{eqnarray}

Equations~\eqref{eff_bb84_nIZ}--\eqref{eff_bb84_nIX} similarly
apply to the hypothetical ``inverted-INFO-basis'' protocol, too
(namely, they apply even if $\pr$ is replaced by $\pr_\invb$).
\end{proposition}
\begin{proof}
First, we prove Equation~\eqref{eff_bb84_nIZ}:
\begin{eqnarray}
&&\pr\left[ \left( \frac{|\bC_{\sI_\sZ}|}{N_{\sI_\sZ}} \ge p_a + \epsilon
\right) \wedge \left( \bT = 1 \right) \right] \nonumber \\
&=& \sum_{t_z} \pr\left[ \left( \frac{|\bC_{\sI_\sZ}|}{t_z} \ge p_a +
\epsilon \right) \wedge \left( \bT = 1 \right)
\mid N_{\sI_\sZ} = t_z \right]
\cdot \pr\left( N_{\sI_\sZ} = t_z \right) \nonumber \\
&=& \sum_{t_z < \delta} \pr\left[ \left( \frac{|\bC_{\sI_\sZ}|}{t_z} \ge
p_a + \epsilon \right) \wedge \left( \bT = 1 \right)
\mid N_{\sI_\sZ} = t_z \right]
\cdot \pr\left( N_{\sI_\sZ} = t_z \right) \nonumber \\
&+& \sum_{\delta \le t_z \le n} \pr\left[ \left(
\frac{|\bC_{\sI_\sZ}|}{t_z} \ge p_a + \epsilon \right) \wedge
\left( \bT = 1 \right) \mid N_{\sI_\sZ} = t_z \right]
\cdot \pr\left( N_{\sI_\sZ} = t_z \right) \nonumber \\
&\le& \sum_{t_z \le \delta} \pr\left( N_{\sI_\sZ} = t_z \right)
\nonumber \\
&+& \max_{\delta \le t_z \le n} \pr\left[ \left(
\frac{|\bC_{\sI_\sZ}|}{t_z} \ge p_a + \epsilon \right) \wedge
\left( \bT = 1 \right) \mid N_{\sI_\sZ} = t_z \right]
\cdot \sum_{\delta \le t_z \le n} \pr\left( N_{\sI_\sZ} = t_z \right)
\nonumber \\
&\le& \pr\left( N_{\sI_\sZ} \le \delta \right)
+ \max_{\delta \le t_z \le n} \pr\left[ \left(
\frac{|\bC_{\sI_\sZ}|}{t_z} \ge p_a + \epsilon \right) \wedge
\left( \bT = 1 \right) \mid N_{\sI_\sZ} = t_z \right].
\end{eqnarray}

The proof of Equation~\eqref{eff_bb84_nIX} is similar.
Both proofs apply both to the real protocol and
to the ``inverted-INFO-basis'' protocol. 
\end{proof}

\begin{theorem}\label{thmsecurity_eff_bb84}
For any $0 < p \le \frac{1}{2}$, $n > 0$, $0 < n_z < \frac{p N}{2}$,
$0 < n_x < \frac{(1 - p) N}{2}$, $p_a > 0$,
and $\epsilon_\secur, \epsilon_\rel > 0$
such that $t \triangleq n (p_a + \epsilon_\secur)$
and $t_\rel \triangleq n (p_a + \epsilon_\rel)$ are both integers,
$p_a + \epsilon_\secur \le \frac{1}{2}$,
and $p_a + \epsilon_\rel \le \frac{1}{2}$,
it holds for the ``efficient BB84'' protocol that
\begin{eqnarray}
\frac{1}{2} \tr \left| \rho_{\mathrm{ABE}}
- \rho_\sU \otimes \rho_\sE \right|
&\le& e^{-\frac{1}{2} N p^2}
+ e^{-2 \left( \frac{n_z}{n + n_z} \right)^2
\left( \frac{p N}{2} - n_z \right) \epsilon_\rel^2} \nonumber \\
&+& e^{-\frac{1}{2} N (1 - p)^2}
+ e^{-2 \left( \frac{n_x}{n + n_x} \right)^2
\left( \frac{(1 - p) N}{2} - n_x \right) \epsilon_\rel^2}
+ 2^{n[H_2(p_a + \epsilon_\rel) - r/n]} \nonumber \\
&+& 2m \sqrt{e^{-\frac{1}{2} N p^2}
+ e^{-2 \left( \frac{n_x}{n + n_x} \right)^2
\left( \frac{p N}{2} - n_z \right) \epsilon_\secur^2} +
e^{-\frac{1}{2} N (1 - p)^2}
+ e^{-2 \left( \frac{n_z}{n + n_z} \right)^2
\left( \frac{(1 - p) N}{2} - n_x \right) \epsilon_\secur^2} +} \nonumber \\
&& \overline{2^{n[H_2(p_a + \epsilon_\secur) - (n-r-m)/n]}},
\end{eqnarray}
where $H_2(x) \triangleq -x \log_2(x) - (1-x) \log_2(1-x)$.
\end{theorem}
\begin{proof}
By using Corollary~\ref{corollary_compos} and
Proposition~\ref{eff_bb84_probs} (and also Corollary~\ref{corollary_ecc_decoding},
similarly to the proof of Theorem~\ref{thmsecurity2}), we get the following bound:
\begin{eqnarray}
&&\frac{1}{2} \tr \left| \rho_{\mathrm{ABE}}
- \rho_\sU \otimes \rho_\sE \right|
\nonumber \\
&\le& \pr\left[ \left( \bk \ne \bk^\sB \right) \wedge
\left( \bT = 1 \right) \right] \nonumber \\
&+& 2m \sqrt{\pr_\invb\left[ \left( \frac{|\bC_\sI|}{n} \ge
\frac{t}{n} \right) \wedge \left( \bT = 1 \right) \right]
+ 2^{n[H_2(t/n) - (n-r-m)/n]}} \nonumber \\
&\le& \pr\left[ \left( \frac{|\bC_\sI|}{n} \ge
p_a + \epsilon_\rel \right) \wedge
\left( \bT = 1 \right) \right]
+ 2^{n[H_2(p_a + \epsilon_\rel) - r/n]} \nonumber \\
&+& 2m \sqrt{\pr_\invb\left[ \left( \frac{|\bC_\sI|}{n} \ge
p_a + \epsilon_\secur \right) \wedge \left( \bT = 1 \right) \right]
+ 2^{n[H_2(p_a + \epsilon_\secur) - (n-r-m)/n]}} \nonumber \\
&\le& \pr\left[ \left( \frac{|\bC_{\sI_\sZ}|}{N_{\sI_\sZ}} \ge
p_a + \epsilon_\rel \right) \wedge
\left( \bT = 1 \right) \right] \nonumber \\
&+& \pr\left[ \left( \frac{|\bC_{\sI_\sX}|}{N_{\sI_\sX}} \ge
p_a + \epsilon_\rel \right) \wedge
\left( \bT = 1 \right) \right]
+ 2^{n[H_2(p_a + \epsilon_\rel) - r/n]} \nonumber \\
&+& 2m \sqrt{\pr_\invb\left[ \left( \frac{|\bC_{\sI_\sZ}|}{N_{\sI_\sZ}} \ge
p_a + \epsilon_\secur \right) \wedge
\left( \bT = 1 \right) \right] +} \nonumber \\
&&\overline{\pr_\invb\left[ \left( \frac{|\bC_{\sI_\sX}|}{N_{\sI_\sX}} \ge
p_a + \epsilon_\secur \right) \wedge
\left( \bT = 1 \right) \right]
+ 2^{n[H_2(p_a + \epsilon_\secur) - (n-r-m)/n]}}.\label{eff_bb84_bound_0}
\end{eqnarray}

Proposition~\ref{eff_bb84_nI} and the definition of $T$
give us the following bounds:
\begin{eqnarray}
&&\pr\left[ \left( \frac{|\bC_{\sI_\sZ}|}{N_{\sI_\sZ}}
\ge p_a + \epsilon_\rel \right) \wedge \left( \bT = 1 \right) \right]
\le \pr\left( N_{\sI_\sZ} \le \frac{p N}{2} - n_z \right) \nonumber \\
&+& \max_{\frac{p N}{2} - n_z \le t_z \le n}
\pr\left[ \left( \frac{|\bC_{\sI_\sZ}|}{t_z} \ge p_a + \epsilon_\rel
\right) \wedge \left( \frac{|\bC_{\sT_\sZ}|}{n_z} \le p_a \right)
\mid N_{\sI_\sZ} = t_z \right],\label{eff_bb84_ineq_iz_tz} \\
&&\pr\left[ \left( \frac{|\bC_{\sI_\sX}|}{N_{\sI_\sX}}
\ge p_a + \epsilon_\rel \right) \wedge \left( \bT = 1 \right) \right]
\le \pr\left( N_{\sI_\sX} \le \frac{(1 - p) N}{2} - n_x \right)
\nonumber \\
&+& \max_{\frac{(1 - p) N}{2} - n_x \le t_x \le n}
\pr\left[ \left( \frac{|\bC_{\sI_\sX}|}{t_x} \ge p_a + \epsilon_\rel
\right) \wedge \left( \frac{|\bC_{\sT_\sX}|}{n_x} \le p_a \right)
\mid N_{\sI_\sX} = t_x \right],\label{eff_bb84_ineq_ix_tx} \\
&&\pr_\invb\left[ \left( \frac{|\bC_{\sI_\sZ}|}{N_{\sI_\sZ}} \ge p_a +
\epsilon_\secur \right) \wedge \left( \bT = 1 \right) \right]
\nonumber \\
&\le& \pr_\invb\left( N_{\sI_\sZ} \le \frac{p N}{2} - n_z \right)
+ \max_{\frac{p N}{2} - n_z \le t_z \le n} \nonumber \\
&&\pr_\invb\left[ \left( \textstyle \frac{|\bC_{\sI_\sZ}|}{t_z} \ge p_a +
\epsilon_\secur \right) \wedge
\left( \textstyle \frac{|\bC_{\sT_\sX}|}{n_x} \le p_a \right)
\mid N_{\sI_\sZ} = t_z \right],\label{eff_bb84_ineq_iz_tx} \\
&&\pr_\invb\left[ \left( \frac{|\bC_{\sI_\sX}|}{N_{\sI_\sX}} \ge p_a +
\epsilon_\secur \right) \wedge \left( \bT = 1 \right) \right]
\nonumber \\
&\le& \pr_\invb\left( N_{\sI_\sX} \le \frac{(1 - p) N}{2} - n_x \right)
+ \max_{\frac{(1 - p) N}{2} - n_x \le t_x \le n} \nonumber \\
&&\pr_\invb\left[ \left( \textstyle \frac{|\bC_{\sI_\sX}|}{t_x} \ge p_a +
\epsilon_\secur \right) \wedge
\left( \textstyle \frac{|\bC_{\sT_\sZ}|}{n_z} \le p_a \right)
\mid N_{\sI_\sX} = t_x \right].\label{eff_bb84_ineq_ix_tz}
\end{eqnarray}

For each one of
Equations~\eqref{eff_bb84_ineq_iz_tz}--\eqref{eff_bb84_ineq_ix_tz},
we need to upper-bound two probabilities.
For bounding the first set of probabilities,
we use the results of Corollary~\ref{corollary_hoeffding_indep}:
\begin{eqnarray}
\pr\left( |\bb| \le \frac{(1 - p) N}{2} \right) &\le&
e^{-\frac{1}{2} N (1 - p)^2}, \\
\pr\left( |\overline{\bb}| \le \frac{p N}{2} \right) &\le&
e^{-\frac{1}{2} N p^2}.
\end{eqnarray}
We notice that $|\bb| = N_{\sI_\sX} + n_x$ and
$|\overline{\bb}| = N_{\sI_\sZ} + n_z$; therefore,
\begin{eqnarray}
\pr\left( N_{\sI_\sX} \le \frac{(1 - p) N}{2} - n_x \right)
&\le& e^{-\frac{1}{2} N (1 - p)^2},\label{eff_bb84_bound_b1}\\
\pr\left( N_{\sI_\sZ} \le \frac{p N}{2} - n_z \right)
&\le& e^{-\frac{1}{2} N p^2},\label{eff_bb84_bound_b2}\\
\pr_\invb\left( N_{\sI_\sX} \le \frac{(1 - p) N}{2} - n_x \right)
&\le& e^{-\frac{1}{2} N (1 - p)^2},\label{eff_bb84_bound_b3}\\
\pr_\invb\left( N_{\sI_\sZ} \le \frac{p N}{2} - n_z \right)
&\le& e^{-\frac{1}{2} N p^2}.\label{eff_bb84_bound_b4}
\end{eqnarray}

For bounding the second set of probabilities,
{\em given} specific values of $N_{\sI_\sZ} = t_z$
and $N_{\sI_\sX} = t_x$,
we use Corollary~\ref{corollary_hoeffding}:

In the real protocol, the
INFO-Z and TEST-Z bits are sent and measured in the $z$ basis, while
the INFO-X and TEST-X bits are sent and measured in the $x$ basis.
Therefore, the random and uniform sampling of the $t_z$ INFO-Z bits
out of the $t_z + n_z$ bits sent in the $z$ basis
(assuming that the INFO-X and TEST-X bits have already been chosen)
does not affect the bases in the real protocol;
similarly, the random and uniform sampling of the $t_x$ INFO-X bits
out of the $t_x + n_x$ bits sent in the $x$ basis
(assuming that the INFO-Z and TEST-Z bits have already been chosen)
does not affect the bases in the real protocol.
We note that these samplings {\em are} uniform,
because the probability $\pr(\bs \mid \bb)$ is uniform
for all the allowed values of $\bb$ and $\bs$.
This means that we can apply Corollary~\ref{corollary_hoeffding}
to both of these samplings, and we get
\begin{eqnarray}
\pr\left[ \left( \frac{|\bC_{\sI_\sZ}|}{t_z} \ge
p_a + \epsilon_\rel \right) \wedge
\left( \frac{|\bC_{\sT_\sZ}|}{n_z} \le p_a \right)
\mid N_{\sI_\sZ} = t_z \right]
&\le& e^{-2 \left( \frac{n_z}{t_z + n_z} \right)^2
t_z \epsilon_\rel^2}, \\
\pr\left[ \left( \frac{|\bC_{\sI_\sX}|}{t_x} \ge
p_a + \epsilon_\rel \right) \wedge
\left( \frac{|\bC_{\sT_\sX}|}{n_x} \le p_a \right)
\mid N_{\sI_\sX} = t_x \right]
&\le& e^{-2 \left( \frac{n_x}{t_x + n_x} \right)^2
t_x \epsilon_\rel^2}.
\end{eqnarray}
Maximizing over $t_z$ and $t_x$, we get:
\begin{eqnarray}
&&\max_{\frac{p N}{2} - n_z \le t_z \le n}
\pr\left[ \left( \frac{|\bC_{\sI_\sZ}|}{t_z} \ge
p_a + \epsilon_\rel \right) \wedge
\left( \frac{|\bC_{\sT_\sZ}|}{n_z} \le p_a \right)
\mid N_{\sI_\sZ} = t_z \right] \nonumber \\
&\le& e^{-2 \left( \frac{n_z}{n + n_z} \right)^2
\left( \frac{p N}{2} - n_z \right) \epsilon_\rel^2},\label{eff_bb84_bound_m1}\\
&&\max_{\frac{(1 - p) N}{2} - n_x \le t_x \le n}
\pr\left[ \left( \frac{|\bC_{\sI_\sX}|}{t_x} \ge
p_a + \epsilon_\rel \right) \wedge
\left( \frac{|\bC_{\sT_\sX}|}{n_x} \le p_a \right)
\mid N_{\sI_\sX} = t_x \right] \nonumber \\
&\le& e^{-2 \left( \frac{n_x}{n + n_x} \right)^2
\left( \frac{(1 - p) N}{2} - n_x \right)
\epsilon_\rel^2}.\label{eff_bb84_bound_m2}
\end{eqnarray}

In the hypothetical ``inverted-INFO-basis'' protocol, the
INFO-X and TEST-Z bits are sent and measured in the $z$ basis, while
the INFO-Z and TEST-X bits are sent and measured in the $x$ basis.
Therefore, the random and uniform sampling of the $t_x$ INFO-X bits
out of the $t_x + n_z$ bits sent in the $z$ basis
(assuming that the INFO-Z and TEST-X bits have already been chosen)
does not affect the bases in the hypothetical protocol;
similarly, the random and uniform sampling of the $t_z$ INFO-Z bits
out of the $t_z + n_x$ bits sent in the $x$ basis
(assuming that the INFO-X and TEST-Z bits have already been chosen)
does not affect the bases in the hypothetical protocol.
We note that these samplings {\em are} uniform,
because the probability $\pr(\bb)$ depends only on $|\bb|$
and is invariant to permutations.
This means that we can apply Corollary~\ref{corollary_hoeffding}
to both of these samplings, and we get
\begin{eqnarray}
&&\pr_\invb\left[ \left( \frac{|\bC_{\sI_\sZ}|}{t_z} \ge
p_a + \epsilon_\secur \right) \wedge
\left( \frac{|\bC_{\sT_\sX}|}{n_x} \le p_a \right)
\mid N_{\sI_\sZ} = t_z \right] \nonumber \\
&\le& e^{-2 \left( \frac{n_x}{t_z + n_x} \right)^2
t_z \epsilon_\secur^2}, \\
&&\pr_\invb\left[ \left( \frac{|\bC_{\sI_\sX}|}{t_x} \ge
p_a + \epsilon_\secur \right) \wedge
\left( \frac{|\bC_{\sT_\sZ}|}{n_z} \le p_a \right)
\mid N_{\sI_\sX} = t_x \right] \nonumber \\
&\le& e^{-2 \left( \frac{n_z}{t_x + n_z} \right)^2
t_x \epsilon_\secur^2}.
\end{eqnarray}
Maximizing over $t_z$ and $t_x$, we get:
\begin{eqnarray}
&&\max_{\frac{p N}{2} - n_z \le t_z \le n}
\pr_\invb\left[ \left( \frac{|\bC_{\sI_\sZ}|}{t_z} \ge
p_a + \epsilon_\secur \right) \wedge
\left( \frac{|\bC_{\sT_\sX}|}{n_x} \le p_a \right)
\mid N_{\sI_\sZ} = t_z \right] \nonumber \\
&\le& e^{-2 \left( \frac{n_x}{n + n_x} \right)^2
\left( \frac{p N}{2} - n_z \right) \epsilon_\secur^2},\label{eff_bb84_bound_m3}\\
&&\max_{\frac{(1 - p) N}{2} - n_x \le t_x \le n}
\pr_\invb\left[ \left( \frac{|\bC_{\sI_\sX}|}{t_x} \ge
p_a + \epsilon_\secur \right) \wedge
\left( \frac{|\bC_{\sT_\sZ}|}{n_z} \le p_a \right)
\mid N_{\sI_\sX} = t_x \right] \nonumber \\
&\le& e^{-2 \left( \frac{n_z}{n + n_z} \right)^2
\left( \frac{(1 - p) N}{2} - n_x \right) \epsilon_\secur^2}.\label{eff_bb84_bound_m4}
\end{eqnarray}

Substituting Equations~\eqref{eff_bb84_bound_b1}--\eqref{eff_bb84_bound_b4},
\eqref{eff_bb84_bound_m1}--\eqref{eff_bb84_bound_m2},
and~\eqref{eff_bb84_bound_m3}--\eqref{eff_bb84_bound_m4}
into Equations~\eqref{eff_bb84_ineq_iz_tz}--\eqref{eff_bb84_ineq_ix_tz},
and substituting the results into Equation~\eqref{eff_bb84_bound_0},
we get the desired bound:
\begin{eqnarray}
\frac{1}{2} \tr \left| \rho_{\mathrm{ABE}}
- \rho_\sU \otimes \rho_\sE \right|
&\le& e^{-\frac{1}{2} N p^2}
+ e^{-2 \left( \frac{n_z}{n + n_z} \right)^2
\left( \frac{p N}{2} - n_z \right) \epsilon_\rel^2} \nonumber \\
&+& e^{-\frac{1}{2} N (1 - p)^2}
+ e^{-2 \left( \frac{n_x}{n + n_x} \right)^2
\left( \frac{(1 - p) N}{2} - n_x \right) \epsilon_\rel^2}
+ 2^{n[H_2(p_a + \epsilon_\rel) - r/n]} \nonumber \\
&+& 2m \sqrt{e^{-\frac{1}{2} N p^2}
+ e^{-2 \left( \frac{n_x}{n + n_x} \right)^2
\left( \frac{p N}{2} - n_z \right) \epsilon_\secur^2} +
e^{-\frac{1}{2} N (1 - p)^2}
+ e^{-2 \left( \frac{n_z}{n + n_z} \right)^2
\left( \frac{(1 - p) N}{2} - n_x \right) \epsilon_\secur^2} +} \nonumber \\
&& \overline{2^{n[H_2(p_a + \epsilon_\secur) - (n-r-m)/n]}}.
\end{eqnarray}
\end{proof}

Similarly to the standard BB84 protocol (see Subsection~\ref{subsec_bb84}),
we can obtain the following upper bound on the bit-rate:
\begin{equation}
R_{\mathrm{secret}} \triangleq \frac{m}{n}
< 1 - H_2(p_a + \epsilon_\secur)
- H_2(p_a + \epsilon_\rel).
\end{equation}

To get the asymptotic error rate threshold, we require
$R_{\mathrm{secret}} > 0$, and we get the asymptotic condition $2 H_2(p_a) < 1$.
This condition gives an asymptotic error rate threshold of $11\%$.

\subsection{\label{subsec_modif_eff_bb84}The ``Modified
Efficient BB84'' Protocol}
A relatively minor property of the definition of the ``efficient BB84''
protocol in~\cite{LCA05} (and in Subsection~\ref{subsec_eff_bb84})
makes both the security bound and the security proof
pretty complicated.
In this Subsection, we describe a modified protocol that has
an easier security proof. The only modification in this protocol
is setting the number of INFO-Z and INFO-X bits to be fixed,
rather than letting them vary probabilistically.
This change simplifies the description of the protocol,
because it is no longer needed to set the probability $p$
and to treat illegal choices of $\bb, \bs$
(see Remark~\ref{eff_bb84_remark});
and it also simplifies the security proof,
because it is no longer needed to probabilistically analyze
the numbers of INFO-Z and INFO-X bits
(as done in Subsection~\ref{subsec_eff_bb84}).

In the ``modified efficient BB84'' protocol,
the strings $\bb$ and $\bs$ together define
a random partition of the set of indexes
$\lbrace 1, 2, \ldots, N \rbrace$ into four disjoint sets:
\begin{itemize}
\item $\sI_\sZ$ (INFO-Z bits, where $s_j = 1$ and $b_j = 0$)
of size $t_z$;
\item $\sI_\sX$ (INFO-X bits, where $s_j = 1$ and $b_j = 1$)
of size $t_x$;
\item $\sT_\sZ$ (TEST-Z bits, where $s_j = 0$ and $b_j = 0$)
of size $n_z$; and
\item $\sT_\sX$ (TEST-X bits, where $s_j = 0$ and $b_j = 1$)
of size $n_x$.
\end{itemize}

Formally, in ``modified efficient BB84'',
Alice and Bob agree on parameters $t_z, t_x, n_z, n_x$
(such that $N = n + n_z + n_x$ and $n = t_z + t_x$),
and we choose $B = \{\bb \in \mathbf{F}^{N}_2 \mid |\bb| = t_x + n_x\}$
and $S_\bb = \{\bs \in \mathbf{F}^{N}_2 \mid (|\bs| = n) \wedge
(|\overline{\bs} \wedge \bb| = n_x)\}$ for all $\bb \in B$
(namely, it is required that there are $t_z$ INFO-Z bits,
$t_x$ INFO-X bits, $n_z$ TEST-Z bits, and $n_x$ TEST-X bits).
The probability distributions $\pr(\bb)$ and $\pr(\bs \mid \bb)$
are uniform (because $|\bb|$, which is the only parameter that
affects $\pr(\bb)$ in Subsection~\ref{subsec_eff_bb84},
is fixed in the modified protocol).

Alice and Bob also agree on an error rate threshold, $p_a$
(applied {\em both} to the TEST-Z bits and to the TEST-X bits).
The testing function $T$ is defined as follows:
\begin{equation}
T(\bi_\sT \oplus \bj_\sT, \bb_\sT, \bs) = 1 ~ \Leftrightarrow ~
\left( |\bi_{\sT_\sZ} \oplus \bj_{\sT_\sZ}| \le n_z \cdot p_a \right)
\wedge
\left( |\bi_{\sT_\sX} \oplus \bj_{\sT_\sX}| \le n_x \cdot p_a \right).
\end{equation}
Namely, the test passes if and only if
the error rate on the TEST-Z bits is at most $p_a$ {\em and}
the error rate on the TEST-X bits is at most $p_a$.

\begin{proposition}\label{modif_eff_bb84_probs}
For any $\epsilon > 0$,
\begin{eqnarray}
\pr\left[ \left( \frac{|\bC_\sI|}{n} \ge p_a + \epsilon \right)
\wedge \left( \bT = 1 \right) \right]
&\le& \pr\left[ \left( \frac{|\bC_{\sI_\sZ}|}{t_z} \ge p_a +
\epsilon \right) \wedge \left( \bT = 1 \right) \right] \nonumber \\
&+& \pr\left[ \left( \frac{|\bC_{\sI_\sX}|}{t_x} \ge p_a + \epsilon
\right) \wedge \left( \bT = 1 \right) \right].
\label{modif_eff_bb84_probs_real}
\end{eqnarray}

Equation~\eqref{modif_eff_bb84_probs_real} similarly applies to the
hypothetical ``inverted-INFO-basis'' protocol, too
(namely, it applies even if $\pr$ is replaced by $\pr_\invb$).
\end{proposition}
\begin{proof}
The same proof as Proposition~\ref{eff_bb84_probs}. 
\end{proof}

\begin{theorem}\label{thmsecurity_modif_eff_bb84}
For any $n, t_z, t_x, n_z, n_x > 0$,
$p_a > 0$, and $\epsilon_\secur, \epsilon_\rel > 0$
such that $t \triangleq n (p_a + \epsilon_\secur)$
and $t_\rel \triangleq n (p_a + \epsilon_\rel)$ are both integers,
$p_a + \epsilon_\secur \le \frac{1}{2}$,
$p_a + \epsilon_\rel \le \frac{1}{2}$,
and $n = t_z + t_x$,
it holds for the ``modified efficient BB84'' protocol that
\begin{eqnarray}
&&\frac{1}{2} \tr \left| \rho_{\mathrm{ABE}}
- \rho_\sU \otimes \rho_\sE \right| \nonumber \\
&\le& e^{-2 \left( \frac{n_z}{t_z + n_z} \right)^2
t_z \epsilon_\rel^2}
+ e^{-2 \left( \frac{n_x}{t_x + n_x} \right)^2
t_x \epsilon_\rel^2}
+ 2^{n[H_2(p_a + \epsilon_\rel) - r/n]} \nonumber \\
&+& 2m \sqrt{e^{-2 \left( \frac{n_x}{t_z + n_x} \right)^2
t_z \epsilon_\secur^2}
+ e^{-2 \left( \frac{n_z}{t_x + n_z} \right)^2
t_x \epsilon_\secur^2}
+ 2^{n[H_2(p_a + \epsilon_\secur) - (n-r-m)/n]}},
\end{eqnarray}
where $H_2(x) \triangleq -x \log_2(x) - (1-x) \log_2(1-x)$.
\end{theorem}
\begin{proof}
By using Corollary~\ref{corollary_compos} and
Proposition~\ref{modif_eff_bb84_probs}
(and also Corollary~\ref{corollary_ecc_decoding},
similarly to the proof of Theorem~\ref{thmsecurity2}), we get the following bound:
\begin{eqnarray}
&&\frac{1}{2} \tr \left| \rho_{\mathrm{ABE}}
- \rho_\sU \otimes \rho_\sE \right| \nonumber \\
&\le& \pr\left[ \left( \bk \ne \bk^\sB \right) \wedge
\left( \bT = 1 \right) \right] \nonumber \\
&+& 2m \sqrt{\pr_\invb\left[ \left( \frac{|\bC_\sI|}{n} \ge
\frac{t}{n} \right) \wedge \left( \bT = 1 \right) \right]
+ 2^{n[H_2(t/n) - (n-r-m)/n]}} \nonumber \\
&\le& \pr\left[ \left( \frac{|\bC_\sI|}{n} \ge
p_a + \epsilon_\rel \right) \wedge
\left( \bT = 1 \right) \right]
+ 2^{n[H_2(p_a + \epsilon_\rel) - r/n]} \nonumber \\
&+& 2m \sqrt{\pr_\invb\left[ \left( \frac{|\bC_\sI|}{n} \ge
p_a + \epsilon_\secur \right) \wedge \left( \bT = 1 \right) \right]
+ 2^{n[H_2(p_a + \epsilon_\secur) - (n-r-m)/n]}} \nonumber \\
&\le& \pr\left[ \left( \frac{|\bC_{\sI_\sZ}|}{t_z} \ge
p_a + \epsilon_\rel \right) \wedge
\left( \bT = 1 \right) \right] \nonumber \\
&+& \pr\left[ \left( \frac{|\bC_{\sI_\sX}|}{t_x} \ge
p_a + \epsilon_\rel \right) \wedge
\left( \bT = 1 \right) \right]
+ 2^{n[H_2(p_a + \epsilon_\rel) - r/n]} \nonumber \\
&+& 2m \sqrt{\pr_\invb\left[ \left( \frac{|\bC_{\sI_\sZ}|}{t_z} \ge
p_a + \epsilon_\secur \right) \wedge
\left( \bT = 1 \right) \right] +} \nonumber \\
&&\overline{\pr_\invb\left[ \left( \frac{|\bC_{\sI_\sX}|}{t_x} \ge
p_a + \epsilon_\secur \right) \wedge
\left( \bT = 1 \right) \right]
+ 2^{n[H_2(p_a + \epsilon_\secur) - (n-r-m)/n]}} \nonumber \\
&\le& \pr\left[ \left( \frac{|\bC_{\sI_\sZ}|}{t_z} \ge
p_a + \epsilon_\rel \right) \wedge
\left( \frac{|\bC_{\sT_\sZ}|}{n_z} \le p_a \right) \right] \nonumber \\
&+& \pr\left[ \left( \frac{|\bC_{\sI_\sX}|}{t_x} \ge
p_a + \epsilon_\rel \right) \wedge
\left( \frac{|\bC_{\sT_\sX}|}{n_x} \le p_a \right) \right]
+ 2^{n[H_2(p_a + \epsilon_\rel) - r/n]} \nonumber \\
&+& 2m \sqrt{\pr_\invb\left[ \left( \frac{|\bC_{\sI_\sZ}|}{t_z} \ge
p_a + \epsilon_\secur \right) \wedge
\left( \frac{|\bC_{\sT_\sX}|}{n_x} \le p_a \right) \right] +} \nonumber \\
&&\overline{\pr_\invb\left[ \left( \frac{|\bC_{\sI_\sX}|}{t_x} \ge
p_a + \epsilon_\secur \right) \wedge
\left( \frac{|\bC_{\sT_\sZ}|}{n_z} \le p_a \right) \right] +} \nonumber \\
&&\overline{2^{n[H_2(p_a + \epsilon_\secur)
- (n-r-m)/n]}}.\label{modif_eff_bb84_bound_0}
\end{eqnarray}
For bounding these probabilities,
we use Corollary~\ref{corollary_hoeffding}:

In the real protocol, the
INFO-Z and TEST-Z bits are sent and measured in the $z$ basis, while
the INFO-X and TEST-X bits are sent and measured in the $x$ basis.
Therefore, the random and uniform sampling of the $t_z$ INFO-Z bits
out of the $t_z + n_z$ bits sent in the $z$ basis
(assuming that the INFO-X and TEST-X bits have already been chosen)
does not affect the bases in the real protocol;
similarly, the random and uniform sampling of the $t_x$ INFO-X bits
out of the $t_x + n_x$ bits sent in the $x$ basis
(assuming that the INFO-Z and TEST-Z bits have already been chosen)
does not affect the bases in the real protocol.
This means that we can apply Corollary~\ref{corollary_hoeffding}
to both of these samplings, and we get
\begin{eqnarray}
\pr\left[ \left( \frac{|\bC_{\sI_\sZ}|}{t_z} \ge
p_a + \epsilon_\rel \right) \wedge
\left( \frac{|\bC_{\sT_\sZ}|}{n_z} \le p_a \right) \right]
&\le& e^{-2 \left( \frac{n_z}{t_z + n_z} \right)^2
t_z \epsilon_\rel^2},\label{modif_eff_bb84_bound_1}\\
\pr\left[ \left( \frac{|\bC_{\sI_\sX}|}{t_x} \ge
p_a + \epsilon_\rel \right) \wedge
\left( \frac{|\bC_{\sT_\sX}|}{n_x} \le p_a \right) \right]
&\le& e^{-2 \left( \frac{n_x}{t_x + n_x} \right)^2
t_x \epsilon_\rel^2}.\label{modif_eff_bb84_bound_2}
\end{eqnarray}

In the hypothetical ``inverted-INFO-basis'' protocol, the
INFO-X and TEST-Z bits are sent and measured in the $z$ basis, while
the INFO-Z and TEST-X bits are sent and measured in the $x$ basis.
Therefore, the random and uniform sampling of the $t_x$ INFO-X bits
out of the $t_x + n_z$ bits sent in the $z$ basis
(assuming that the INFO-Z and TEST-X bits have already been chosen)
does not affect the bases in the hypothetical protocol;
similarly, the random and uniform sampling of the $t_z$ INFO-Z bits
out of the $t_z + n_x$ bits sent in the $x$ basis
(assuming that the INFO-X and TEST-Z bits have already been chosen)
does not affect the bases in the hypothetical protocol.
This means that we can apply Corollary~\ref{corollary_hoeffding}
to both of these samplings, and we get
\begin{eqnarray}
&&\pr_\invb\left[ \left( \frac{|\bC_{\sI_\sZ}|}{t_z} \ge
p_a + \epsilon_\secur \right) \wedge
\left( \frac{|\bC_{\sT_\sX}|}{n_x} \le p_a \right) \right] \nonumber \\
&\le& e^{-2 \left( \frac{n_x}{t_z + n_x} \right)^2
t_z \epsilon_\secur^2},\label{modif_eff_bb84_bound_3}\\
&&\pr_\invb\left[ \left( \frac{|\bC_{\sI_\sX}|}{t_x} \ge
p_a + \epsilon_\secur \right) \wedge
\left( \frac{|\bC_{\sT_\sZ}|}{n_z} \le p_a \right) \right] \nonumber \\
&\le& e^{-2 \left( \frac{n_z}{t_x + n_z} \right)^2
t_x \epsilon_\secur^2}.\label{modif_eff_bb84_bound_4}
\end{eqnarray}

Substituting
Equations~\eqref{modif_eff_bb84_bound_1}--\eqref{modif_eff_bb84_bound_4}
into Equation~\eqref{modif_eff_bb84_bound_0}, we get the following bound:
\begin{eqnarray}
&&\frac{1}{2} \tr \left| \rho_{\mathrm{ABE}}
- \rho_\sU \otimes \rho_\sE \right| \nonumber \\
&\le& e^{-2 \left( \frac{n_z}{t_z + n_z} \right)^2
t_z \epsilon_\rel^2}
+ e^{-2 \left( \frac{n_x}{t_x + n_x} \right)^2
t_x \epsilon_\rel^2}
+ 2^{n[H_2(p_a + \epsilon_\rel) - r/n]} \nonumber \\
&+& 2m \sqrt{e^{-2 \left( \frac{n_x}{t_z + n_x} \right)^2
t_z \epsilon_\secur^2}
+ e^{-2 \left( \frac{n_z}{t_x + n_z} \right)^2
t_x \epsilon_\secur^2}
+ 2^{n[H_2(p_a + \epsilon_\secur) - (n-r-m)/n]}}.
\end{eqnarray}
\end{proof}

Similarly to the standard BB84 and ``efficient BB84'' protocols
(see Subsections~\ref{subsec_bb84} and~\ref{subsec_eff_bb84}),
we can obtain the following upper bound on the bit-rate:
\begin{equation}
R_{\mathrm{secret}} \triangleq \frac{m}{n}
< 1 - H_2(p_a + \epsilon_\secur)
- H_2(p_a + \epsilon_\rel).
\end{equation}

To get the asymptotic error rate threshold, we require
$R_{\mathrm{secret}} > 0$, and we get the asymptotic condition $2 H_2(p_a) < 1$.
This condition gives an asymptotic error rate threshold of $11\%$.

\section{Conclusion}
To sum up, we have found a new way for proving the composable
security of BB84 and of similar QKD protocols. This proof is relatively
simple and is mostly self-contained.

\section*{Acknowledgments}
The authors thank Eli Biham for useful discussions.
The work of T.M.\ and R.L.\ was partly supported
by the Israeli MOD Research and Technology Unit.
The work of R.L.\ was also partly supported by the Canada Research Chair Program,
the Technion's Helen Diller Quantum Center (Haifa, Israel),
the Government of Spain (FIS2020--TRANQI and Severo Ochoa CEX2019--000910--S),
Fundaci\'o Cellex, Fundaci\'o Mir--Puig, Generalitat de Catalunya (CERCA program),
and the EU NextGen Funds.

\bibliographystyle{unsrturl}
\bibliography{composability}

\end{document}